\documentclass[10pt,  draftclsnofoot, onecolumn, article, romanappendices]{IEEEtran}
\usepackage{setspace}

\usepackage{amsmath,amssymb,amsthm,mathrsfs,amsfonts,dsfont,stmaryrd,wasysym,marvosym}
\usepackage{mathtools}
\usepackage{color}
\usepackage{graphicx, subfig} 
\usepackage{psfrag, xcolor}
\usepackage{enumitem}
\usepackage{cite,epsfig,dblfloatfix,url}

\usepackage{algorithm} 
\usepackage{algpseudocode}

\algnewcommand{\Inputs}[1]{%
	\State \textbf{input:}
	\parbox[t]{.8\linewidth}{\raggedright #1}
}
\algnewcommand{\Initialize}[1]{%
	\State \textbf{initialization}
	\parbox[t]{.95\linewidth}{\raggedright #1}
}
\algnewcommand{\Outputs}[1]{%
	\State \textbf{output:}
	\parbox[t]{.8\linewidth}{\raggedright #1}
}

\allowdisplaybreaks[4]
\newtheorem{theorem}{Theorem}
\newtheorem{lemma}{Lemma}  
\newtheorem{corollary}{Corollary}  
\newtheorem{example}{Example}
\newtheorem{proposition}{Proposition}   
\newtheorem{definition}{Definition} 
\newtheorem{remark}{Remark} 
\newcommand{\dv}{\mathbf} 
\newcommand{\mc}{\mathcal} 
\newcommand{\mkv}{-\!\!\!\!\minuso\!\!\!\!-}

\newcommand{\kbar}{{\bar{k}}}
\newcommand{\Mod}[1]{\ (\mathrm{mod}\ #1)}
\newcommand{\E}{\mathbb{E}}
\DeclareMathOperator*{\argmin}{min}

\newcommand{\squeezeup}{\vspace{-1em}}

\newcommand{\bqed}{\tag*{$\blacksquare$}}
\newcommand*{\qedblack}{\hfill\ensuremath{\blacksquare}}

\newcommand*\xbar[1]{%
    \hbox{%
		 \vbox{%
		 \hrule height 0.5pt 
		 \kern0.5ex
		 \hbox{%
		 \kern-0.1em
		\ensuremath{#1}%
		\kern-0.1em
		}%
		}%
		}%
		}

\begin{document}
\fontencoding{OT1}\fontsize{10}{11}\selectfont

\title{Vector Gaussian CEO Problem Under Logarithmic Loss and Applications}

\author{\vspace{0.3cm}
	Yi{\u{g}}it U{\u{g}}ur $^{\dagger}$$^{\ddagger}$ \qquad \quad I\~naki Estella Aguerri $^{\dagger}$ \qquad \quad Abdellatif Zaidi $^{\dagger}$$^{\ddagger}$ \vspace{0.3cm} \\   
	{\small
		$^{\dagger}$ Mathematical and Algorithmic Sciences Lab, Paris Research Center, Huawei Technologies,\\ Boulogne-Billancourt, 92100, France\\
		$^{\ddagger}$ Universit\'e Paris-Est, Champs-sur-Marne, 77454, France\\
		\vspace{0.1cm}
		\{\tt yigit.ugur@gmail.com,  inaki.estella@gmail.com, abdellatif.zaidi@u-pem.fr\} }
	\thanks{The results of this paper have been presented in part at the 2017 IEEE Information Theory Workshop~\cite{UEZ17} and in part at the 2018 IEEE Information Theory Workshop~\cite{UEZ18}.}
	}

\markboth{A\MakeLowercase{ccepted for publication in} IEEE T\MakeLowercase{ransactions of} I\MakeLowercase{nformation} T\MakeLowercase{heory}, 2020}%
{UGUR \MakeLowercase{\textit{et al.}}: VECTOR GAUSSIAN CEO PROBLEM UNDER LOGARITHMIC LOSS AND APPLICATIONS}


\maketitle

\begin{abstract}
In this paper, we study the vector Gaussian Chief Executive Officer (CEO) problem under logarithmic loss distortion measure. Specifically, $K \geq 2$ agents observe independently corrupted Gaussian noisy versions of a remote vector Gaussian source, and communicate independently with a decoder or CEO over rate-constrained noise-free links. The CEO also has its own Gaussian noisy observation of the source and  wants to reconstruct the remote source to within some prescribed distortion level where the incurred distortion is measured under the logarithmic loss penalty criterion. We find an explicit characterization of the rate-distortion region of this model. The result can be seen as the counterpart to the vector Gaussian setting of that by Courtade-Weissman which provides the rate-distortion region of the model in the discrete memoryless setting. For the proof of this result, we obtain an outer bound by means of a technique that relies on the de Bruijn identity and the properties of Fisher information. The approach is similar to Ekrem-Ulukus outer bounding technique for the vector Gaussian CEO problem under quadratic distortion measure, for which it was there found generally non-tight; but it is shown here to yield a complete characterization of the region for the case of logarithmic loss measure. Also, we show that Gaussian test channels with time-sharing exhaust the Berger-Tung inner bound, which is optimal. Furthermore, application of our results allows us to find the complete solutions of two related problems: a quadratic vector Gaussian CEO problem with \textit{determinant} constraint and the vector Gaussian distributed Information Bottleneck problem. Finally, we develop Blahut-Arimoto type algorithms that allow to compute numerically the regions provided in this paper, for both discrete and Gaussian models. With the known relevance of the logarithmic loss fidelity measure in the context of learning and prediction, the proposed algorithms may find usefulness in a variety of applications where learning is performed distributively. We illustrate the efficiency of our algorithms through some numerical examples.
\end{abstract}

\begin{figure}[t!]
	\centering
	\includegraphics[width=0.7\linewidth]{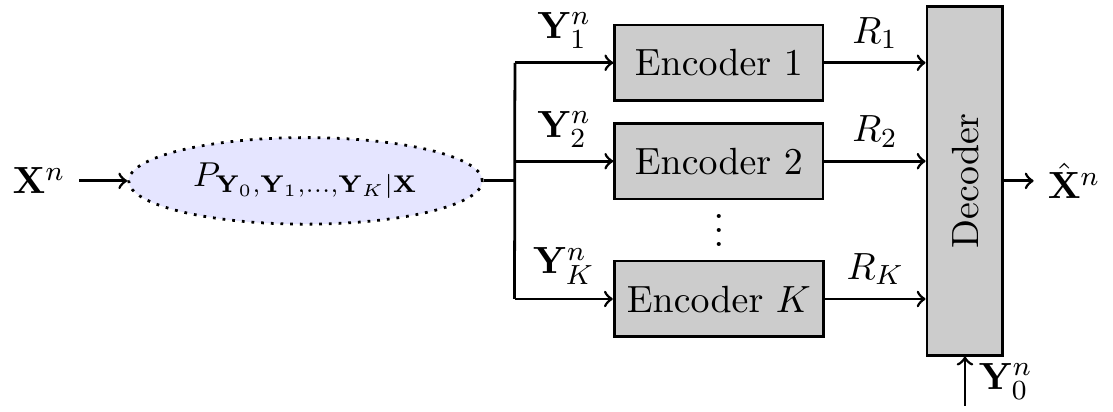}
	\caption{Chief Executive Officer (CEO) source coding problem with side information.}	
	\label{fig-system-model-CEO}
\end{figure}

\section{Introduction}\label{section-intorduction}

Consider the vector Gaussian Chief Executive Officer (CEO) problem shown in Figure~\ref{fig-system-model-CEO}. In this model, there are $K \geq 2$ agents each observing a noisy version of a vector Gaussian source $\dv X$. The goal of the agents is to describe the source to a central unit, which wants to reconstruct this source to within a prescribed distortion level. The incurred distortion is measured according to some loss measure $d\: :\: \mc X \times \hat{\mc X} \rightarrow \mathbb{R}$, where $\hat{\mc X}$ designates the reconstruction alphabet. For quadratic distortion measure, i.e., 
\begin{equation*}
d(x,\hat{x})=|x-\hat{x}|^2 \;,
\end{equation*}
the rate-distortion region of the vector Gaussian CEO problem is still unknown in general, except in few special cases the most important of which is perhaps the case of scalar sources, i.e., scalar Gaussian CEO problem. For this case, a complete solution, in terms of characterization of the optimal rate-distortion region, was found independently by Oohama in~\cite{O05} and by Prabhakaran \textit{et al.} in~\cite{PTR04}. Key to establishing this result is a judicious application of the entropy power inequality. The extension of this argument to the case of vector Gaussian sources, however, is not straightforward as the entropy power inequality is known to be non-tight in this setting. The reader may refer also to~\cite{CW11, WC12} where non-tight outer bounds on the rate-distortion region of the vector Gaussian CEO problem under quadratic distortion measure are obtained by establishing some extremal inequalities that are similar to Liu-Viswanath~\cite{LV07}, and to~\cite{XW16} where a strengthened extremal inequality yields a complete characterization of the region of the vector Gaussian CEO problem in the special case of trace distortion constraint.

In this paper, we study the CEO problem of Figure~\ref{fig-system-model-CEO} in the case in which $(\dv X, \dv Y_0, \dv Y_1,\ldots,\dv Y_K)$ is jointly Gaussian and the distortion is measured using the logarithmic loss criterion, i.e., 
\begin{equation}~\label{equation-log-loss-distortion-measure-n-letter}
d^{(n)}(x^n, \hat{x}^n) = \frac{1}{n} \sum_{i=1}^{n} d(x_i,\hat{x}_i) \;,
\end{equation}
with the letter-wise distortion given by  
\begin{equation}~\label{equation-log-loss-distortion-measure}
d(x,\hat{x}) = \log\frac{1}{\hat{x}(x)} \;,
\end{equation}
where $\hat{x}(\cdot)$ designates a probability distribution on $\mc X$ and $\hat{x}(x)$ is the value of this distribution evaluated for the outcome $x \in \mc X$.

The logarithmic loss distortion measure, often referred to as \textit{self-information loss} in the literature about prediction, plays a central role in settings in which reconstructions are allowed to be `soft', rather than `hard' or deterministic. That is, rather than just assigning a deterministic value to each sample of the source, the decoder also gives an assessment of the degree of confidence or reliability on each estimate, in the form of weights or probabilities. This measure, which was introduced in the context of rate-distortion theory by Courtade \textit{et al.}~\cite{CW11-2, CW14}, has appreciable mathematical properties~\cite{JCVW15,NW15}, such as a deep connection to lossless coding for which fundamental limits are well developed (e.g., see~\cite{SRV17} for recent results on universal lossy compression under logarithmic loss that are built on this connection). Also, it is widely used as a penalty criterion in various contexts, including clustering and classification~\cite{TPB99}, pattern recognition, learning and prediction~\cite{C-BL06}, image processing~\cite{AABG06}, secrecy~\cite{KCOSW16} and others.

\subsection{Main Contributions}

The main contribution of this paper is a complete characterization of the rate-distortion region of the vector Gaussian CEO problem of Figure~\ref{fig-system-model-CEO} under logarithmic loss distortion measure. The result can be seen as the counterpart to the vector Gaussian case of that by Courtade and Weissman~\cite[Theorem 10]{CW14}, who established the rate-distortion region of the CEO problem under logarithmic loss in the discrete memoryless (DM) case. For the proof of this result, we derive a matching outer bound by means of a technique that relies of the de Bruijn identity, a connection between differential entropy and Fisher information, along with the properties of minimum mean square error (MMSE) and Fisher information. By opposition to the case of quadratic distortion measure, for which the application of this technique was shown in~\cite{EU14} to result in an outer bound that is generally non-tight, we show that this approach is successful in the case of logarithmic distortion measure and yields a complete characterization of the region. On this aspect, it is noteworthy that, in the specific case of scalar Gaussian sources, an alternate converse proof may be obtained by extending that of the scalar Gaussian many-help-one source coding problem by Oahama~\cite{O05} and Prabhakaran \textit{et al.}~\cite{PTR04} by accounting for side information and replacing the original mean square error distortion constraint with conditional entropy. However, such approach does not seem to lead to a conclusive result in the vector case as the entropy power inequality is known to be generally non-tight in this setting~\cite{TP05,WSS06}. The proof of the achievability part simply follows by evaluating a straightforward extension to the continuous alphabet case of the solution of the DM model using Gaussian test channels and \textit{no} time-sharing. Because this does \textit{not} necessarily imply that Gaussian test channels also exhaust the Berger-Tung inner bound, we investigate the question and we show that they \textit{do} if time-sharing is allowed. 

Furthermore, we show that application of our results allows us to find complete solutions to two related problems. The first is a quadratic vector Gaussian CEO problem with reconstruction constraint on the \textit{determinant} of the error covariance matrix that we introduce here, and for which we also characterize the optimal rate-distortion region. Key to establishing this result, we show that the rate-distortion region of vector Gaussian CEO problem under logarithmic loss which is found in this paper translates into an outer bound on the rate region of the quadratic vector Gaussian CEO problem with \textit{determinant} constraint. The reader may refer to, e.g.,~\cite{PCL03} and~\cite{SSBGS02} for examples of usage of such a determinant constraint in the context of equalization and others. The second is an extension of Tishby's single-encoder Information Bottleneck (IB) method~\cite{TPB99} to the case of multiple encoders. Information theoretically, this problem is known to be essentially a remote source coding problem with logarithmic loss distortion measure~\cite{HT07}; and, so, we use our result for the vector Gaussian CEO problem under logarithmic loss to infer a full characterization of the optimal trade-off between \textit{complexity} (or rate) and \textit{accuracy} (or information) for the distributed vector Gaussian IB problem.

Finally, for both DM and memoryless Gaussian settings we develop Blahut-Arimoto (BA)~\cite{B72,A72} type iterative algorithms that allow to compute (approximations of) the rate regions that are established in this paper; and prove their convergence to stationary points. We do so through a variational formulation that allows to determine the set of self-consistent equations that are satisfied by the stationary solutions. In the Gaussian case, we show that the algorithm reduces to an appropriate updating rule of the parameters of noisy linear projections. We note that the computation of the rate-distortion regions of multiterminal and CEO source coding problems is important \textit{per-se} as it involves non-trivial optimization problems over distributions of auxiliary random variables. Also, since the logarithmic loss function is instrumental in connecting problems of multiterminal rate-distortion theory with those of distributed learning and estimation, the algorithms that are developed in this paper also find usefulness in emerging applications in those areas. For example, our algorithm for the DM CEO problem under logarithm loss measure can be seen as a generalization of Tishby's IB method~\cite{TPB99} to the distributed learning setting. Similarly, our algorithm for the vector Gaussian CEO problem under logarithm loss measure can be seen as a generalization of that of~\cite{CGTW05, WM14, WFM14} to the distributed learning setting. For other extensions of the BA algorithm in the context of multiterminal data transmission and compression, the reader may refer to related works on point-to-point~\cite{CSX05,CB04} and broadcast and multiple access multiterminal settings~\cite{DYW04,RG04}. 

\subsection{Related Works}
 
As we already mentioned, this paper mostly relates to~\cite{CW14} in which the authors establish the rate-distortion region of the DM CEO problem under logarithmic loss in the case of an arbitrary number of encoders and no side information at the decoder, as well as that of the DM multiterminal source coding problem under logarithmic loss in the case of two encoders and no side information at the decoder. Motivated by the increasing interest for problems of learning and prediction, a growing body of works study point-to-point and multiterminal source coding models under logarithmic loss. In~\cite{JCVW15}, Jiao \textit{et al.} provide a fundamental justification for inference using logarithmic loss, by showing that under some mild conditions (the loss function satisfying some data processing property and alphabet size larger than two) the reduction in optimal risk in the presence of side information is uniquely characterized by mutual information, and the corresponding loss function coincides with the logarithmic loss. Somewhat related, in~\cite{PW18} Painsky and Wornell show that for binary classification problems the logarithmic loss dominates ``universally" any other convenient (i.e., smooth, proper and convex) loss function, in the sense that by minimizing the logarithmic loss one minimizes the regret that is associated with any such measures. More specifically, the divergence associated any smooth, proper and convex loss function is shown to be bounded from above by the Kullback-Leibler divergence, up to a multiplicative normalization constant. In~\cite{SRV17}, the authors study the problem of universal lossy compression under logarithmic loss, and derive bounds on the non-asymptotic fundamental limit of fixed-length universal coding with respect to a family of distributions that generalize the well-known minimax bounds for universal lossless source coding. In~\cite{LWOG18}, the minimax approach is studied for a problem of remote prediction and is shown to correspond to a one-shot minimax noisy source coding problem. The setting of remote prediction of~\cite{LWOG18} provides an approximate one-shot operational interpretation of the Information Bottleneck method of~\cite{TPB99},   which is also sometimes interpreted as a remote source coding problem under logarithmic loss~\cite{HT07}.  

Logarithmic loss is also instrumental in problems of data compression under a mutual information constraint~\cite{TC09}, and problems of relaying with relay nodes that are constrained not to know the users' codebooks (sometimes termed ``oblivious" or nomadic processing) which is studied in the single user case first by Sanderovich \textit{et al.} in~\cite{SSSK08} and then by Simeone \textit{et al.} in~\cite{SES11}, and in the multiple user multiple relay case by Aguerri \textit{et al.} in~\cite{EZSC17-2} and~\cite{EZSC17}. Other applications in which the logarithmic loss function can be used include secrecy and privacy~\cite{KCOSW16, CMMVCD17}, hypothesis testing against independence~\cite{AC86,H87,RW12,TC08, SWT18} and others.

\subsection{Outline and Notation}

The rest of this paper is organized as follows. Section~\ref{section-problem-formulation} provides a formal description of the vector Gaussian CEO model that we study in this paper, as well as some definitions that are related to it. Section~\ref{section-vector-CEO} contains the main results of this paper: an explicit characterization of the rate-distortion region of the memoryless vector Gaussian CEO problem with side information under logarithmic loss as well as the proof that Gaussian test channels with time-sharing exhaust the Berger-Tung rate region which is optimal. In Section~\ref{secV} we use our results on the CEO problem under logarithmic loss to infer complete solutions of two related problems: a quadratic vector Gaussian CEO problem with a determinant constraint on the covariance matrix error and the vector Gaussian distributed Information Bottleneck problem. Section~\ref{section-BA-algorithm} provides BA-type algorithms for the computation of the rate-distortion regions that are established in this paper in both DM and Gaussian cases as well as proofs of their convergence and some numerical examples. The proofs are deferred to the appendices section.  

\vspace{1em}
Throughout this paper, we use the following notation. Upper case letters are used to denote random variables, e.g., $X$; lower case letters are used to denote realizations of random variables, e.g., $x$; and calligraphic letters denote sets, e.g., $\mc X$.  The cardinality of a set $\mc X$ is denoted by $|\mc X|$. The length-$n$ sequence $(X_1,\ldots,X_n)$ is denoted as  $X^n$; and, for integers $j$ and $k$ such that $1 \leq k \leq j \leq n$, the sub-sequence $(X_k,X_{k+1},\ldots, X_j)$ is denoted as  $X_{k}^j$. Probability mass functions (pmfs) are denoted by $P_X(x)=\mathrm{Pr}\{X=x\}$; and, sometimes, for short, as $p(x)$. We use $\mc P(\mc X)$ to denote the set of discrete probability distributions on $\mc X$.  Boldface upper case letters denote vectors or matrices, e.g., $\dv X$, where context should make the distinction clear. For an integer $K \geq 1$, we denote the set of integers smaller or equal $K$ as $\mc K = \{ k \in \mathbb{N} \: : \: 1 \leq k \leq K\}$. For a set of integers $\mc S \subseteq \mc K$, the complementary set of $\mc S$ is denoted by $\mc S^c$, i.e., $\mc S^c = \{k \in \mathbb{N} \: : \: k \in \mc K \setminus \mc S\}$. Sometimes, for convenience we will need to define $\bar{\mc S}$ as $\bar{\mc S} = \{0\} \cup \mc S^c$. For a set of integers $\mc S \subseteq \mc K$; the notation $X_{\mc S}$ designates the set of random variables $\{X_k\}$ with indices in the set $\mc S$, i.e., $X_{\mc S}=\{X_k\}_{k \in \mc S}$. We denote the covariance of a zero mean, complex-valued, vector $\dv X$ by $\mathbf{\Sigma}_{\mathbf{x}} =\mathbb{E}[\mathbf{XX}^{\dag}]$, where $(\cdot)^{\dag}$ indicates conjugate transpose. Similarly, we denote the cross-correlation of two zero-mean vectors $\dv X$  and $\dv Y$ as $\mathbf{\Sigma}_{\mathbf{x},\mathbf{y}} = \mathbb{E}[\mathbf{XY}^{\dag}]$, and the conditional correlation matrix of $\mathbf{X}$ given $\mathbf{Y}$ as $\mathbf{\Sigma}_{\mathbf{x}|\mathbf{y}} = \mathbb{E}\big[\big(\dv X - \mathbb{E}[\dv X|\dv Y]\big)\big(\dv X - \mathbb{E}[\dv X|\dv Y]\big)^{\dag}\big]$ i.e., $\mathbf{\Sigma}_{\mathbf{x}|\mathbf{y}} = \mathbf{\Sigma}_{\mathbf{x}}-\mathbf{\Sigma}_{\mathbf{x},\mathbf{y}}\mathbf{\Sigma}_{\mathbf{y}}^{-1}\mathbf{\Sigma}_{\mathbf{y},\mathbf{x}}$. For matrices $\dv A$ and $\dv B$, the notation $\mathrm{diag}(\dv A, \dv B)$ denotes the block diagonal matrix whose diagonal elements are the matrices $\dv A$ and $\dv B$ and its off-diagonal elements are the all zero matrices. Also, for a set of integers $\mc J \subset \mathbb{N}$ and a family of matrices $\{\dv A_i\}_{i \in \mc J}$ of the same size, the notation $\dv A_{\mc J}$ is used to denote the (super) matrix obtained by concatenating vertically the matrices $\{\dv A_i\}_{i \in \mc J}$, where the indices are sorted in the ascending order, e.g, $\dv A_{\{0,2\}}=[\dv A^{\dag}_0, \dv A^{\dag}_2]^{\dag}$.

\section{Problem Formulation}\label{section-problem-formulation}  

Consider the $K$-encoder CEO problem  with side information shown in Figure~\ref{fig-system-model-CEO}, where the agents’ observations are assumed to be Gaussian noisy versions of a remote vector Gaussian source. Specifically, let $(\dv X, \dv Y_0, \dv Y_1, \ldots, \dv Y_K)$  be a jointly Gaussian random vector, with zero mean and covariance matrix ${\dv \Sigma}_{(\dv x, \dv y_0, \dv y_1, \ldots, \dv y_K)}$. Without loss of generality, the remote vector source $ \dv X \in \mathbb{C}^{n_x}$ is assumed complex-valued, has $n_x$-dimensions, and is assumed to be Gaussian with zero mean and covariance matrix $\dv\Sigma_{\dv x} \succeq \dv 0$. $\dv X^n=(\dv X_1, \ldots, \dv X_n)$ denotes a collection of $n$ independent copies of $\dv X$. The agents' observations are Gaussian noisy versions of the remote vector source, with the observation at agent $k \in \mc K$ given by
\begin{equation}~\label{mimo-gaussian-model}
\dv Y_{k,i} = \dv H_k \dv X_i + \dv N_{k,i} \;, \quad\text{for}\:\: i = 1,\ldots,n \;,
\end{equation}
where $\dv H_k \in \mathds{C}^{n_k\times n_x}$ represents the channel matrix connecting the remote vector source to the $k$-th agent; and $\dv N_{k,i} \in \mathds{C}^{n_k}$ is the noise vector at this agent, assumed to be i.i.d. Gaussian with zero-mean and independent from $\dv X_i$. The decoder has its own noisy observation of the remote vector source, in the form of a correlated jointly Gaussian side information stream $\dv Y^n_0$, with
\begin{equation}~\label{mimo-gaussian-model-2}
\dv Y_{0,i} = \dv H_0 \dv X_i  + \dv N_{0,i} \;, \quad\text{for}\:\: i = 1,\ldots,n \;,
\end{equation}
where, similar to the above, $\dv H_0 \in \mathds{C}^{n_0\times n_x}$ is the channel matrix connecting the remote vector source to the CEO; and $\dv N_{0,i} \in \mathds{C}^{n_0}$ is the noise vector at the CEO, assumed to be Gaussian with zero-mean and covariance matrix $\dv\Sigma_0 \succeq \dv 0$ and independent from $\dv X_i$. In this section, it is assumed that the agents' observations are independent conditionally given the remote vector source $\dv X^n$ and the side information $\dv Y^n_0$, i.e., for all $\mc S \subseteq \mc K$,
\begin{equation}~\label{markov-chain-assumption-gaussian-model} 
\dv Y^n_{\mc S} \mkv  (\dv X^n, \dv Y^n_0)  \mkv \dv Y^n_{\mc S^c} \;. \end{equation}
Using~\eqref{mimo-gaussian-model} and~\eqref{mimo-gaussian-model-2}, it is easy to see that the assumption~\eqref{markov-chain-assumption-gaussian-model} is equivalent to that the noises at the agents are independent conditionally given $\dv N_0$. For notational simplicity, $\dv\Sigma_k$ denotes the conditional covariance matrix of the noise $\dv N_k$ at the $k$-th agent given $\dv N_0$, i.e., $\dv\Sigma_k := \dv\Sigma_{\dv n_k|\dv n_0}$. Recalling that for a set $\mc S \subseteq \mc K$, $\dv N_{\mc S}$ designates the collection of noise vectors with indices in the set $\mc S$, in what follows we denote the covariance matrix of $\dv N_{\mc S}$ as $\dv\Sigma_{\dv n_{\mc S}}$.

In this model, Encoder (or agent) $k\in \mc K$ uses $R_k$ bits per sample to describe its observation $\dv Y^n_k$ to the decoder. The decoder wants to reconstruct the remote source $\dv X^n$ to within a prescribed fidelity level. Similar to~\cite{CW14},  we consider the reproduction alphabet to be equal to the set of probability distributions over the source alphabet $\mathbb{C}^{n\times n_x}$. In other words, the decoder generates `soft` estimates of the remote source's sequences. We consider the logarithmic loss distortion measure defined as in~\eqref{equation-log-loss-distortion-measure-n-letter}, where the letter-wise distortion measure is given by~\eqref{equation-log-loss-distortion-measure}.  

\begin{definition}~\label{definition-encoder}
A rate-distortion code (of blocklength $n$) for the model of Figure~\ref{fig-system-model-CEO} consists of $K$ encoding functions \begin{equation*}
\phi^{(n)}_k \: : \:  \mathbb{C}^{n \times n_k} \rightarrow \{1,\ldots,M^{(n)}_k\} \;, \quad\text{for}\:\: k=1,\ldots,K \;,
\end{equation*}
and a decoding function
\begin{equation*}
\psi^{(n)} \: : \:  \{1,\ldots,M^{(n)}_1\}\times \cdots \times \{1,\ldots,M^{(n)}_K\}\times \mathbb{C}^{n \times n_0} \rightarrow \mc P(\mathbb{C}^{n \times n_x}) \;,
\end{equation*}
where $\mc P(\mathbb{C}^{n \times n_x})$ designates the set of probability distributions over the $n$-Cartesian product of $\mathbb{C}^{n_x}$. \qedblack
\end{definition}

\begin{definition}~\label{definition-rate-distortion-region}
A rate-distortion tuple $(R_1,\ldots,R_K,D)$ is achievable for the vector Gaussian CEO problem with side information if there exist a blocklength  $n$, $K$ encoding functions $\{\phi^{(n)}_k\}_{k=1}^K$ and a decoding function $\psi^{(n)}$ such that
\begin{align*}
R_k &\geq \frac{1}{n}\log M^{(n)}_k \;, \quad\text{for}\:\: k=1,\ldots,K \;,\\
D &\geq \E\big[ d^{(n)}\big( \dv X^n, \psi^{(n)}(\phi^{(n)}_1(\dv Y^n_1), \ldots,\phi^{(n)}_K(\dv Y^n_K), \dv Y_0^n) \big) \big] \;.
\end{align*} 
	
\noindent The rate-distortion region $\mc{RD}_{\mathrm{VG}\text{-}\mathrm{CEO}}^\star$ of the vector Gaussian CEO problem under logarithmic loss is defined as the union of all non-negative tuples $(R_1,\ldots,R_K,D)$ that are achievable. \qedblack
\end{definition}

The main goal of this paper is to characterize the rate-distortion region $\mc{RD}_{\mathrm{VG}\text{-}\mathrm{CEO}}^\star$ of the vector Gaussian CEO problem under logarithmic loss.

\section{Main Results}\label{section-vector-CEO}

In this section we provide an explicit characterization of the rate-distortion region $\mc{RD}_{\mathrm{VG}\text{-}\mathrm{CEO}}^\star$ of the vector Gaussian CEO problem under logarithmic loss. Also, we show that Gaussian test channels with time-sharing exhaust the Berger-Tung region which is optimal.

\subsection{Rate-Distortion Region}

We first state the following theorem which follows essentially by an easy application of~\cite[Theorem 10]{CW14} that provides the rate-distortion region of the DM version of the problem.

\begin{definition}~\label{defintion-continous-RD1-CEO}
For given tuple of auxiliary random variables $(U_1,\ldots,U_K,Q)$ with distribution $P_{U_{\mc K},Q}(u_{\mc K},q)$ such that $P_{\dv X, \dv Y_0,\dv Y_{\mc K}, U_{\mc K},Q}(\dv x,\dv y_0,\dv y_{\mc K},u_{\mc K},q)$ factorizes as
\begin{equation}~\label{equation-joint-measure-auxiliary-random-variables-continous-1-CEO-problem}
P_{\dv X,\dv Y_0}(\dv x,\dv y_0) \prod_{k=1}^K P_{\dv Y_k|\dv X,\dv Y_0}(\dv y_k|\dv x,\dv y_0) \: P_Q(q) \prod_{k=1}^{K} P_{U_k|\dv Y_k,Q}(u_k|\dv y_k,q) \;,
\end{equation}
define $\mc{RD}_\mathrm{CEO}^\mathrm{I}(U_1,\ldots,U_K,Q)$ as the set of all non-negative rate-distortion tuples $(R_1,\ldots,R_K,D)$ that satisfy, for all subsets $\mc S \subseteq \mc K$, 
\begin{equation}~\label{equation-continous-RD1-CEO}
D + \sum_{k \in \mc S} R_k \geq \sum_{k \in \mc S} I(\dv Y_k;U_k|\dv X,\dv Y_0,Q) + h(\dv X|U_{\mc S^c},\dv Y_0,Q) \;.
\end{equation}	
Also, let $\mc{RD}_\mathrm{CEO}^\mathrm{I} := \bigcup \mc{RD}_\mathrm{CEO}^\mathrm{I} (U_1,\ldots,U_K,Q)$ where the union is taken over all tuples $(U_1,\ldots,U_K,Q)$ with distributions that satisfy~\eqref{equation-joint-measure-auxiliary-random-variables-continous-1-CEO-problem}.
\qedblack
\end{definition}	

\begin{definition}~\label{defintion-continous-RD2-CEO}
For given tuple of auxiliary random variables $(V_1,\ldots,V_K,Q^\prime)$ with distribution $P_{V_{\mc K},Q^\prime}(v_{\mc K},q^\prime)$ such that $P_{\dv X, \dv Y_0,\dv Y_{\mc K}, V_{\mc K},Q^\prime}(\dv x,\dv y_0,\dv y_{\mc K},v_{\mc K},q^\prime)$ factorizes as
\begin{equation}~\label{equation-joint-measure-auxiliary-random-variables-continous-2-CEO-problem}
P_{\dv X,\dv Y_0}(\dv x,\dv y_0) \prod_{k=1}^K P_{\dv Y_k|\dv X,\dv Y_0}(\dv y_k|\dv x,\dv y_0) \: P_{Q^\prime}(q^\prime) \prod_{k=1}^{K} P_{V_k|\dv Y_k,Q^\prime}(v_k|\dv y_k,q^\prime) \;,
\end{equation}
define $\mc{RD}_\mathrm{CEO}^\mathrm{II}(V_1,\ldots,V_K,Q^\prime)$ as the set of all non-negative rate-distortion tuples $(R_1,\ldots,R_K,D)$ that satisfy, for all subsets $\mc S \subseteq \mc K$, 
\begin{align} 
\sum_{k \in \mc S} R_k &\geq I(\dv Y_{\mc S};V_{\mc S}|V_{\mc S^c},\dv Y_0,Q^\prime) \label{equation-rate-constraints}\\
D &\geq h(\dv X|V_1,\ldots,V_K,\dv Y_0,Q^\prime) \;.
\end{align}
Also, let $\mc{RD}_\mathrm{CEO}^\mathrm{II} := \bigcup \mc{RD}_\mathrm{CEO}^\mathrm{II} (V_1,\ldots,V_K,Q')$ where the union is taken over all tuples $(V_1,\ldots,V_K,Q')$ with distributions that satisfy~\eqref{equation-joint-measure-auxiliary-random-variables-continous-2-CEO-problem}. \qedblack
\end{definition}

\begin{theorem}~\label{theorem-continuous-RD1-CEO}
The rate-distortion region for the vector Gaussian CEO problem under logarithmic loss is given by  
\begin{equation*}
\mc{RD}_{\mathrm{VG}\text{-}\mathrm{CEO}}^\star = \mc{RD}_\mathrm{CEO}^\mathrm{I} = \mc{RD}_\mathrm{CEO}^\mathrm{II} \;.
\end{equation*}
\end{theorem}

\begin{proof}
The proof of Theorem~\ref{theorem-continuous-RD1-CEO} is given in Appendix~\ref{proof-continous-RD1-CEO}.
\end{proof}

For convenience, we now introduce the following notation which will be instrumental in what follows. Let, for every set $\mc S \subseteq \mc K$, the set $\bar{\mc S} = \{0\} \cup \mc S^c$. Also, for $\mc S \subseteq \mc K$ and given matrices $\{\dv\Omega_k\}_{k=1}^K$ such that $\dv 0 \preceq \dv\Omega_k \preceq \dv\Sigma_k^{-1}$, let $\boldsymbol{\Lambda}_{\bar{\mc S}}$ designate the block-diagonal matrix given by
\begin{align}~\label{equation-definition-T}
\boldsymbol{\Lambda}_{\bar{\mc S}} :=
\begin{bmatrix}
\dv 0 & \dv 0 \\
\dv 0 & \mathrm{diag}(\{ \dv\Sigma_k - \dv\Sigma_k \dv\Omega_k \dv\Sigma_k \}_{k\in\mc S^c})
\end{bmatrix} \;,
\end{align}
where $\dv 0$ in the principal diagonal elements is the $n_0{\times}n_0$-all zero matrix. Besides, the notation $\dv H_{\mc S}$ is used to denote the (super) matrix obtained by concatenating vertically the matrices $\{\dv H_i\}_{i \in \mc S}$, where the indices are sorted in the ascending order, e.g, $\dv H_{\{0,2\}}=[\dv H^{\dag}_0, \dv H^{\dag}_2]^{\dag}$.

\noindent The following theorem is the main contribution of this paper, which is an explicit characterization of the rate-distortion region of the vector Gaussian CEO problem with side information under logarithmic loss measure.

\begin{theorem}~\label{theorem-Gauss-RD-CEO}
The rate-distortion region $\mc{RD}_{\mathrm{VG}\text{-}\mathrm{CEO}}^\star$ of the vector Gaussian CEO problem under logarithmic loss is given by the set of all non-negative rate-distortion tuples $(R_1,\ldots, R_K,D)$ that satisfy, for all subsets $\mc S \subseteq \mc K$,  
\begin{align*}
D + \sum_{k \in \mc S} R_k \geq  \sum_{k\in \mc S} \log \frac{1}{\left| \dv I - \dv\Omega_k \dv\Sigma_k \right|} +
\log\left| (\pi e) \left( \dv\Sigma_{\dv x}^{-1} + \dv H_{\bar{\mc S}}^\dagger \dv\Sigma_{\dv n_{\bar{\mc S}}}^{-1} \big( \dv I - \boldsymbol{\Lambda}_{\bar{\mc S}} \dv\Sigma_{\dv n_{\bar{\mc S}}}^{-1}  \big) \dv H_{\bar{\mc S}} \right)^{-1} \right| \;,
\end{align*}
for matrices $\{\dv\Omega_k\}_{k=1}^K$ such that $\dv 0 \preceq \dv\Omega_k \preceq \dv\Sigma_k^{-1}$, where $\bar{\mc S}=\{0\} \cup \mc S^c$ and $\boldsymbol{\Lambda}_{\bar{\mc S}}$ is as defined by~\eqref{equation-definition-T}.  
\end{theorem}

\begin{proof}
The proof of the direct part of Theorem~\ref{theorem-Gauss-RD-CEO} follows simply by evaluating the region $\mc{RD}_\mathrm{CEO}^\mathrm{I}$ as described in Theorem~\ref{theorem-continuous-RD1-CEO} using Gaussian test channels and no time-sharing. Specifically, we set $Q= \emptyset$ and $p(u_k|\dv y_k,q) = \mc{CN}(\dv y_k,  \dv\Sigma_k^{1/2}(\dv\Omega_k-\dv I) \dv\Sigma_k^{1/2})$, $k\in \mc K$. The proof of the converse appears in Appendix~\ref{proof-converse-Gauss-RD-CEO}.	
\end{proof}

In the case in which the noises at the agents are independent among them and from the noise $\dv N_0$ at the CEO, the result of Theorem~\ref{theorem-Gauss-RD-CEO} takes a simpler form which is stated in the following corollary.       

\begin{corollary}~\label{corollary-Gauss-RD-CEO}
Consider the vector Gaussian CEO problem described by~\eqref{mimo-gaussian-model} and~\eqref{mimo-gaussian-model-2} with the noises $(\dv N_1,\ldots,\dv N_K)$ being independent among them and with $\dv N_0$. Under logarithmic loss, the rate-distortion region this model is given by the set of all non-negative tuples $(R_1,\ldots, R_K,D)$ that satisfy, for all subsets $\mc S \subseteq \mc K$,  
\begin{align*}
D + \sum_{k \in \mc S} R_k \geq  \sum_{k\in \mc S} \log \frac{1}{\left| \dv I - \dv\Omega_k \dv\Sigma_k \right|} + \log \bigg| (\pi e) \big( \dv\Sigma_{\dv x}^{-1} + \dv H_0^\dagger \dv\Sigma_0^{-1} \dv H_0 + \sum_{k \in \mc S^c} \dv H_k^\dagger \dv\Omega_k \dv H_k \big)^{-1} \bigg| \;,
\end{align*}
for some matrices $\{\dv\Omega_k\}_{k=1}^K$ such that $\dv 0 \preceq \dv\Omega_k \preceq \dv\Sigma_k^{-1}$. \qedblack
\end{corollary}

\begin{remark}The direct part of Theorem~\ref{theorem-Gauss-RD-CEO} shows that Gaussian test channels and no-time sharing exhaust the region. For the converse proof of Theorem~\ref{theorem-Gauss-RD-CEO}, we derive an outer bound on the region $\mc{RD}_\mathrm{CEO}^\mathrm{I}$. In doing so, we use the de Bruijn identity, a connection between differential entropy and Fisher information, along with the properties of MMSE and Fisher information. By opposition to the case of quadratic distortion measure for which the application of this technique was shown in~\cite{EU14} to result in an outer bound that is generally non-tight, Theorem~\ref{theorem-Gauss-RD-CEO} shows that the approach is successful in the case of logarithmic loss distortion measure as it yields a complete characterization of the region. On this aspect, note that in the specific case of scalar Gaussian sources, an alternate converse proof may be obtained by extending that of the scalar Gaussian many-help-one source coding problem by Oahama~\cite{O05} and Prabhakaran \textit{et al.}~\cite{PTR04} through accounting for additional side information at CEO and replacing the original mean square error distortion constraint with conditional entropy. However, such approach does not seem conclusive in the vector case, as the entropy power inequality is known to be generally non-tight in this setting~\cite{TP05,WSS06}.
\qedblack
\end{remark}

\begin{remark}
The result of Theorem~\ref{theorem-Gauss-RD-CEO} generalizes that of~\cite{TC09}, which considers the case of only one agent, i.e., the remote vector Gaussian Wyner-Ziv model under logarithmic loss, to the case of an arbitrarily number of agents. The converse proof of~\cite{TC09}, which relies on the technique of orthogonal transform to reduce the vector setting to one of parallel scalar Gaussian settings, seems insufficient to diagonalize all the noise covariance matrices simultaneously in the case of more than one agent. The result of Theorem~\ref{theorem-Gauss-RD-CEO} is also connected to recent developments on characterizing the capacity of multiple-input multiple-output (MIMO) relay channels in which the relay nodes are connected to the receiver through error-free finite-capacity links (i.e., the so-called cloud radio access networks). In particular, the reader may refer to~\cite[Theorem 4]{ZXYC16} where important progress is done, and \cite{EZSC17} where compress-and-forward with joint decompression-decoding is shown to be optimal under the constraint of oblivious relay processing. \qedblack
\end{remark}

\subsection{Gaussian Test Channels with Time-Sharing Exhaust the Berger-Tung Region}\label{section-Berger-Tung}

Theorem~\ref{theorem-continuous-RD1-CEO} shows that the union of all rate-distortion tuples that satisfy~\eqref{equation-continous-RD1-CEO} for all subsets $\mc S \subseteq \mc K$ coincides
with the Berger-Tung inner bound in which time-sharing is used. The direct part of Theorem~\ref{theorem-Gauss-RD-CEO} is obtained by evaluating~\eqref{equation-continous-RD1-CEO} using Gaussian test channels and no time-sharing, i.e., $Q=\emptyset$, not the Berger-Tung inner bound. The reader may wonder: i) whether Gaussian test channels also exhaust the Berger-Tung inner bound for the vector Gaussian CEO problem that we study here, and ii) whether time-sharing is needed with the Berger-Tung scheme. In this section, we answer both questions in the affirmative. In particular, we show that the Berger-Tung coding scheme with Gaussian test channels and time-sharing achieves distortion levels that are not larger than any other coding scheme.


\begin{proposition}~\label{proposition-continous-RD2-CEO}
The rate-distortion region for the vector Gaussian CEO problem under logarithmic loss is given by  
\begin{equation*}
\mc{RD}_{\mathrm{VG}\text{-}\mathrm{CEO}}^\star = \bigcup \mc{RD}_\mathrm{CEO}^\mathrm{II} (V_1^\mathrm{G},\ldots,V_K^\mathrm{G},Q^\prime) \;,
\end{equation*}
where $\mc{RD}_\mathrm{CEO}^\mathrm{II} (\cdot)$ is as given in Definition~\ref{defintion-continous-RD2-CEO} and the superscript $\mathrm{G}$ is used to denote that the union is taken over Gaussian distributed $V_k^{\mathrm{G}}\sim p(v_k|\dv y_k,q')$ conditionally on $(\dv Y_k, Q')$.
\end{proposition}

\begin{proof} For the proof of Proposition~\ref{proposition-continous-RD2-CEO},  it is sufficient to show that, for fixed Gaussian conditional distributions $\{p(u_k|\dv y_k)\}_{k=1}^K$, the extreme points of the polytopes defined by~\eqref{equation-continous-RD1-CEO} are \textit{dominated} by points that are in $\mc{RD}_\mathrm{CEO}^\mathrm{II}$ and which are achievable using Gaussian conditional distributions $\{p(v_k|\dv y_k, q')\}_{k=1}^K$. Hereafter, we give a brief outline of proof for the case $K=2$. The proof for $K \geq 2$ follows similarly; and is omitted for brevity. Consider the inequalities~\eqref{equation-continous-RD1-CEO} with $Q=\emptyset$ and $(U_1,U_2) := (U^\mathrm{G}_1,U^\mathrm{G}_2)$ chosen to be Gaussian (see Theorem~\ref{theorem-Gauss-RD-CEO}). Consider now the extreme points of the polytopes defined by the obtained inequalities:
\begin{align*}
P_1 &= (0,0,I(\dv Y_1;U^\mathrm{G}_1|\dv X,\dv Y_0) + I(\dv Y_2;U^\mathrm{G}_2|\dv X,\dv Y_0)+ h(\dv X|\dv Y_0))\\
P_2 &= (I(\dv Y_1;U^\mathrm{G}_1|\dv Y_0),0,I(U^\mathrm{G}_2;\dv Y_2|\dv X,\dv Y_0) + h(\dv X|U^\mathrm{G}_1,\dv Y_0))\\
P_3 &= (0,I(\dv Y_2;U^\mathrm{G}_2|\dv Y_0),I(U^\mathrm{G}_1;\dv Y_1|\dv X,\dv Y_0) + h(\dv X|U^\mathrm{G}_2,\dv Y_0))\\
P_4 &= (I(\dv Y_1;U^\mathrm{G}_1|\dv Y_0),I(\dv Y_2;U^\mathrm{G}_2|U^\mathrm{G}_1,\dv Y_0), h(\dv X|U^\mathrm{G}_1,U^\mathrm{G}_2,\dv Y_0))\\
P_5 &= (I(\dv Y_1;U^\mathrm{G}_1|U^\mathrm{G}_2,\dv Y_0), I(\dv Y_2;U^\mathrm{G}_2|\dv Y_0), h(\dv X|U^\mathrm{G}_1,U^\mathrm{G}_2,\dv Y_0)) \;,
\end{align*} \hspace{-0.4em}
where the point $P_j$ is a a triple $(R_1^{(j)}, R_2^{(j)}, D^{(j)})$. It is easy to see that each of these points is \textit{dominated} by a point in $\mc{RD}_\mathrm{CEO}^\mathrm{II}$, i.e., there exists $(R_1,R_2,D) \in \mc{RD}_\mathrm{CEO}^\mathrm{II}$ for which $R_1 \leq R_1^{(j)}$, $R_2\leq R_2^{(j)}$ and $D \leq D^{(j)}$. To see this, first note that $P_4$ and $P_5$ are both in $\mc{RD}_\mathrm{CEO}^\mathrm{II}$. Next, observe that the point $(0,0,h(\dv X|\dv Y_0))$ is in $\mc{RD}_\mathrm{CEO}^\mathrm{II}$, which is clearly achievable by letting $(V_1,V_2,Q') = (\emptyset,\emptyset,\emptyset)$, dominates $P_1$. Also, by using letting $(V_1,V_2,Q') = (U^\mathrm{G}_1,\emptyset,\emptyset)$,  we have that the point $(I(\dv Y_1;U_1|\dv Y_0),0, h(\dv X|U_1,\dv Y_0))$ is in $\mc{RD}_\mathrm{CEO}^\mathrm{II}$, and dominates the point $P_2$. A similar argument shows that $P_3$ is dominated by a point in $\mc{RD}_\mathrm{CEO}^\mathrm{II}$. The proof is terminated by observing that, for all above corner points, $V_k$ is set either equal $U^\mathrm{G}_k$ (which is Gaussian distributed conditionally on $\dv Y_k$) or a constant. 	
\end{proof} 

\begin{remark}
Proposition~\ref{proposition-continous-RD2-CEO} shows that for the vector Gaussian CEO problem with side information under a logarithmic loss constraint, vector Gaussian quantization codebooks with time-sharing are optimal. In the case of quadratic distortion constraint, however, a characterization of the rate-distortion region is still to be found in general, and it is not known yet whether vector Gaussian quantization codebooks (with or without time-sharing) are optimal, except in few special cases such as that of scalar Gaussian sources or the case of only one agent, i.e., the remote vector Gaussian Wyner-Ziv problem whose rate-distortion region is found in~\cite{TC09}. In~\cite{TC09}, Tian and Chen also found the rate-distortion region of the remote vector Gaussian Wyner-Ziv problem under logarithmic loss, which they showed achievable using Gaussian quantization codebooks that are different from those (also Gaussian) that are optimal in the case of quadratic distortion. As we already mentioned, our result of Theorem~\ref{theorem-Gauss-RD-CEO} generalizes that of~\cite{TC09} to the case of an arbitrary number of agents. \qedblack 
\end{remark}

\begin{remark}
One may wonder whether giving the decoder side information $\dv Y_0$ to the encoders is beneficial. Similar to the well known result in Wyner-Ziv source coding of scalar Gaussian sources, our result of Theorem~\ref{theorem-Gauss-RD-CEO} shows that encoder side information does not help. \qedblack
\end{remark}

\section{Applications}\label{secV}

In this section, we show that application of the result of Theorem~\ref{theorem-Gauss-RD-CEO} allows us to find the complete solutions of two related problems: a quadratic vector Gaussian CEO problem with determinant constraint and the vector Gaussian distributed Information Bottleneck problem. For the case of discrete data, we provide an example application to distributed pattern classification. 

\subsection{Quadratic Vector Gaussian CEO Problem with Determinant Constraint}\label{secV_subsecB}

We now turn to the case in which the distortion is measured under quadratic loss. In this case, the mean square error matrix is defined by
\begin{equation}~\label{equation-mean-square-error}
\dv D^{(n)} := \frac{1}{n} \sum_{i=1}^{n} \E [ (\dv X_i - \hat{\dv X}_i) (\dv X_i - \hat{\dv X}_i)^\dagger ] \;.
\end{equation}
Under a (general) error constraint of the form 
\begin{equation}~\label{vector-Gaussian-ceo-quadratic-measure-matrix-constraint}
\dv D^{(n)} \preceq  \dv D \;,
\end{equation}
where $\dv D$ designates here a prescribed positive definite error matrix, a complete solution is still to be found in general. In what follows, we replace the constraint~\eqref{vector-Gaussian-ceo-quadratic-measure-matrix-constraint} with one on the \textit{determinant} of the error matrix $\dv D^{(n)}$, i.e.,
\begin{equation}~\label{vector-Gaussian-ceo-quadratic-measure-det-constraint}
|\dv D^{(n)}| \leq D \;,
\end{equation}
($D$ is a scalar here). We note that since the error matrix $\dv D^{(n)}$ is minimized by choosing the decoding as  
\begin{equation}\label{equation-decoder}
\hat{\dv X}_i = \E [\dv X_i|\breve{\phi}^{(n)}_1(\dv Y_1^n),\ldots,\breve{\phi}^{(n)}_K(\dv Y_K^n),\dv Y_0^n] \;,
\end{equation}
where $\{\breve{\phi}^{(n)}_k\}_{k=1}^K$ denote the encoding functions, without loss of generality we can write~\eqref{equation-mean-square-error} as 
\begin{equation}
\dv D^{(n)} = \frac{1}{n} \sum_{i=1}^{n} \mathrm{mmse}(\dv X_i|\breve{\phi}^{(n)}_1(\dv Y_1^n),\ldots,\breve{\phi}^{(n)}_K (\dv Y_K^n),\dv Y_0^n) \;. 
\end{equation}

\begin{definition}
A rate-distortion tuple $(R_1,\ldots,R_K,D)$ is achievable for the quadratic vector Gaussian CEO problem with determinant constraint if there exist a blocklength $n$, $K$ encoding functions $\{\breve{\phi}^{(n)}_k\}^K_{k=1}$ such that
\begin{align*}
R_k &\geq \frac{1}{n}\log M^{(n)}_k, \quad\text{for}\:\: k=1,\ldots,K,\\
D &\geq \bigg|\frac{1}{n} \sum_{i=1}^{n} \mathrm{mmse}(\dv X_i|\breve{\phi}^{(n)}_1(\dv Y_1^n),\ldots,\breve{\phi}^{(n)}_K(\dv Y_K^n),\dv Y_0^n) \bigg| \;.
\end{align*} 
\noindent The rate-distortion region $\mc{RD}_{\mathrm{VG}\text{-}\mathrm{CEO}}^\mathrm{det}$ is defined as the closure of all non-negative tuples $(R_1,\ldots,R_K,D)$ that are achievable. 
\qedblack
\end{definition}

The following theorem characterizes the rate-distortion region of the quadratic vector Gaussian CEO problem with determinant constraint.

\begin{theorem}~\label{theorem-quadratic-RD-CEO}
The rate-distortion region $\mc{RD}_{\mathrm{VG}\text{-}\mathrm{CEO}}^\mathrm{det}$ of the quadratic vector Gaussian CEO problem with determinant constraint is given by the set of all non-negative tuples $(R_1,\ldots,R_K,D)$ that satisfy, for all subsets $\mc S \subseteq \mc K$, 
\begin{align*}
\log \frac{1}{D} \leq \sum_{k \in \mc S} R_k  + \log|\dv I-\dv\Omega_k \dv\Sigma_k|  
+ \log\left| \dv\Sigma_{\dv x}^{-1} + \dv H_{\bar{\mc S}}^\dagger \dv\Sigma_{\dv n_{\bar{\mc S}}}^{-1} \big( \dv I - \dv\Lambda_{\bar{\mc S}} \dv\Sigma_{\dv n_{\bar{\mc S}}}^{-1}  \big) \dv H_{\bar{\mc S}}  \right| \;,
\end{align*}
for matrices $\{\dv\Omega_k\}_{k=1}^K$ such that $\dv 0 \preceq \dv\Omega_k \preceq \dv\Sigma_k^{-1}$, where $\bar{\mc S}=\{0\} \cup \mc S^c$ and $\boldsymbol{\Lambda}_{\bar{\mc S}}$ is as defined by~\eqref{equation-definition-T}. 
\end{theorem}

\begin{proof} 
The proof of Theorem~\ref{theorem-quadratic-RD-CEO} is given in Appendix~\ref{proof-quadratic-RD-CEO}.	
\end{proof}

\begin{remark}
It is believed that the approach of this section, which connects the quadratic vector Gaussian CEO problem to that under logarithmic loss, can also be exploited to possibly infer other new results on the quadratic vector Gaussian CEO problem. Alternatively, it can also be used to derive new converses on the quadratic vector Gaussian CEO problem. For example, in the case of scalar sources Theorem~\ref{theorem-quadratic-RD-CEO} and Lemma~\ref{lemma-connection} readily provide an alternate converse proof to those of~\cite{O05, PTR04} for this model. Similar connections were made in~\cite{C15,C18} where it was observed that the results of~\cite{CW14} can be used to recover known results on the scalar Gaussian CEO problem (such as the sum rate-distortion region of~\cite{WTV08}) and the scalar Gaussian two-encoder distributed source coding problem. We also point out that similar information constraints have been applied to log-determinant reproduction constraints previously in~\cite{CJ14}.
\qedblack 
\end{remark}

\subsection{Distributed Vector Gaussian Information Bottleneck}\label{secV_subsecC}

Consider now the vector Gaussian CEO problem with side information of Section~\ref{section-vector-CEO}, and let the logarithmic loss distortion constraint be replaced by the mutual information constraint 
\begin{equation}
I\left(\dv X^n; \psi^{(n)}\left(\phi^{(n)}_1(\dv Y^n_1),\ldots,\phi^{(n)}_K(\dv Y^n_K), \dv Y^n_0\right)\right) \geq n \Delta \;.
\end{equation}
In this case, the region of optimal tuples $(R_1,\ldots,R_K,\Delta)$ generalizes the \textit{Gaussian Information Bottleneck Function} of~\cite{CGTW05,WM14,WFM14} to the setting in which the decoder observes correlated side information $\dv Y_0$ and the inference is done in a distributed manner by $K$ learners. This region can be obtained readily from Theorem~\ref{theorem-Gauss-RD-CEO} by substituting therein $\Delta := h(\dv X) - D$. The following corollary states the result.

\begin{corollary}~\label{corollary-distributed-Gaussian-IB}
For the problem of distributed Gaussian Information Bottleneck with side information at the predictor, the complexity-relevance region is given by the union of all non-negative tuples $(R_1,\ldots,R_K,\Delta)$ that satisfy, for every $\mc{S} \subseteq \mc{K}$,
\begin{align*}
\Delta &\leq \sum_{k \in \mc S} \big(R_k + \log \left| \dv I - \dv\Omega_k \dv\Sigma_k \right| \big) + \log \big| \dv I + \dv\Sigma_{\dv x} \dv H_{\bar{\mc S}}^\dagger \dv\Sigma_{\dv n_{\bar{\mc S}}}^{-1}\big( \dv I - \boldsymbol{\Lambda}_{\bar{\mc S}} \dv\Sigma_{\dv n_{\bar{\mc S}}}^{-1}\big)\dv H_{\bar{\mc S}}\big| \;,
\end{align*}
for matrices $\{\dv\Omega_k\}_{k=1}^K$ such that $\dv 0 \preceq \dv\Omega_k \preceq \dv\Sigma_k^{-1}$, where $\bar{\mc S}=\{0\} \cup \mc S^c$ and $\boldsymbol{\Lambda}_{\bar{\mc S}}$ is given by~\eqref{equation-definition-T}. 
\qedblack
\end{corollary}

\noindent In particular, if $K=1$ and $\dv Y_0=\emptyset$, with the substitutions $\dv Y := \dv Y_1$, $R := R_1$, $\dv H := \dv H_1$, $\dv\Sigma := \dv\Sigma_1$, and $\dv\Omega_1 := \dv\Omega$, the region of Corollary~\ref{corollary-distributed-Gaussian-IB} reduces to the set of pairs $(R,\Delta)$  that satisfy 
\begin{subequations}~\label{equivalent-form-rate-distortion-vector-Gaussian-information-bottleneck-method}	
\begin{align}
\Delta &\leq  \log \big|\dv I + \dv\Sigma_{\dv x} \dv H^\dagger \dv\Omega \dv H \big|\\
\Delta  &\leq R + \log \big| \dv I - \dv\Omega \dv\Sigma \big| \;, 
\end{align} 
\end{subequations}
for some matrix $\dv\Omega$ such that $\dv 0 \preceq \dv\Omega \preceq \dv\Sigma^{-1}$.
		
\noindent Expression~\eqref{equivalent-form-rate-distortion-vector-Gaussian-information-bottleneck-method} is known as the \textit{Gaussian Information Bottleneck Function}~\cite{CGTW05,WM14,WFM14}, which is the solution of the Information Bottleneck method of~\cite{TPB99} in the case of jointly Gaussian variables. More precisely, using the terminology of~\cite{TPB99}, the inequalities~\eqref{equivalent-form-rate-distortion-vector-Gaussian-information-bottleneck-method} describe the optimal trade-off between the complexity (or rate) $R$ and the relevance (or accuracy) $\Delta$. The concept of Information Bottleneck was found useful in various learning applications, such as for data clustering~\cite{ST01}, feature selection~\cite{BE-YL04} and others, indluding in distributed settings~\cite{EZ18,EZ20}.

\noindent Furthermore, if in~\eqref{mimo-gaussian-model} and~\eqref{mimo-gaussian-model-2} the noises are independent among them and from $\dv N_0$, the relevance-complexity region of Corollary~\ref{corollary-distributed-Gaussian-IB} reduces to the union of all non-negative tuples $(R_1,\ldots,R_K,\Delta)$ that satisfy, for every $\mc{S} \subseteq \mc{K}$,
\begin{align*}
\Delta &\leq \sum_{k \in \mc S} \big(R_k + \log \left| \dv I - \dv\Omega_k \dv\Sigma_k \right| \big) + \log \big| \dv I + \dv\Sigma_{\dv x} \big(\dv H_0^\dagger \dv\Sigma_0^{-1} \dv H_0 + \sum_{k \in \mc S^c} \dv H_k^\dagger \dv\Omega_k \dv H_k\big) \big| \;,
\end{align*}
for some matrices $\{\dv\Omega_k\}_{k=1}^K$ such that $\dv 0 \preceq \dv\Omega_k \preceq \dv\Sigma_k^{-1}$.

\vspace{0.5em}
\begin{figure}[h!]
	\centering
	\includegraphics[width=0.65\linewidth]{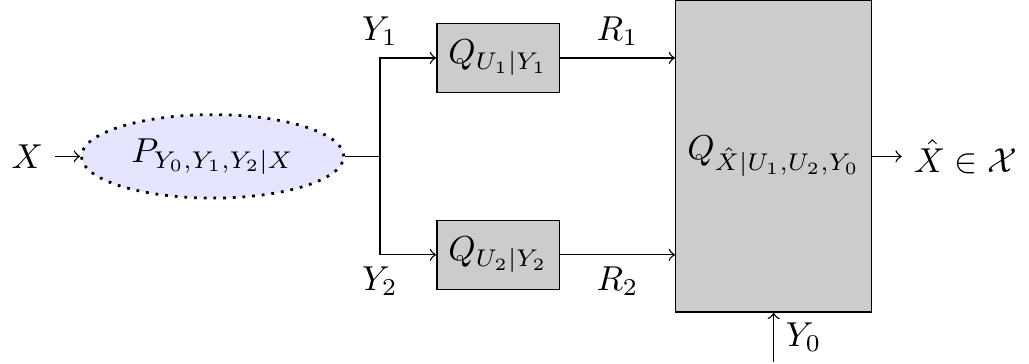}
	\caption{An example of distributed pattern classification.} 
	\squeezeup
	\label{fig-distributed-pattern-classification}
\end{figure}

\subsection{Distributed Pattern Classification}\label{section-classification} 

Consider the problem of distributed pattern classification shown in Figure~\ref{fig-distributed-pattern-classification}. In this example, the decoder is a predictor whose role is to guess the unknown class $X \in \mc X$ of a measurable pair $(Y_1,Y_2) \in \mc Y_1 \times \mc Y_2$ on the basis of inputs from two learners as well as its own observation about the target class, in the form of some correlated  $Y_0 \in \mc Y_0$. It is assumed that $ Y_1 \mkv (X,Y_0) \mkv Y_2$. The first learner produces its input based only on $Y_1 \in \mc Y_1$; and the second learner produces its input based only on $Y_2 \in \mc Y_2$. For the sake of a smaller \textit{generalization gap}\footnote{The generalization gap, defined as the difference between the empirical risk (average risk over a finite training sample) and the population risk  (average risk over the true joint distribution), can be upper bounded using the mutual information between the learner's inputs and outputs, see, e.g., \cite{RZ15, XR17} and the recent~\cite{AAV18}, which provides a fundamental justification of the use of the \textit{minimum description length} (MDL) constraint on the learners mappings as a regularizer term.}, the inputs of the learners are restricted to have description lengths that are no more than $R_1$ and $R_2$ bits per sample, respectively. Let $Q_{U_1|Y_1}\: :\:  \mc Y_1 \longrightarrow \mc P(\mc U_1)$ and  $Q_{U_2|Y_2}\: :\:  \mc Y_2 \longrightarrow \mc P(\mc U_2)$ be two (stochastic) such learners. Also, let $Q_{\hat{X}|U_1,U_2,Y_0}\: :\: \mc U_1 \times \mc U_2 \times \mc Y_0 \longrightarrow \mc P(\mc X)$ be a soft-decoder or predictor that maps the pair of representations $(U_1,U_2)$ and $Y_0$ to a probability distribution on the label space $\mc X$. The pair of learners and predictor induce a classifier 
\begin{align}
Q_{\hat{X}|Y_0,Y_1,Y_2}(x|y_0,y_1,y_2) &= \sum_{u_1 \in \mc U_1} Q_{U_1|Y_1}(u_1|y_1) \sum_{u_2 \in \mc U_2} Q_{U_2|Y_2}(u_2|y_2) Q_{\hat{X}|U_1,U_2,Y_0}(x|u_1,u_2,y_0) \nonumber\\ 
 &= \mathbb{E}_{Q_{U_1|Y_1}} \mathbb{E}_{Q_{U_2|Y_2}} [Q_{\hat{X}|U_1,U_2,Y_0}(x|U_1,U_2,y_0)]  \;,
\label{definition-equivalent-classifier}
\end{align} 
whose probability of classification error is defined as
\begin{equation}~\label{definition-probability-classification-error}
P_{\mathcal E}(Q_{\hat{X}|Y_0,Y_1,Y_2}) = 1 - \mathbb{E}_{P_{X,Y_0,Y_1,Y_2}} [Q_{\hat{X}|Y_0,Y_1,Y_2}(X|Y_0,Y_1,Y_2)] \;.
\end{equation}
Let $\mc{RD}_\mathrm{CEO}^\star$ be the rate-distortion region of the associated two-encoder DM CEO problem with side information as given by Theorem~\ref{theorem-continuous-RD1-CEO}. The following proposition shows that there exists a classifier $Q^{\star}_{\hat{X}|Y_0,Y_1,Y_2}$ for which the probability of misclassification can be upper bounded in terms of the minimal average logarithmic loss distortion that is achievable for the rate pair $(R_1,R_2)$ in $\mc{RD}_\mathrm{CEO}^\star$.

\begin{proposition}~\label{proposition-example-distributed-classification}
For the problem of distributed pattern classification of Figure~\ref{fig-distributed-pattern-classification}, there exists a classifier  $Q^{\star}_{\hat{X}|Y_0,Y_1,Y_2}$ for which the probability of classification error satisfies
\begin{equation*}
P_{\mathcal E}(Q^{\star}_{\hat{X}|Y_0,Y_1,Y_2})  \leq 1 - \exp\left( - \inf\{ D \: : \: (R_1,R_2, D) \in \mc{RD}_\mathrm{CEO}^\star \} \right) \;,
\end{equation*}
where $\mc{RD}_\mathrm{CEO}^\star$ is the rate-distortion region of the associated two-encoder DM CEO problem with side information as given by Theorem~\ref{theorem-continuous-RD1-CEO}.
\end{proposition}

\begin{proof}
The proof of Proposition~\ref{proposition-example-distributed-classification} is given in Appendix~\ref{proof-proposition-example-distributed-classification}.
\end{proof}

\vspace{1em}
To make the above example more concrete, consider the following scenario where $Y_0$ plays the role of information about the sub-class of the label class $X \in \{0,1,2,3\}$.  More specifically, let $S$ be a random variable that is uniformly distributed over $\{1,2\}$. Also, let $X_1$ and $X_2$ be two random variables that are independent between them and from $S$, distributed uniformly over $\{1,3\}$ and $\{0,2\}$ respectively. The state $S$ acts as a random switch that connects $X_1$ or $X_2$ to $X$, i.e.,
\begin{equation}~\label{example-distributed-classification-label-variable}
X = X_S \;.
\end{equation}
That is, if $S=1$ then $X=X_1$, and if $S=2$ then $X=X_2$. Thus, the value of $S$ indicates whether $X$ is odd- or even-valued (i.e., the sub-class of $X$). Also, let  
\begin{subequations}~\label{example-distributed-classification-attributes}
\begin{align}
 Y_0 &= S \\
 Y_1 &= X_S \oplus Z_1 \\
 Y_2 &= X_S \oplus Z_2 \;,
\end{align} 
\end{subequations}
where $Z_1$ and $Z_2$ are Bernoulli-$(p)$ random variables, $p\in(0,1)$, that are independent between them, and from $(S,X_1,X_2)$, and the addition is modulo $4$. For simplification, we let $R_1=R_2=R$. We numerically approximate the set of $(R,D)$ pairs such that $(R,R,D)$ is in the rate-distortion region $\mc{RD}_\mathrm{CEO}^\star$ corresponding to the CEO network of this example. The algorithm that we use for the computation will be described in detail in Section~\ref{section-BA-algortihm-DM}. The lower convex envelope of these $(R,D)$ pairs is plotted in Figure~\ref{fig-bound-on-classification-error-subfiga} for $p \in \{0.01, 0.1, 0.25, 0.5\}$. Continuing our example, we also compute the upper bound on the probability of classification error according to Proposition~\ref{proposition-example-distributed-classification}. The result is given in Figure~\ref{fig-bound-on-classification-error-subfigb}. Observe that if $Y_1$ and $Y_2$ are high-quality estimates of $X$ (e.g., $p=0.01$), then a small increase in the \textit{complexity} $R$ results in a large relative improvement of the (bound on) the probability of classification error. On the other hand, if $Y_1$ and $Y_2$ are low-quality estimates of $X$ (e.g., $p=0.25$) then we require a large increase of $R$ in order to obtain an appreciable reduction in the error probability. Recalling that larger $R$ implies lesser generalization capability~\cite{RZ15, XR17, AAV18}, these numerical results are consistent with the fact that classifiers should strike a good balance between accuracy and their ability to generalize well to unseen data. Figure~\ref{fig-bound-on-classification-error-subfigc} quantifies the value of side information $S$ given to both learners and predictor, none of them, or only the predictor, for $p=0.25$.

\begin{figure*}[!h]
	\begin{center}
		\subfloat[]{ 
			\resizebox{.321\linewidth}{!}{\includegraphics{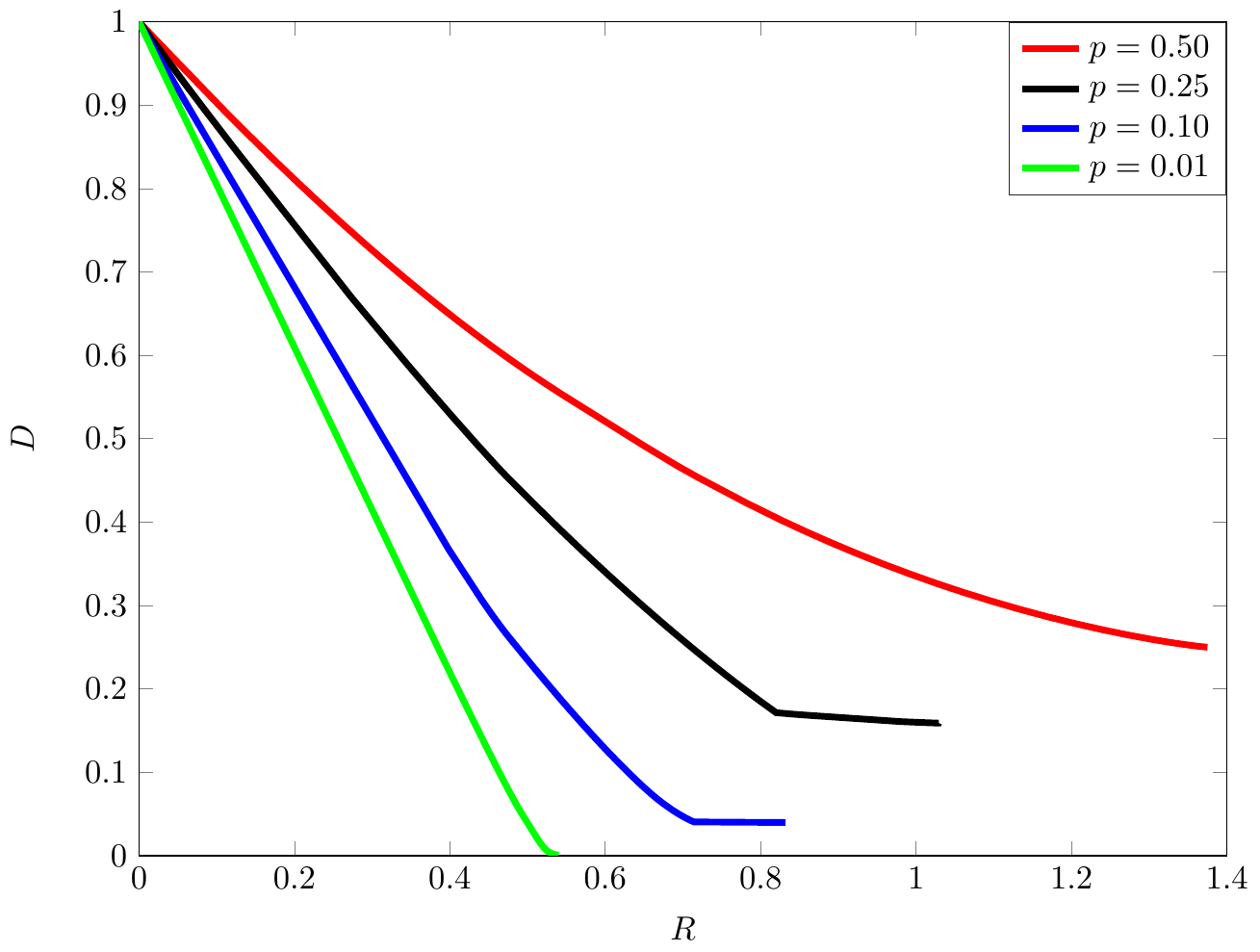}}
			\label{fig-bound-on-classification-error-subfiga} 
		} 
		\subfloat[]{ 
			\resizebox{.321\linewidth}{!}{\includegraphics{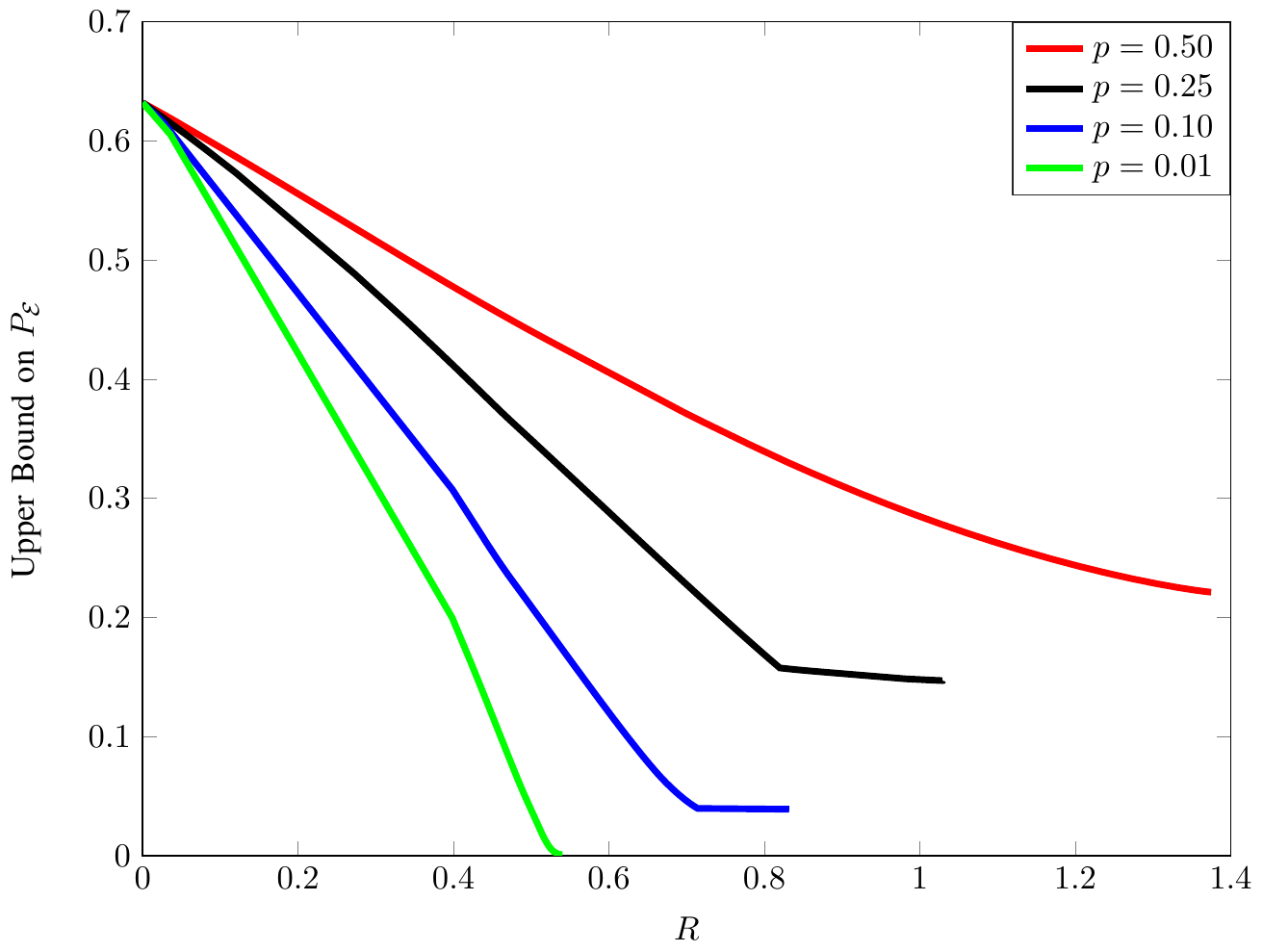}}
			\label{fig-bound-on-classification-error-subfigb} 
		} 
		\subfloat[]{
			\resizebox{.321\linewidth}{!}{\includegraphics{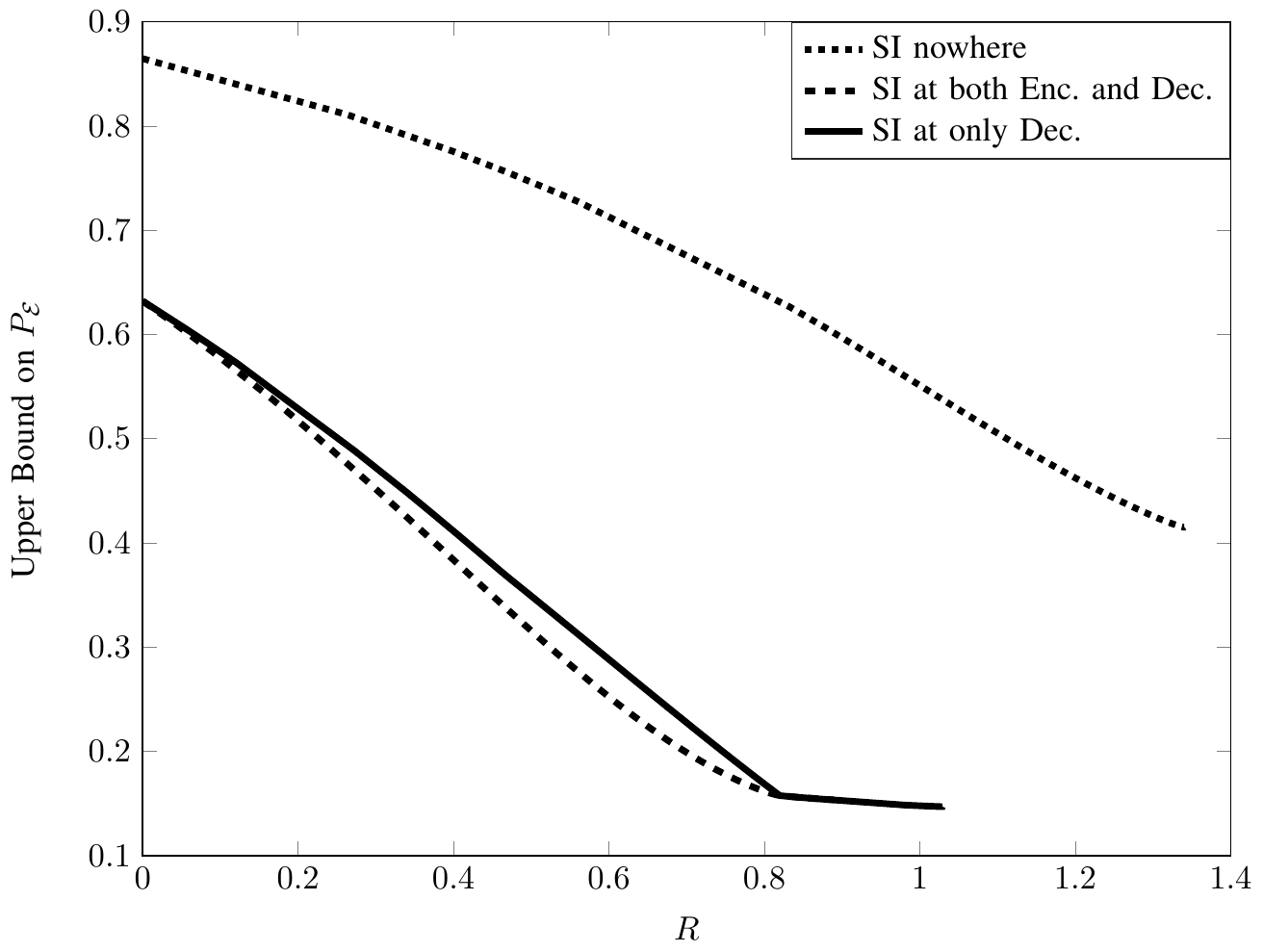}}
			\label{fig-bound-on-classification-error-subfigc}
		}
		\vspace{-0.5em}
		\caption{Illustration of the bound on the probability of classification error of Proposition~\ref{proposition-example-distributed-classification}  for the example described by~\eqref{example-distributed-classification-label-variable} and~\eqref{example-distributed-classification-attributes}. (a) Distortion-rate function of the network of Figure~\ref{fig-distributed-pattern-classification} computed for $p \in \{0.01, 0.1, 0.25, 0.5\}$. (b) Upper bound on the probability of classification error computed according to Proposition~\ref{proposition-example-distributed-classification}. (c) Effect of side information (SI) $Y_0$ when given to both learners and the predictor, only the predictor or none of them.}
		\label{fig-bound-on-classification-error} 
		\vspace{-1em}
	\end{center}
\end{figure*}

\section{Blahut-Arimoto Type Algorithms}\label{section-BA-algorithm}

In this section, we develop iterative algorithms that allow to compute the rate-distortion regions of the DM and vector Gaussian CEO problems numerically. We illustrate the efficiency of our algorithms through some numerical examples.

\subsection{Discrete Case}\label{section-BA-algortihm-DM}

Here we develop a BA-type algorithm that allows to compute the convex region $\mc{RD}_\mathrm{CEO}^\star$ for general discrete memoryless sources. To develop the algorithm, we use the Berger-Tung form of the region given in Definition~\ref{defintion-continous-RD2-CEO} for $K=2$. The outline of the proposed method is as follows. First, we rewrite the rate-distortion region $\mc{RD}_\mathrm{CEO}^\star$ in terms of the union of two simpler regions in Proposition~\ref{proposition-region-convex-hull}. The tuples lying on the boundary of each region are parametrically given in Theorem~\ref{theorem-CEO-parametrization}. Then, the boundary points of each simpler region are computed numerically via an alternating minimization method derived in Section~\ref{section-BA-computation-region1} and detailed in Algorithm~\ref{algo-CEO}. Finally, the original rate-distortion region is obtained as the convex hull of the union of the tuples obtained for the two simple regions.

\vspace{0.5em}
\subsubsection{Equivalent Parametrization}~\label{section-BA-equivalent-parametrization}

Define the two regions $\mc{RD}_\mathrm{CEO}^k$, $k = 1, 2,$ as  
\begin{equation}~\label{constituent-rate-distortion-regions}
\mathcal{RD}_\mathrm{CEO}^{k}=\{ (R_{1},R_{2},D) \::\: D\geq D_\mathrm{CEO}^k (R_{1},R_{2}) \} \;,
\end{equation}
with
\begin{align}
D_\mathrm{CEO}^k (R_1,R_2) := & \argmin \; H(X|U_1,U_2,Y_0) \label{equation-CEO-distortion-rate-function}\\
&\text{s.t.} \quad 
R_k \geq I(Y_k;U_k|U_\kbar,Y_0) \quad \text{and} \quad R_\kbar \geq I(X_\kbar;U_\kbar|Y_0) \;,  \nonumber
\end{align}
and the minimization is over set of joint measures $P_{U_1,U_2,X,Y_0,Y_1,Y_2}$ that satisfy $U_1 \mkv Y_1 \mkv (X, Y_0) \mkv Y_2 \mkv U_2$. (We define $\kbar:=k \Mod{2} + 1$ for $k=1,2$.)

As stated in the following proposition, the region $\mc{RD}_\mathrm{CEO}^\star$ of Theorem~\ref{theorem-continuous-RD1-CEO} coincides with the convex hull of the union of the two regions $\mc{RD}_\mathrm{CEO}^1$ and  $\mc{RD}_\mathrm{CEO}^2$. 

\begin{proposition}~\label{proposition-region-convex-hull}
The region $\mc{RD}_\mathrm{CEO}^\star$ is given by
\begin{align}~\label{equation-region-convex-hull}
\mc{RD}_\mathrm{CEO}^\star = \mathrm{conv} ( \mc{RD}_\mathrm{CEO}^1 \cup \mc{RD}_\mathrm{CEO}^2 ) \;.
\end{align}
\end{proposition}

\begin{proof}
An outline of the proof is as follows. Let $P_{U_1,U_2,X,Y_0,Y_1,Y_2}$ and $P_Q$ be such that $(R_1,R_2,D) \in \mc{RD}_\mathrm{CEO}^\star$. The polytope defined by the rate constraints~\eqref{equation-rate-constraints}, denoted by $\mc V$, forms a contra-polymatroid with $2!$ extreme points (vertices)~\cite{CW14,CB08}. Given a permutation $\pi$ on $\{1,2\}$, the tuple
\begin{equation*}
\tilde{R}_{\pi(1)} = I(Y_{\pi(1)};U_{\pi(1)}|Y_0) \;, \quad \tilde{R}_{\pi(2)} = I(Y_{\pi(2)};U_{\pi(2)}|U_{\pi(1)},Y_0) \;,
\end{equation*}
defines an extreme point  of $\mc V$ for each permutation. As shown in \cite{CW14}, for every extreme point $(\tilde{R}_1,\tilde{R}_2)$ of $\mc V$, the point $(\tilde{R}_1,\tilde{R}_2,D)$ is achieved by time-sharing two successive Wyner-Ziv (WZ) strategies. The set of  achievable tuples with such successive WZ scheme is characterized by the convex hull of $\mc{RD}_\mathrm{CEO}^{\pi(1)}$. Convexifying the union of both regions as in~\eqref{equation-region-convex-hull}, we obtain the full rate-distortion region $\mc{RD}^{\star}_\mathrm{CEO}$.
\end{proof}

The main advantage of Proposition~\ref{proposition-region-convex-hull} is that it reduces the computation of region $\mc{RD}_\mathrm{CEO}^\star$ to the computation of the two regions $\mc{RD}_\mathrm{CEO}^k$, $k=1,2$, whose boundary can be efficiently parametrized, leading to an efficient computational method. In what follows, we concentrate on  $\mc{RD}_\mathrm{CEO}^1$. The computation of $\mc{RD}_\mathrm{CEO}^2$ follows similarly, and is omitted for brevity. Next theorem provides a parametrization of the boundary tuples of the region  $\mc{RD}_\mathrm{CEO}^1$ in terms, each of them, of an optimization problem over the pmfs $\dv P := \{ P_{U_1|Y_1}, P_{U_2|Y_2} \}$.

\begin{theorem}~\label{theorem-CEO-parametrization}
For each $\dv s := [s_1, s_2]$, $s_1 > 0$, $s_2 > 0$, define a tuple $(R_{1,\dv s}, R_{2,\dv s} ,D_{\dv s})$ parametrically given by 
\begin{align}
&D_{\dv s} = -s_1 R_{1,\dv s}-s_2R_{2, \dv s}+ \min_{\dv P}F_{\dv s}(\dv P) \;, \label{equation-CEO-parametrization-D}\\  
&R_{1, \dv s} = I(Y_1;U_1^\star|U_2^\star,Y_0) \;, \quad\quad R_{2, \dv s} = I(Y_2;U_2^\star|Y_0) \;, \label{equation-CEO-parametrization-R}
\end{align}
where $F_{\dv s}(\dv P)$ is given as follows 
\begin{equation}~\label{equation-CEO-objective1}
F_{\dv s}(\dv P) := H(X|U_1,U_2,Y_0) + s_1 I(Y_1;U_1|U_2,Y_0) + s_2 I(Y_2;U_2|Y_0) \;,
\end{equation}
and; $\dv P^\star$ are the conditional pmfs yielding the minimum in \eqref{equation-CEO-parametrization-D} and $U_1^\star, U_2^\star$ are the auxiliary variables induced by $\dv P^\star$. Then, we have: 
\begin{enumerate}
\item Each value of  $\dv s$ leads to a tuple $(R_{1,\dv s}, R_{2,\dv s}, D_{\dv s})$ on the distortion-rate curve $D_{\dv s} = D_\mathrm{CEO}^1 (R_{1,\dv s}, R_{2,\dv s})$.
\item For every point on the distortion-rate curve, there is an $\dv s$ for which \eqref{equation-CEO-parametrization-D} and \eqref{equation-CEO-parametrization-R} hold.
\end{enumerate}
\end{theorem}

\begin{proof}
Suppose that $\dv P^\star$ yields the minimum in \eqref{equation-CEO-parametrization-D}. For this $\dv P$, we have $ I(Y_1;U_1|U_2,Y_0) = R_{1, \dv s}$ and $I(Y_2;U_2|Y_0) = R_{2, \dv s}$. Then, we have 
\begin{align}
D_{\dv s} &= -s_1 R_{1,\dv s} - s_2 R_{2, \dv s} + F_{\dv s}(\dv P^\star) \nonumber\\
& = -s_1 R_{1,\dv s} - s_2 R_{2, \dv s} + [ H(X|U_1^\star,U_2^\star,Y_0) + s_1 R_{1,\dv s} + s_2 R_{2, \dv s}] \nonumber\\
&= H(X|U_1^\star,U_2^\star,Y_0) \geq D_\mathrm{CEO}^1(R_{1,\dv s}, R_{2,\dv s}) \;. \label{equation-CEO-parametrization-proof-1}
\end{align}
Conversely, if $\dv P^\star$ is the solution to the minimization in \eqref{equation-CEO-distortion-rate-function}, then $I(Y_1;U_1^\star|U_2^\star,Y_0)\leq R_{1}$ and $I(Y_2;U_2^\star|Y_0)\leq R_{2}$ and for any $\dv s$, 
\begin{align*}
D_\mathrm{CEO}^{1}(R_1,R_2) &= H(X|U_1^\star,U_2^\star,Y_0) \\
&\geq H(X|U_1^\star,U_2^\star,Y_0) + s_1 ( I(Y_1;U_1^\star|U_2^\star,Y_0) - R_1 ) + s_2 ( I(Y_2;U_2^\star|Y_0) - R_2 )\\
&= D_{\dv s} + s_1 ( R_{1, \dv s} - R_{1} ) + s_2 ( R_{2, \dv s} - R_{2} ) \;.
\end{align*}
Given $\dv s$, and hence $(R_{1,\dv s},R_{2,\dv s}, D_{\dv s})$, letting $(R_1,R_2) = (R_{1,\dv s}, R_{2,\dv s})$ yields $D_\mathrm{CEO}^{1}(R_{1,\dv s}, R_{2,\dv s}) \geq D_{\dv s}$, which proves, together with \eqref{equation-CEO-parametrization-proof-1},  statement 1) and 2).
\end{proof}
                                                                               
Next, we show that it is sufficient to run the algorithm for $s_1 \in (0,1]$.

\begin{lemma}~\label{lemma-range}
The range of the parameter $s_1$ can be restricted to $(0,1]$.
\end{lemma}

\begin{proof}
Let $F^\star = \min_{\dv P} F_{\dv s }(\dv P)$. If we set $U_1 = \emptyset$, then we have the relation $F^\star \leq H(X|U_2,Y_0) + s_2 I(Y_2;U_2|Y_0)$. For $s_1 > 1$, we have
\begin{equation*}
F_{\dv s }(\dv P) 
\stackrel{(a)}{\geq} (1-s_1) H(X|U_1,U_2,Y_0) + s_1 H(X|U_2,Y_0) + s_2 I(Y_2;U_2|Y_0)
\stackrel{(b)}{\geq} H(X|U_2,Y_0) + s_2 I(Y_2;U_2|Y_0) \;,
\end{equation*}
where $(a)$ follows since mutual information is always positive, i.e., $I(Y_1;U_1|X,Y_0) \geq 0$; $(b)$ holds since conditioning reduces entropy and $1 - s_1 <0$. Then $F^\star = H(X|U_2,Y_0) + s_2 I(Y_2;U_2|Y_0)$ for $s_1 > 1$. Hence we can restrict the range of $s_1$ to $s_1 \in (0,1]$. 
\end{proof}

\subsubsection{Computation of $\mc{RD}_\mathrm{CEO}^1$}~\label{section-BA-computation-region1}

In this section, we derive an algorithm to solve~\eqref{equation-CEO-parametrization-D} for a given parameter value $\dv s$. To that end, we define a variational bound on $F_{\dv s}(\dv P)$, and optimize it instead of~\eqref{equation-CEO-parametrization-D}. Let $\dv Q$ be a set of some auxiliary pmfs defined as
\begin{equation}
\dv Q := \{ Q_{U_1} \;, \; Q_{U_2} \;, \; Q_{X|U_1,U_2,Y_0} \;, \; Q_{X|U_1,Y_0} \;, \; Q_{X|U_2,Y_0} \;, \; Q_{Y_0|U_1} \;, \; Q_{Y_0|U_2} \} \;.
\end{equation}
	
In the following we define the variational cost function $F_{\dv s }(\dv P, \dv Q)$ 
\begin{align}~\label{equation-CEO-objective2}
F_{\dv s}(\dv P,\dv Q) :=& - s_1 H(X|Y_0) - (s_1 + s_2) H(Y_0) \nonumber\\
&+ \E_{P_{X,Y_0,Y_1,Y_2}} \Big[(1-s_1) \E_{P_{U_1| Y_1}}\E_{P_{U_2| Y_2}}[- \log Q_{X|U_1,U_2,Y_0}] + s_1 \E_{P_{U_1| Y_1}}[- \log Q_{X|U_1,Y_0}] \nonumber\\
&\hspace{6.5em}  + s_1 \E_{P_{U_2| Y_2}}[- \log Q_{X|U_2,Y_0}]  + s_1 D_\mathrm{KL}(P_{U_1|Y_1}\|Q_{U_1}) + s_2 D_\mathrm{KL}(P_{U_2|Y_2}\|Q_{U_2}) \nonumber\\   
&\hspace{6.5em} + s_1 \E_{P_{U_1| Y_1}}[- \log Q_{Y_0|U_1}] + s_2 \E_{P_{U_2| Y_2}}[- \log Q_{Y_0|U_2}] \Big] \;.
\end{align}
	
The following lemma states that $F_s(\dv P, \dv Q)$ is an upper bound on $F_s(\dv P)$ for all distributions $\dv Q$. 
	
\begin{lemma}~\label{lemma-CEO-fixed-P}
For fixed $\dv P$, we have
\begin{equation*}
F_s(\dv P, \dv Q) \geq F_s(\dv P) \;, \quad \text{for all } \dv Q \;.
\end{equation*} 
In addition, there exists a $\dv Q$ that achieves the minimum $\argmin_{\dv Q} F_{\dv s}(\dv P, \dv Q) = F_{\dv s}(\dv P)$, given by
\begin{equation}~\label{equation-CEO-optimal-Q}
\begin{aligned}
&Q_{U_k} = P_{U_k} \;, \quad Q_{X|U_k,Y_0} = P_{X|U_k,Y_0} \;, \quad Q_{Y_0|U_k} = P_{Y_0|U_k} \;, \quad\text{for }\: k=1, 2 \;,\\
&Q_{X|U_1,U_2,Y_0} = P_{X|U_1,U_2,Y_0} \;. 
\end{aligned}
\end{equation}
\end{lemma}
	
\begin{proof}
The proof of Lemma~\ref{lemma-CEO-fixed-P} is given in Appendix~\ref{proof-lemma-CEO-fixed-P}.	
\end{proof}

Using the lemma above, the minimization in~\eqref{equation-CEO-parametrization-D} can be written in terms of the variational cost function as follows 
\begin{equation}~\label{equation-BA-double-minimization}
\argmin_{\dv P} F_{\dv s}(\dv P) = \argmin_{\dv P} \argmin_{\dv Q} F_{\dv s}(\dv P, \dv Q) \;. 
\end{equation}

Motivated by the BA algorithm \cite{B72,A72}, we propose an alternate optimization procedure over the set of pmfs $\dv P$ and $\dv Q$ as stated in Algorithm~\ref{algo-CEO}. The main idea is that at iteration $t$, for fixed $\dv P^{(t-1)}$ the optimal $\dv Q^{(t)}$ minimizing $F_{\dv s}(\dv P, \dv Q)$ can be found analytically; next, for given $\dv Q^{(t)}$ the optimal $\dv P^{(t)}$ that minimizes $F_{\dv s}(\dv P, \dv Q)$ has also a closed form. So, starting with a random initialization $\dv P^{(0)}$, the algorithm iterates over distributions $\dv Q$ and $\dv P$ minimizing $F_{\dv s}(\dv P, \dv Q)$ until the convergence, as stated below   
\begin{equation*}
\dv P^{(0)} \rightarrow \dv Q^{(1)} \rightarrow \dv P^{(1)} \rightarrow \ldots \rightarrow \dv P^{(t-1)} \rightarrow \dv Q^{(t)} \rightarrow \ldots \rightarrow \dv P^\star \rightarrow \dv Q^\star \;.
\end{equation*}

	\setlength{\textfloatsep}{10pt} 
	\begin{algorithm}[!]
		\caption{BA-type algorithm to compute $\mc{RD}_\mathrm{CEO}^1$}
		\label{algo-CEO}
		{\fontsize{10}{10}\selectfont
			\begin{algorithmic}[1]
				\smallskip
				\Inputs{pmf $P_{X,Y_0,Y_1,Y_2}$, parameters $1 \geq s_1 > 0$, $s_2 > 0$.}
				\Outputs{Optimal $P_{U_1|Y_1}^\star$, $P_{U_2|Y_2}^\star$; triple $(R_{1,\dv s}, R_{2,\dv s}, D_{\dv s})$.}
				\Initialize{Set $t=0$. Set $\dv P^{(0)}$ randomly.}
				\Repeat 
				\State\label{algo-CEO-step7}Update the following pmfs for $k = 1, 2$  
				\begin{equation*}
				\footnotesize
				\begin{aligned}
				p^{(t+1)}(u_k) &= \sum\nolimits_{y_k} p^{(t)}(u_k|y_k) p(y_k), \\
				p^{(t+1)}(u_k|y_0) &= \sum\nolimits_{y_k} p^{(t)}(u_k|y_k) p(y_k|y_0), \\
				p^{(t+1)}(u_k|x,y_0) &= \sum\nolimits_{y_k} p^{(t)}(u_k|y_k) p(y_k|x,y_0), \\ 
				p^{(t+1)}(x|u_1,u_2,y_0) &= \frac{p^{(t+1)}(u_1|x,y_0) p^{(t+1)}(u_2|x,y_0) p(x,y_0)}{\sum_{x} p^{(t+1)}(u_1|x,y_0) p^{(t+1)}(u_2|x,y_0) p(x,y_0)}.  
				\end{aligned}
				\end{equation*}	
				\State\label{algo-CEO-step8}Update $\dv Q^{(t+1)}$ by using \eqref{equation-CEO-optimal-Q}.
				\State\label{algo-CEO-step6}Update $\dv P^{(t+1)}$ by using \eqref{equation-CEO-optimal-P}.
				\State $t \leftarrow t+1$.
				\Until{convergence.}
		\end{algorithmic} }
	\end{algorithm}

\noindent
At each iteration, the optimal values of $\dv P$ and $\dv Q$ are found by solving a convex optimization problems. We have the following lemma.

\begin{lemma}~\label{lemma-CEO-convex}
$F_{\dv s}(\dv P, \dv Q)$ is convex in $\dv P$ and convex in $\dv Q$.
\end{lemma}

\begin{proof}
The proof of Lemma~\ref{lemma-CEO-convex} follows from the log-sum inequality.   
\end{proof}	

For fixed $\dv P^{(t-1)}$, the optimal $\dv Q^{(t)}$ minimizing the variational bound in~\eqref{equation-CEO-objective2} can be found from Lemma~\ref{lemma-CEO-fixed-P} and given by~\eqref{equation-CEO-optimal-Q}. For fixed $\dv Q^{(t)}$, the optimal $\dv P^{(t)}$ minimizing~\eqref{equation-CEO-objective2} can be found by using the next lemma.

\begin{lemma}~\label{lemma-CEO-fixed-Q}
For fixed $\dv Q$, there exists a $\dv P$ that achieves the minimum $\argmin_{\dv P} F_{\dv s}(\dv P, \dv Q)$, where $P_{U_k|Y_k}$ is given by
\begin{equation}~\label{equation-CEO-optimal-P}		
p(u_k|y_k) = q(u_k) \frac {\exp[-\psi_k(u_k,y_k)]} {\sum_{u_k} q(u_k)\exp[-\psi_k(u_k,y_k)]} \;, \quad\text{for}\:\: k=1,2 \;,
\end{equation}
where $\psi_k(u_k,y_k)$, $k=1,2$, are defined as follows
\begin{equation}~\label{equation-CEO-psi} 
\psi_k(u_k,y_k) \!\! := \!\! \frac{1-s_1}{s_k} \E_{U_\kbar,Y_0|y_k} \! D(P_{X|y_k,U_\kbar,Y_0}\|Q_{X|u_k,U_\kbar,Y_0})  
+ \frac{s_1}{s_k} \! \E_{Y_0|y_k} \! D(P_{X|y_k,Y_0}\|Q_{X|u_k,Y_0})  + \! D(P_{Y_0|y_k}\|Q_{Y_0|u_k}) .
\end{equation}
\end{lemma}

\begin{proof}
The proof of Lemma~\ref{lemma-CEO-fixed-Q} is given in Appendix~\ref{proof-lemma-CEO-fixed-Q}.
\end{proof}

At each iteration of Algorithm~\ref{algo-CEO}, $F_{\dv s}(\dv P^{(t)}, \dv Q^{(t)})$ decreases until eventually it converges. However, since $F_{\dv s}(\dv P, \dv Q)$ is convex in each argument but not necessarily jointly convex, Algorithm~\ref{algo-CEO} does not necessarily converge to the global optimum. In particular, next proposition shows that Algorithm~\ref{algo-CEO} converges to a stationary solution of the minimization in~\eqref{equation-CEO-parametrization-D}. 

\begin{proposition}
Every limit point of $\dv P^{(t)}$ generated by Algorithm~\ref{algo-CEO} converges to a stationary solution of \eqref{equation-CEO-parametrization-D}.
\end{proposition}

\begin{proof}
Algorithm~\ref{algo-CEO} falls into the class of so-called ``Successive Upper-bound Minimization" (SUM) algorithms~\cite{RHL13}, in which $F_{\dv s}(\dv P, \dv Q)$ acts as a globally tight upper bound on $F_{\dv s}(\dv P)$. Let $\dv Q^\star(\dv P):= \arg\min_{\dv Q} F_{\dv s}(\dv P, \dv Q)$. From Lemma~\ref{lemma-CEO-fixed-P},  $F_{\dv s}(\dv P, \dv Q^\star(\dv P'))\geq F_{\dv s}(\dv P, \dv Q^\star(\dv P)) = F_{\dv s}(\dv P) $ for $\dv P'\neq \dv P$.  It follows that $F_{\dv s}(\dv P)$ and $F_{\dv s}(\dv P, \dv Q^\star(\dv P'))$ satisfy \cite[Proposition 1]{RHL13}  and thus $F_{\dv s}(\dv P, \dv Q^\star(\dv P'))$  satisfies (A1)--(A4) in \cite{RHL13}. Convergence to a stationary point of~\eqref{equation-CEO-parametrization-D} follows from \cite[Theorem 1]{RHL13}.
\end{proof}

\begin{remark}
Algorithm~\ref{algo-CEO} generates a sequence that is non-increasing. Since this sequence is lower bounded, convergence to a stationary point is guaranteed. This \textit{per-se}, however, does not necessarily imply that such a point is a stationary solution of the original problem described by~\eqref{equation-CEO-parametrization-D}. Instead, this is guaranteed here by showing that the Algorithm~\ref{algo-CEO} is of SUM-type with the function $F_{\dv s}(\dv P, \dv Q)$ satisfying the necessary conditions~\cite[(A1)--(A4)]{RHL13}. \vspace{0.2em}\qedblack \vspace{-0.2em}              	
\end{remark}

\setlength{\textfloatsep}{10pt} 
\begin{algorithm}[]
	\caption{BA-type algorithm for the Gaussian vector CEO}
	\label{algo-CEO-Gauss}
	{\fontsize{10}{10}\selectfont
		\begin{algorithmic}[1]
			\smallskip
			\Inputs{Covariance ${\dv \Sigma}_{(\dv x, \dv y_0, \dv y_1, \dv y_2)}$, parameters$1 \geq s_1 > 0$, $s_2 > 0$.}
			\Outputs{Optimal pairs $(\dv A_k^\star,\dv \Sigma_{\dv z_k^\star} )$, $k=1,2$.}
			\Initialize{Set $t=0$. Set randomly $\dv A_k^{0}$ and $\dv \Sigma_{\dv z_k^{0}} \succeq 0$ for $k=1,2$.}
			\Repeat 
			\State For  $k=1, 2$, update the following
			\begin{align*}
			\dv\Sigma_{\dv u_k^t} &= \dv A_k^t \dv\Sigma_{\dv y_k} {\dv A_k^t}^\dagger + \dv\Sigma_{\dv z_k^t}, \\
			\dv\Sigma_{\dv u_k^t|(\dv x, \dv y)} &= \dv A_k^t \dv\Sigma_k {\dv A_k^t}^\dagger + \dv\Sigma_{\dv z_k^t},
			\end{align*}                                                                     
			and update $\dv\Sigma_{\dv u_k^t|(\dv u_\kbar^t, \dv y)}$, $\dv\Sigma_{\dv u_2^t|\dv y}$ and $\dv\Sigma_{\dv y_k^t|(\dv u_\kbar^t, \dv y)}$ from their definitions by using the following 
			\begin{align*}
			\dv\Sigma_{\dv u_1^t, \dv u_2^t} &= \dv A_1^t \dv H_1 \dv\Sigma_{\dv x} \dv H_2^\dagger \dv A_2^{t^\dagger}, \\
			\dv\Sigma_{\dv u_k^t, \dv y} &= \dv A_k^t \dv H_k \dv\Sigma_{\dv x} \dv H_0^\dagger, \\
			\dv\Sigma_{\dv y_k, \dv u_\kbar^t} &= \dv H_k \dv\Sigma_{\dv x} \dv H_\kbar^\dagger {\dv A_\kbar^t}^\dagger.         
			\end{align*}
			\State Compute $\dv \Sigma_{\dv z_k^{t+1}}$ as in \eqref{equation-CEO-Gauss-Sigma-update} for $k=1,2$.
			\State Compute $\dv A_k^{t+1}$ as \eqref{equation-CEO-Gauss-A-update} for $k=1,2$.
			\State $t \leftarrow t+1$.
			\Until{convergence.}
	\end{algorithmic} }
\end{algorithm}

\subsection{Vector Gaussian Case}\label{section-BA-algorithm-Gaussian}

Computing the rate-distortion region $\mc{RD}_{\mathrm{VG}\text{-}\mathrm{CEO}}^\star$ of the vector Gaussian CEO problem as given by Theorem~\ref{theorem-Gauss-RD-CEO} is a convex optimization problem on $\{\dv\Omega_k\}_{k=1}^K$ which can be solved using, e.g., the popular generic optimization tool CVX~\cite{cvx}. Alternatively, the region can be computed using an extension of Algorithm~\ref{algo-CEO} to memoryless Gaussian sources as given in the rest of this section. 

\noindent For discrete sources with (small) alphabets, the updating rules of $\dv Q^{(t+1)}$ and $\dv P^{(t+1)}$ of Algorithm~\ref{algo-CEO} are relatively easy computationally. However, they become computationally unfeasible for continuous alphabet sources. Here, we leverage on the optimality of Gaussian test channels as shown by Theorem~\ref{theorem-Gauss-RD-CEO} to restrict the optimization of $\dv P$ to Gaussian distributions, which allows to reduce the search of update rules to those of the associated parameters, namely covariance matrices. In particular, we show that if $P_{\dv U_k|\dv Y_k}^{(t)}$, $k= 1, 2$, is Gaussian and such that
\begin{equation}~\label{eq:testChan}
\dv U_k^t = \dv A_k^t\dv Y_k +\dv Z_k^t \;,
\end{equation}
where $\dv Z_k^t\sim\mc{CN}(\dv 0,\dv\Sigma_{\dv z_k^t})$ then $P_{\dv U_k|\dv Y_k}^{(t+1)}$ too is Gaussian, with
\begin{equation}
\dv U_k^{t+1} = \dv A_k^{t+1} \dv Y_k + \dv Z_k^{t+1} \;,
\end{equation}
where $\dv Z_k^{t+1} \sim\mc{CN}(\dv 0, \dv\Sigma_{\dv z_k^{t+1}})$ and the parameters $\dv A_k^{t+1}$ and $\dv\Sigma_{\dv z_k^{t+1}}$ are given by
\begin{subequations}~\label{equation-CEO-Gauss-Sigma-and-A-update}
\begin{align}
\label{equation-CEO-Gauss-Sigma-update} 
\dv\Sigma_{\dv z_k^{t+1}} =& \: \left( \frac{1}{s_k} \dv\Sigma_{\dv u_k^t|(\dv x, \dv y_0)}^{-1} - \frac{1-s_1}{s_k} \dv\Sigma_{\dv u_k^t|(\dv u_\kbar^t,\dv y_0)}^{-1} + \frac{s_k-s_1}{s_k} \dv\Sigma_{\dv u_k^t|\dv y_0}^{-1} \right)^{-1} \\[0.5em]
\dv A_k^{t+1} =& \:\: \dv\Sigma_{\dv z_k^{t+1}} \left( \frac{1}{s_k} \dv\Sigma_{\dv u_k^t|(\dv x,\dv y_0)}^{-1} \dv A_k^t  (\dv I - \dv\Sigma_{\dv y_k|(\dv x,\dv y_0)} \dv\Sigma_{\dv y_k}^{-1} ) \right) \nonumber\\
&\:  - \dv\Sigma_{\dv z_k^{t+1}} \left( \frac{1-s_1}{s_k} \dv\Sigma_{\dv u_k^t|(\dv u_\kbar^t,\dv y_0)}^{-1} \dv A_k^t (\dv I - \dv\Sigma_{\dv y_k|(\dv u_\kbar^t,\dv y_0)} \dv\Sigma_{\dv y_k}^{-1}) - \frac{s_k-s_1}{s_k} \dv\Sigma_{\dv u_k^t|\dv y_0}^{-1} \dv A_k^t (\dv I - \dv\Sigma_{\dv y_k|\dv y_0} \dv\Sigma_{\dv y_k}^{-1}) \right). \label{equation-CEO-Gauss-A-update}
\end{align}
\end{subequations}
The updating steps are provided in Algorithm~\ref{algo-CEO-Gauss}. The proof of~\eqref{equation-CEO-Gauss-Sigma-and-A-update} can be found in Appendix~\ref{proof-CEO-Gauss-algorithm}.

\vspace{0.5em}
\subsection{Numerical Examples}\label{section-BA-numerical-examples}  

In this section, we discuss two examples, a binary CEO example and a vector Gaussian CEO example.  

\begin{figure*}[!ht]
\begin{center}
\subfloat[]{ 
\resizebox{.49\linewidth}{!}{\includegraphics{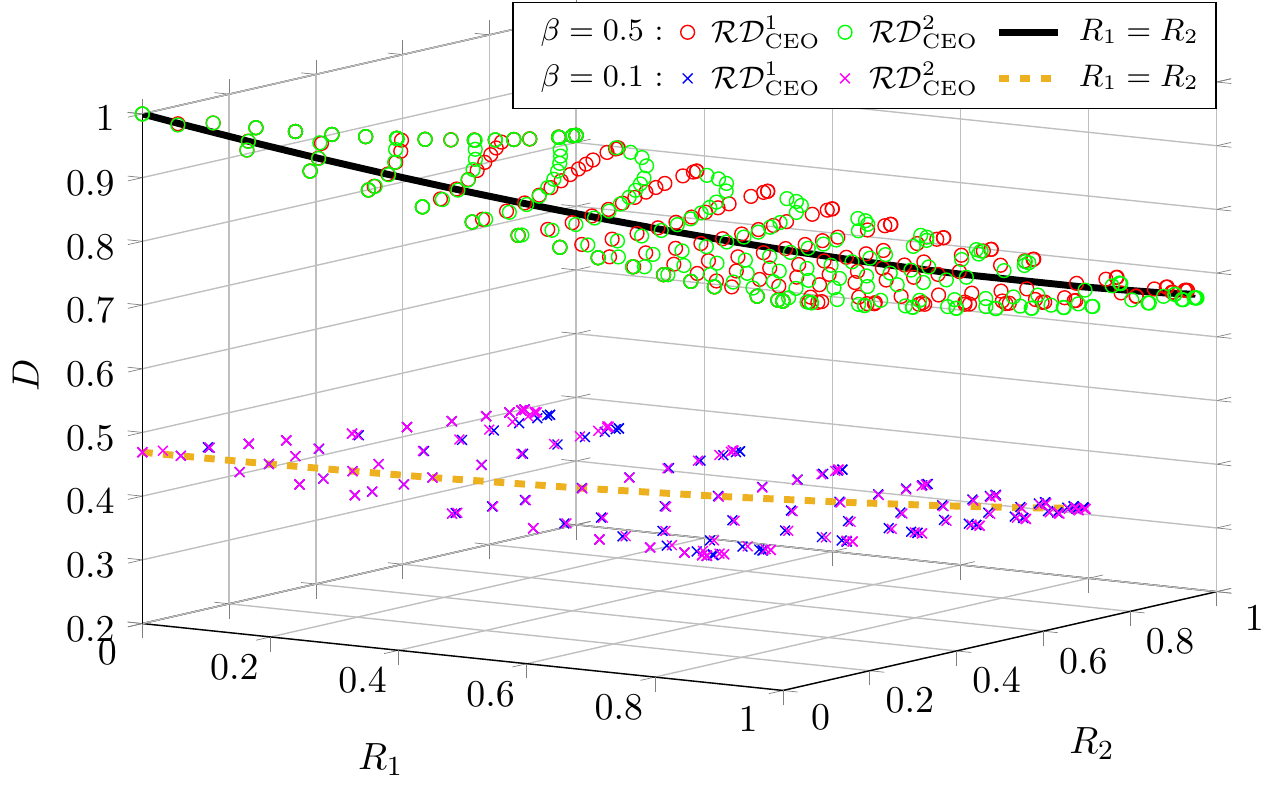}}
\label{fig-3D-DM} 
} 
\subfloat[]{
\resizebox{.49\linewidth}{!}{\includegraphics{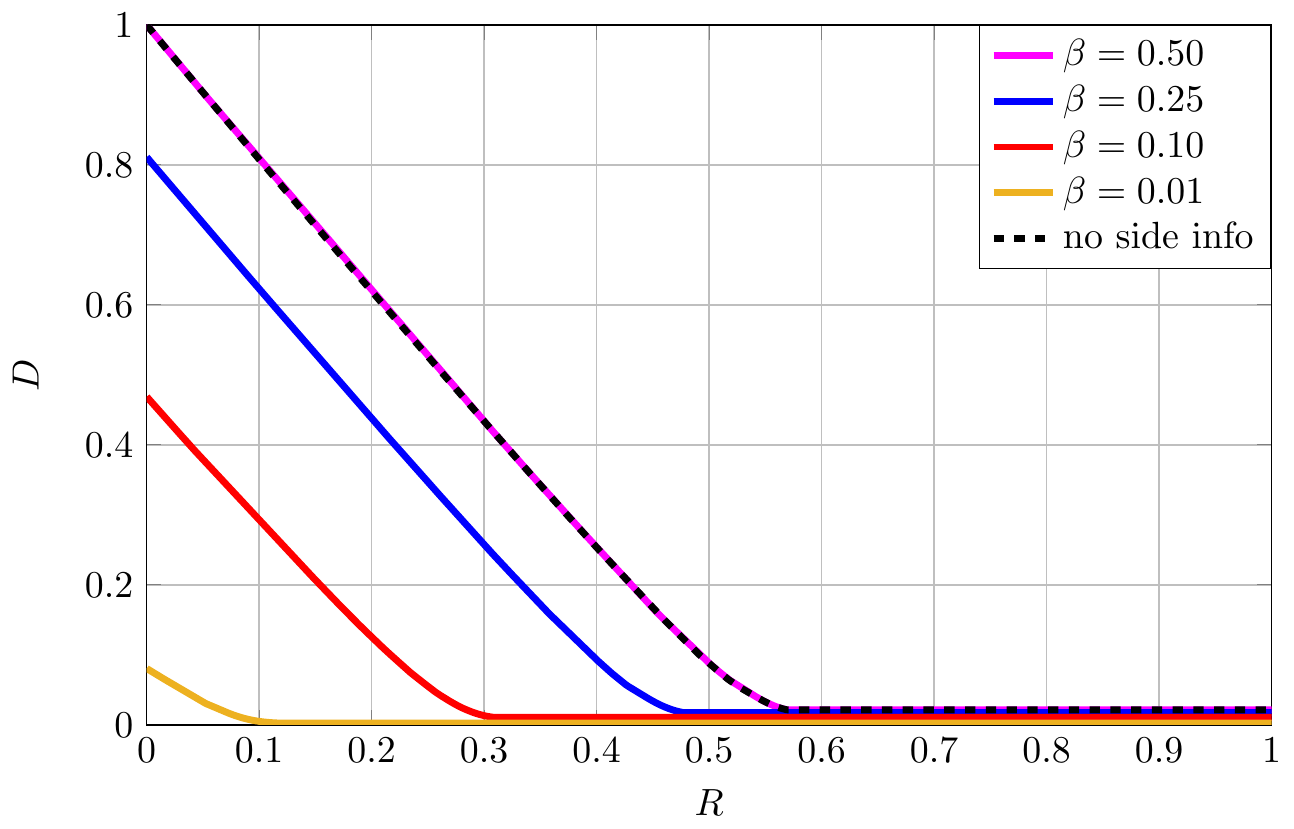}}
\label{fig-2D-DM}
}
\caption{Rate-distortion region of the binary CEO network of Example~\ref{binary-ceo-example}, computed using Algorithm~\ref{algo-CEO}. (a): set of $(R_1,R_2,D)$ triples such $(R_1,R_2,D) \in \mc{RD}_\mathrm{CEO}^1 \cup \mc{RD}_\mathrm{CEO}^2$, for $\alpha_1=\alpha_2=0.25$ and $\beta \in \{0.1, 0.25\}$. (b): set of $(R,D)$ pairs such $(R,R,D) \in \mc{RD}_\mathrm{CEO}^1 \cup \mc{RD}_\mathrm{CEO}^2$, for $\alpha_1=\alpha_2=0.01$ and $\beta \in \{0.01, 0.1, 0.25, 0.5\}$.}
\label{fig-rate-distortion-binary-ceo-example} 
\end{center}
\end{figure*}

\vspace{-0.7em}
\begin{example}~\label{binary-ceo-example}
Consider the following binary CEO problem. A memoryless binary source $X$, modeled as a Bernoulli-$(1/2)$ random variable, i.e., $X \sim \operatorname{Bern}(1/2)$, is observed remotely at two agents who communicate with a central unit decoder over error-free rate-limited links of capacity $R_1$ and $R_2$, respectively. The decoder wants to estimate the remote source $X$ to within some average fidelity level $D$, where the distortion is measured under the logarithmic loss criterion. The noisy observation $Y_1$ at Agent $1$ is modeled as the output of a binary symmetric channel (BSC) with crossover probability $\alpha_1 \in [0,1]$, whose input is $X$, i.e., $Y_1 = X \oplus S_1$ with $S_1 \sim \operatorname{Bern}(\alpha_1)$. Similarly, the noisy observation $Y_2$ at Agent $2$ is modeled as the output of a $\text{BSC}(\alpha_2)$ channel, $\alpha_2 \in [0,1]$, whose has input $X$, i.e., $Y_2 = X \oplus S_2$ with $S_2 \sim \operatorname{Bern}(\alpha_2)$. Also, the central unit decoder observes its own side information $Y_0$ in the form of the output of a $\text{BSC}(\beta)$ channel, $\beta \in [0,1]$, whose input is $X$, i.e., $Y_0= X \oplus S_0$ with $S_0 \sim \operatorname{Bern}(\beta)$. It is assumed that the binary noises $S_0$, $S_1$ and $S_2$ are independent between them and with the remote source $X$. 

\noindent We use Algorithm~\ref{algo-CEO} to numerically approximate\footnote{We remind the reader that, as already mentioned, Algorithm~\ref{algo-CEO} only converges to stationary points of the rate-distortion region.} the set of $(R_1,R_2,D)$ triples such that $(R_1,R_2,D)$ is in the union of the achievable regions $\mc{RD}_\mathrm{CEO}^1$ and $\mc{RD}_\mathrm{CEO}^2$ as given by~\eqref{constituent-rate-distortion-regions}. The regions are depicted in Figure~\ref{fig-3D-DM} for the values $\alpha_1=\alpha_2=0.25$ and $\beta \in \{0.1, 0.25\}$.  Note that for both values of $\beta$, an approximation of the rate-distortion region $\mc{RD}_\mathrm{CEO}$ is easily found as the convex hull of the union of the shown two regions. For simplicity, Figure~\ref{fig-2D-DM} shows achievable rate-distortion pairs $(R,D)$ in the case in which the rates of the two encoders are constrained to be at most $R$ bits per channel use each, i.e., $R_1=R_2=R$,  higher quality agents' observations $(Y_1,Y_2)$ corresponding to $\alpha_1=\alpha_2=0.01$ and $\beta \in \{0.01, 0.1, 0.25, 0.5\}$. In this figure, observe that, as expected, smaller values of $\beta$ correspond to higher quality estimate side information $Y_0$ at the decoder; and lead to smaller distortion values for given rate $R$. The choice $\beta = 0.5$ corresponds to the case of no or independent side information at decoder; and it is easy to check that the associated $(R,D)$ curve coincides with the one obtained through exhaustive search in~\cite[Figure 3]{CW14}. \qedblack
\end{example}

\vspace{-0.3em}
\begin{figure*}[!ht]
\begin{center}
\subfloat[]{ 
\resizebox{.49\linewidth}{!}{\includegraphics{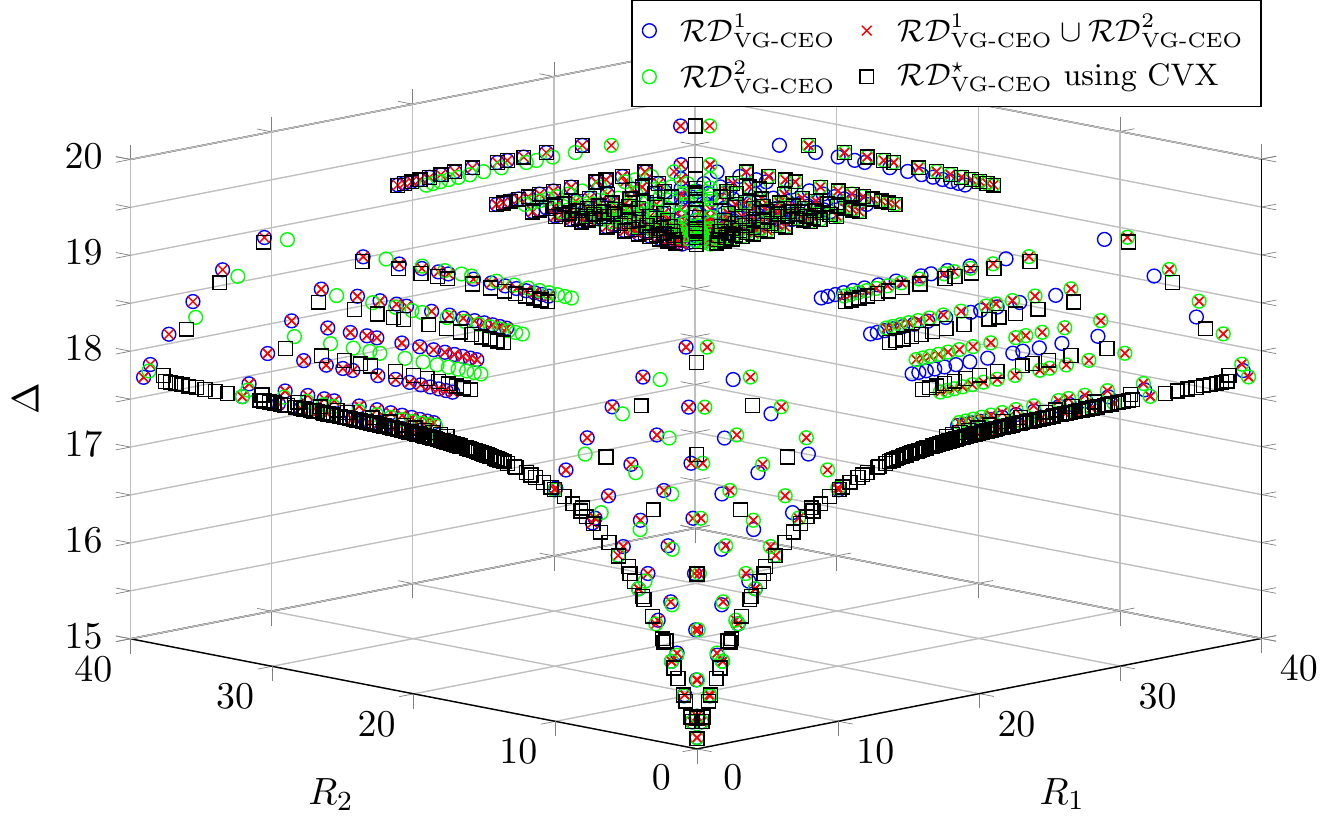}}
\label{fig-3D-Gauss} 
} 
\subfloat[]{
\resizebox{.49\linewidth}{!}{\includegraphics{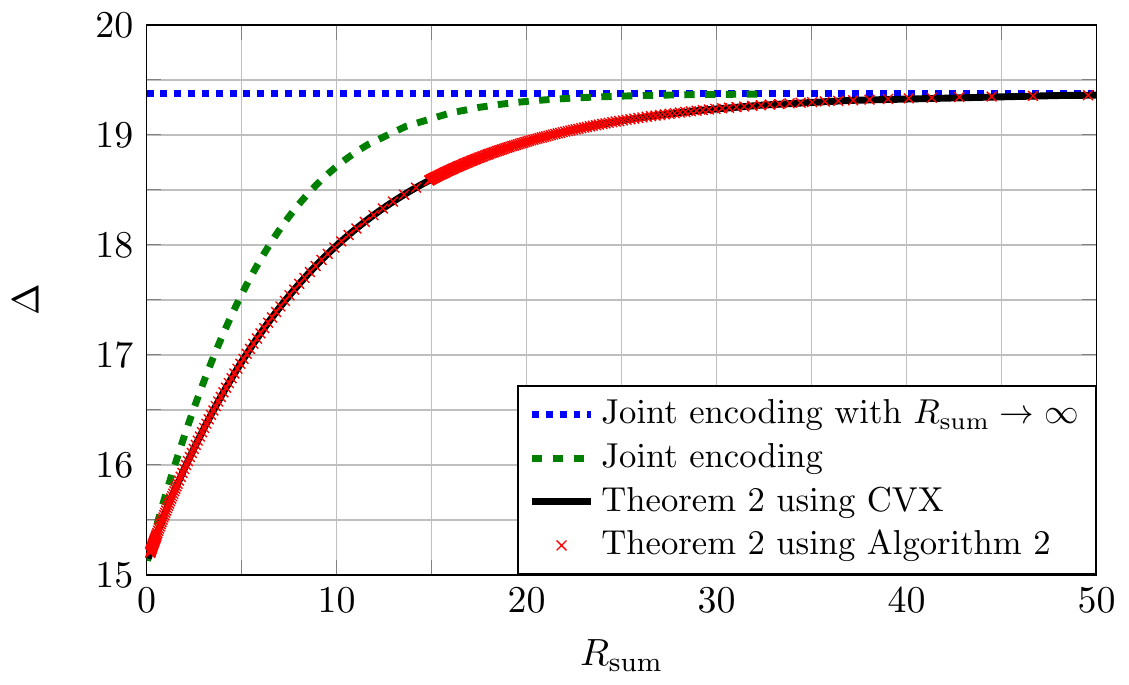}}
\label{fig-2D-Gauss}
}
\caption{Rate-information region of the vector Gaussian CEO network of Example~\ref{vector-Gaussian-ceo-example}. Numerical values are $n_x =3$ and $n_0 = n_1 = n_2 =4$. (a): set of $(R_1,R_2,\Delta)$ triples such $(R_1,R_2,h(\dv X) - \Delta) \in \mc{RD}_{\mathrm{VG}\text{-}\mathrm{CEO}}^1 \cup \mc{RD}_{\mathrm{VG}\text{-}\mathrm{CEO}}^2$, computed using Algorithm~\ref{algo-CEO-Gauss}. (b): set of $(R_\mathrm{sum}, \Delta)$ pairs such $R_\mathrm{sum}=R_1+R_2$ for some $(R_1,R_2)$ for which $(R_1,R_2, h(\dv X)- \Delta) \in \mc{RD}_{\mathrm{VG}\text{-}\mathrm{CEO}}^1 \cup \mc{RD}_{\mathrm{VG}\text{-}\mathrm{CEO}}^2$.}
\label{fig-rate-distortion-vector-Gaussian-ceo-example} 
\end{center}
\end{figure*}

\vspace{-0.7em}
\begin{example}~\label{vector-Gaussian-ceo-example}
Consider an instance of the memoryless vector Gaussian CEO problem as described by~\eqref{mimo-gaussian-model} and~\eqref{mimo-gaussian-model-2} obtained by setting $K=2$, $n_x=3$ and $n_0 = n_1 = n_2 =4$.  We use Algorithm~\ref{algo-CEO-Gauss} to numerically approximate the set of $(R_1,R_2,\Delta)$ triples such $(R_1,R_2,h(\dv X) - \Delta)$ is in the union of the achievable regions $\mc{RD}_{\mathrm{VG}\text{-}\mathrm{CEO}}^1$ and $\mc{RD}_{\mathrm{VG}\text{-}\mathrm{CEO}}^2$. The result is depicted in Figure~\ref{fig-3D-Gauss}. The figure also shows the set of $(R_1,R_2,\Delta)$ triples such $(R_1,R_2,h(\dv X) - \Delta)$ lies in the region given by Theorem~\ref{theorem-Gauss-RD-CEO} evaluated for the example at hand. Figure~\ref{fig-2D-Gauss} shows the set of $(R_\mathrm{sum}, \Delta)$ pairs such $R_\mathrm{sum}:=R_1+R_2$ for some $(R_1,R_2)$ for which $(R_1,R_2, h(\dv X)- \Delta)$ is in the union of  $\mc{RD}_{\mathrm{VG}\text{-}\mathrm{CEO}}^1$ and $\mc{RD}_{\mathrm{VG}\text{-}\mathrm{CEO}}^2$. The region is computed using two different approaches: i) using Algorithm~\ref{algo-CEO-Gauss} and ii) by directly evaluating the region obtained from Theorem~\ref{theorem-Gauss-RD-CEO} using the CVX optimization tool to find the maximizing covariances matrices $(\dv\Omega_1, \dv\Omega_2)$ (note that this problem is convex and so CVX finds the optimal solution). It is worth-noting that Algorithm~\ref{algo-CEO-Gauss} converges to the optimal solution for the studied vector Gaussian CEO example, as is visible from the figure. For comparisons reasons, the figure also shows the performance of \textit{centralized} or \textit{joint} encoding, i.e., the case both agents observe both $\dv Y_1$ and $\dv Y_2$, 
\begin{equation}~\label{information-rate-region-centralized-encoding}
\Delta(R_{\mathrm{sum}}) = \max_{P_{U|\dv Y_1,\dv Y_2} \: : \: I(U;\dv Y_1,\dv Y_2| \dv Y_0) \leq R_\mathrm{sum}} I(U,\dv Y_0; \dv X) \;.
\end{equation}  
Finally, we note that the information/sum-rate function~\eqref{information-rate-region-centralized-encoding} can be seen an extension of Chechik \textit{et al.} Gaussian Information Bottleneck~\cite{CGTW05} to the case of side information $\dv Y_0$ at the decoder. Figure~\ref{fig-2D-Gauss} shows the loss in terms of information/sum-rate that is incurred by restricting the encoders to operate separately, i.e., distributed Information Bottleneck with side information at decoder. \qedblack
\end{example}

\appendices

\section{Proof of Theorem~\ref{theorem-continuous-RD1-CEO}}\label{proof-continous-RD1-CEO}

 For convenience, consider first the DM version of the CEO problem with decoder side information under logarithmic loss of Figure~\ref{fig-system-model-CEO}. It is assumed that for all $\mc S \subseteq \mc K := \{1,\ldots,K\}$, 
\begin{equation}~\label{eq:MKChain_pmf}
\dv Y_{\mc S} \mkv  (\dv X, \dv Y_0)  \mkv \dv Y_{{\mc S}^c} \;,
\end{equation}
forms a Markov chain in that order. The definitions for this model are similar to Definition~\ref{definition-encoder} and Definition~\ref{definition-rate-distortion-region} and are omitted for brevity. The rate-distortion region of this problem can be obtained readily by applying~\cite[Theorem 10]{CW14}, which provides the rate-distortion region of the model without side information at decoder, to the modified setting in which the remote source is $\tilde{\dv X}=(\dv X,\dv Y_0)$, another agent (agent $K+1$) observes $\dv Y_{K+1}=\dv Y_0$ and communicates at large rate $R_{K+1}=\infty$ with the CEO, which wishes to estimates $\tilde{\dv X}$ to within average logarithmic distortion $D$ and has no own side information stream\footnote{Note that for the modified CEO setting the agents' observations are conditionally independent given the remote source $\tilde{\dv X}$.}. More specifically, it is given by the union of the set of all non-negative tuples $(R_1,\ldots, R_K,D)$ that satisfy, for all subsets $\mc S \subseteq \mc K$,
\begin{equation}~\label{rate-distortion-region-DM-CEO-model}
D + \sum_{k \in \mc S} R_k \geq \sum_{k \in \mc S} I(\dv Y_k;U_k|\dv X,\dv Y_0,Q) + H(\dv X | U_{\mc S^c},\dv Y_0,Q) \;,
\end{equation}
for some joint measure of the form $P_{\dv Y_0,\dv Y_{\mc K},\dv X}(\dv y_0,\dv y_{\mc K},\dv x)P_Q(q) \prod_{k=1}^{K} P_{U_k|\dv Y_k,Q}(u_k|\dv y_k,q)$.

\noindent Also, let us define for this model the rate-information region $\mc{RI}_\mathrm{CEO}^\star$ as the closure of all rate-information tuples $(R_1,\ldots,R_K,\Delta)$ for which there exist a blocklength $n$, encoding functions $\{\phi^{(n)}_k\}^K_{k=1}$ and a decoding function $\psi^{(n)}$ such that  
\begin{align*}
R_k &\geq \frac{1}{n}\log M^{(n)}_k \;, \quad\text{for}\:\: k=1,\ldots,K \;, \\
\Delta &\leq \frac{1}{n} I(\dv X^n;\psi^{(n)}(\phi_1^{(n)}(\dv Y_1^n),\ldots,\phi_K^{(n)}(\dv Y_K^n),\dv Y_0^n)) \;.
\end{align*}  
It is easy to see that a characterization of $\mc{RI}_\mathrm{CEO}^\star$ can be obtained by using~\eqref{rate-distortion-region-DM-CEO-model} and substituting distortion levels $D$ therein with $(\Delta: = H(\dv X) - D)$. 
\begin{proposition}~\label{proposition-DM-RI-CEO}
The rate-information region $\mc{RI}_\mathrm{CEO}^\star$ of the vector DM CEO problem under logarithmic loss is given by the set of all non-negative tuples $(R_1,\ldots, R_K,\Delta)$ that satisfy, for all subsets $\mc S \subseteq \mc K$,
\begin{equation*}
\sum_{k \in \mc S} R_k \geq \sum_{k \in \mc S} I(\dv Y_k;U_k|\dv X,\dv Y_0,Q) - I(\dv X;U_{\mc S^c},\dv Y_0,Q) + \Delta \;,
\end{equation*}
for some joint measure of the form $P_{\dv Y_0,\dv Y_{\mc K},\dv X}(\dv y_0,\dv y_{\mc K},\dv x)P_Q(q) \prod_{k=1}^{K} P_{U_k|\dv Y_k,Q}(u_k|\dv y_k,q)$. \qedblack
\end{proposition}

\noindent The region $\mc{RI}_\mathrm{CEO}^\star$ involves mutual information terms only (not entropies); and, so, using a standard discretization argument, it can be easily shown that a characterization of this region in the case of continuous alphabets is also given by Proposition~\ref{proposition-DM-RI-CEO}. It is well-known that the rate region of the DM CEO problem under logarithmic loss can also be used  in the case of continuous alphabets (e.g., Gaussian sources~\cite{CJW14, C12}).

\noindent Let us now return to the vector Gaussian CEO problem under logarithmic loss that we study in this section. First, we state the following lemma, whose proof is easy and is omitted for brevity.

\begin{lemma}~\label{lemma-relation-RD-RI}
$(R_1,\ldots,R_K,D) \in \mc{RD}_{\mathrm{VG}\text{-}\mathrm{CEO}}^\star$ if and only if $(R_1,\ldots,R_K,h(\dv X)-D) \in \mc{RI}_\mathrm{CEO}^\star$. \qedblack
\end{lemma}

\noindent Summarizing, using Proposition~\ref{proposition-DM-RI-CEO} and Lemma~\ref{lemma-relation-RD-RI} it follows that $\mc{RD}_{\mathrm{VG}\text{-}\mathrm{CEO}}^\star = \mc{RD}_\mathrm{CEO}^\mathrm{I}$. To complete the proof of Theorem~\ref{theorem-continuous-RD1-CEO}, it remains to show that $\mc{RD}_\mathrm{CEO}^\mathrm{I} = \mc{RD}_\mathrm{CEO}^\mathrm{II}$; and this follows by reasoning along the submodularity arguments of the proof of~\cite[Theorem 10]{CW14}.

\section{Proof of Converse of Theorem~\ref{theorem-Gauss-RD-CEO}}\label{proof-converse-Gauss-RD-CEO}

The proof of Theorem~\ref{theorem-Gauss-RD-CEO} relies on deriving an outer bound on the region $\mc{RD}_\mathrm{CEO}^\mathrm{I}$ given by Theorem~\ref{theorem-continuous-RD1-CEO}. In doing so, we use the technique of~\cite[Theorem 8]{EU14} which relies on the de Bruijn identity and the properties of Fisher information; and extend the argument to account for the time-sharing variable $Q$ and side information $\dv Y_0$.

We first state the following lemma.

\begin{lemma}{\cite{DCT91,EU14}}~\label{lemma-fisher}
Let $(\mathbf{X,Y})$  be a pair of random vectors with pmf $p(\mathbf{x},\mathbf{y})$. We have
\begin{equation*}
\log|(\pi e) \dv J^{-1}(\dv X|\dv Y)| \leq h(\dv X|\dv Y) \leq \log|(\pi e) \mathrm{mmse}(\dv X|\dv Y)| \;,
\end{equation*}
where the conditional Fisher information matrix is defined as
\begin{equation*}
\dv J(\dv X|\dv Y) := \E[\nabla\log p(\dv X|\dv Y) \nabla\log p(\dv X|\dv Y)^\dagger] \;,
\end{equation*}
and the minimum mean squared error (MMSE) matrix is 
\begin{equation*}
\mathrm{mmse}(\dv X|\dv Y) := \E[(\dv X-\E[\dv X|\dv Y])(\dv X-\E[\dv X|\dv Y])^\dagger] \;. \bqed 
\end{equation*}
\end{lemma}

Now, we derive an outer bound on~\eqref{equation-continous-RD1-CEO} as follows. For each $ q\in \mc{Q}$ and fixed pmf $\prod_{k=1}^K p(u_k|\dv y_k,q)$, choose $\{\dv\Omega_{k,q}\}_{k=1}^K$ satisfying $\dv 0 \preceq \dv\Omega_{k,q} \preceq \dv\Sigma_k^{-1}$ such that 
\begin{equation}~\label{equation-outer-1-mmse}
\mathrm{mmse}(\dv Y_k|\dv X, U_{k,q}, \dv Y_0, q) = \dv\Sigma_k - \dv\Sigma_k \dv\Omega_{k,q} \dv\Sigma_k \;.
\end{equation}
Such $\dv\Omega_{k,q}$ always exists since, for all $q\in \mc Q$, $k\in \mc K$, we have
\begin{equation*}~\label{equation-outer-2-covariance-noise}
\dv 0 \preceq \mathrm{mmse}(\dv Y_k|\dv X,U_{k,q},\dv Y_0,q) \preceq \dv\Sigma_{\dv y_k|(\dv x, \dv y_0)} = \dv\Sigma_{\dv n_k| \dv n_0} = \dv\Sigma_k \;.
\end{equation*} 

\noindent Then, for $k\in \mc K$ and $q\in \mc Q$, we have
\begin{align}
I(\dv Y_k; U_k|\dv X, \dv Y_0, Q=q) &= \log|(\pi e)\dv\Sigma_k| - h(\dv Y_k|\dv X, U_{k,q}, \dv Y_0, Q=q) \nonumber\\
&\stackrel{(a)}{\geq} \log|\dv\Sigma_k| - \log|\mathrm{mmse}(\dv Y_k|\dv X, U_{k,q}, \dv Y_0, Q=q)| \nonumber\\
&\stackrel{(b)}{=} -\log|\dv I- \dv\Omega_{k,q}\dv\Sigma_k| \label{equation-Gausss-CEO-first-inequality-q}  \;,
\end{align}
where $(a)$ is due to Lemma~\ref{lemma-fisher}; and $(b)$ is due to~\eqref{equation-outer-1-mmse}.

\noindent For convenience, the matrix $\dv\Lambda_{\bar{\mc S},q}$ is defined as follows 
\begin{align}~\label{equation-definition-Tq}
\dv\Lambda_{\bar{\mc S},q} :=
\begin{bmatrix}
\dv 0 & \dv 0 \\
\dv 0 & \mathrm{diag}(\{ \dv\Sigma_k - \dv\Sigma_k \dv\Omega_{k,q} \dv\Sigma_k \}_{k\in\mc S^c})
\end{bmatrix}. 
\end{align}
Then, for $q\in \mc Q$ and $\mc {S}\subseteq \mc K$, we have 
\begin{align}
h(\dv X|U_{S^c,q}, \dv Y_0, Q=q) 
&\stackrel{(a)}{\geq} \log|(\pi e) \dv J^{-1}(\dv X| \dv U_{S^c,q}, \dv Y_0, q)| \nonumber\\
&\stackrel{(b)}{=} \log\left| (\pi e) \left( \dv\Sigma_{\dv x}^{-1} + \dv H_{\bar{\mc S}}^\dagger \dv\Sigma_{\dv n_{\bar{\mc S}}}^{-1} \big( \dv I - \dv\Lambda_{\bar{\mc S},q} \dv\Sigma_{\dv n_{\bar{\mc S}}}^{-1} \big) \dv H_{\bar{\mc S}} \right)^{-1} \right| \;, \label{equation-Gausss-CEO-second-inequality-q}
\end{align}
where $(a)$ follows from Lemma~\ref{lemma-fisher}; and for $(b)$, we use the connection of the MMSE and the Fisher information to show the following equality
\begin{equation}~\label{equation-Fisher-equality}
\dv J(\dv X|U_{S^c,q}, \dv Y_0, q) = \dv\Sigma_{\dv x}^{-1} + \dv H_{\bar{\mc S}}^\dagger \dv\Sigma_{\dv n_{\bar{\mc S}}}^{-1} \big( \dv I - \dv\Lambda_{\bar{\mc S},q} \dv\Sigma_{\dv n_{\bar{\mc S}}}^{-1} \big) \dv H_{\bar{\mc S}} \;.
\end{equation}
In order to proof~\eqref{equation-Fisher-equality}, we use de Brujin identity to relate the Fisher information with the MMSE as given in the following lemma. 

\begin{lemma}{\cite{EU14,PV06}}~\label{lemma-Brujin}
Let $(\dv V_1,\dv V_2)$ be a random vector with finite second moments and $\dv Z\sim\mc{CN}(\dv 0, \dv\Sigma_{\dv z})$ independent of $(\dv V_1,\dv V_2)$. Then
\begin{equation*}
\mathrm{mmse}(\dv V_2|\dv V_1,\dv V_2+\dv Z) = \dv\Sigma_{\dv z} - \dv\Sigma_{\dv z} \dv J(\dv V_2+\dv Z|\dv V_1) \dv\Sigma_{\dv z} \;. \bqed 
\end{equation*}
\end{lemma}

From MMSE estimation of Gaussian random vectors, for $\mc {S}\subseteq \mc K$, we have  
\begin{equation}~\label{equation-outer-3}
\dv X = \E[\dv X|\dv Y_{\bar{\mc S}}] + \dv W_{\bar{\mc S}} = \dv G_{\bar{\mc S}} \dv Y_{\bar{\mc S}} + \dv W_{\bar{\mc S}} \;,
\end{equation}
where $\dv G_{\bar{\mc S}} := \dv\Sigma_{\dv w_{\bar{\mc S}}} \dv H_{\bar{\mc S}}^\dagger \dv\Sigma_{\dv n_{\bar{\mc S}}}^{-1}$, and $\dv W_{\bar{\mc S}} \sim \mathcal{CN}(\dv 0, \dv\Sigma_{\dv w_{\bar{\mc S}}})$ is a Gaussian vector that is independent of $\dv Y_{\bar{\mc S}}$ and  
\begin{equation}~\label{equation-covariance-w}
\dv\Sigma_{\dv w_{\bar{\mc S}}}^{-1} :=  \dv\Sigma_{\dv x}^{-1} + \dv H_{\bar{\mc S}}^\dagger \dv\Sigma_{\dv n_{\bar{\mc S}}}^{-1} \dv H_{\bar{\mc S}} \;.
\end{equation}

\noindent Now we show that the cross-terms of $\mathrm{mmse}\left(\dv Y_{\mc S^c}|\dv X,U_{\mc S^{c},q},\dv Y_0,q \right)$ are zero (similarly to~\cite[Appendix V]{EU14}). For $i\in\mc S^c$ and $j\neq i$, we have 
\begin{align}
\E&\big[(Y_i-\E[Y_i|\dv X,U_{\mc S^{c},q},\dv Y_0,q])(Y_j-\E[Y_j|\dv X,U_{\mc S^{c},q},\dv Y_0,q])^\dagger\big] \nonumber\\
&\stackrel{(a)}{=}  \E\bigg[ \E\big[ (Y_i-\E[Y_i|\dv X,U_{\mc S^{c},q},\dv Y_0,q])(Y_j-\E[Y_j|\dv X,U_{\mc S^{c},q},\dv Y_0,q])^\dagger|\dv X,\dv Y_0 \big] \bigg] \nonumber\\
&\stackrel{(b)}{=}  \E\bigg[ \E\big[ (Y_i-\E[Y_i|\dv X,U_{\mc S^{c},q},\dv Y_0,q])|\dv X,\dv Y_0 \big] \E\big[(Y_j-\E[Y_j|\dv X,U_{\mc S^{c},q},\dv Y_0,q])^\dagger|\dv X,\dv Y_0 \big] \bigg] 
= \dv 0 \;, \label{equation-cross-terms}
\end{align}
where $(a)$ is due to the law of total expectation; $(b)$ is due to the Markov chain $\dv Y_k \mkv (\dv X,\dv Y_0) \mkv \dv Y_{\mc K \setminus k}$. 

\noindent Then, for $k\in \mc K$ and $q\in \mc Q$, we have
\begin{align}
\mathrm{mmse}\big(\dv G_{\bar{\mc S}} \dv Y_{\bar{\mc S}} \big|\dv X,U_{\mc S^c,q},\dv Y_0,q \big) 
&= \dv G_{\bar{\mc S}}  \: \mathrm{mmse}\left(\dv Y_{\bar{\mc S}}|\dv X,U_{\mc S^c,q},\dv Y_0,q \right) \dv G_{\bar{\mc S}}^\dagger \nonumber\\
&\stackrel{(a)}{=} 
\dv G_{\bar{\mc S}}
\begin{bmatrix}
\dv 0 & \dv 0 \\
\dv 0 & \mathrm{diag}(\{\mathrm{mmse}(\dv Y_k|\dv X,U_{\mc S^c,q},\dv Y_0,q)\}_{k\in\mc S^c})
\end{bmatrix}
\dv G_{\bar{\mc S}}^\dagger \nonumber\\ 
&\stackrel{(b)}{=} \dv G_{\bar{\mc S}} \dv\Lambda_{\bar{\mc S},q} \dv G_{\bar{\mc S}}^\dagger \;, \label{equation-outer-4}
\end{align}
where $(a)$ follows since the cross-terms are zero as shown in~\eqref{equation-cross-terms}; and $(b)$ follows due to~\eqref{equation-outer-1-mmse} and the definition of $\dv\Lambda_{\bar{\mc S},q}$ given in~\eqref{equation-definition-Tq}.

\noindent Finally, we obtain the equality~\eqref{equation-Fisher-equality} by applying Lemma~\ref{lemma-Brujin} and noting \eqref{equation-outer-3} as follows   
\begin{align*}    
\dv J(\dv X|U_{S^c,q},\dv Y_0,q) 
&\stackrel{(a)}{=} \dv\Sigma_{\dv w_{\bar{\mc S}}}^{-1} - \dv\Sigma_{\dv w_{\bar{\mc S}}}^{-1} \: \mathrm{mmse} \big( \dv G_{\bar{\mc S}} \dv Y_{\bar{\mc S}} \big| \dv X,U_{\mc S^{c},q},\dv Y_0,q \big) \dv\Sigma_{\dv w_{\bar{\mc S}}}^{-1} \\[0.1em]
&\stackrel{(b)}{=} \dv\Sigma_{\dv w_{\bar{\mc S}}}^{-1} - \dv\Sigma_{\dv w_{\bar{\mc S}}}^{-1} \dv G_{\bar{\mc S}} \dv\Lambda_{\bar{\mc S},q} \dv G_{\bar{\mc S}}^\dagger \dv\Sigma_{\dv w_{\bar{\mc S}}}^{-1} \\[0.1em]
&\stackrel{(c)}{=} \dv\Sigma_{\dv x}^{-1} + \dv H_{\bar{\mc S}}^\dagger \dv\Sigma_{\dv n_{\bar{\mc S}}}^{-1} \dv H_{\bar{\mc S}} -  \dv H_{\bar{\mc S}}^\dagger \dv\Sigma_{\dv n_{\bar{\mc S}}}^{-1} \dv\Lambda_{\bar{\mc S},q} \dv\Sigma_{\dv n_{\bar{\mc S}}}^{-1} \dv H_{\bar{\mc S}}\\[0.1em]
&= \dv\Sigma_{\dv x}^{-1} + \dv H_{\bar{\mc S}}^\dagger \dv\Sigma_{\dv n_{\bar{\mc S}}}^{-1} \big( \dv I - \dv\Lambda_{\bar{\mc S},q} \dv\Sigma_{\dv n_{\bar{\mc S}}}^{-1}  \big) \dv H_{\bar{\mc S}} \;,
\end{align*}
where $(a)$ is due to Lemma~\ref{lemma-Brujin}; $(b)$ is due to~\eqref{equation-outer-4}; and $(c)$ follows due to the definitions of $\dv\Sigma_{\dv w_{\bar{\mc S}}}^{-1}$ and $\dv G_{\bar{\mc S}}$.

Next, we average~\eqref{equation-Gausss-CEO-first-inequality-q} and~\eqref{equation-Gausss-CEO-second-inequality-q} over the time-sharing $Q$ and letting $\dv\Omega_k := \sum_{q\in \mc Q}p(q) \dv\Omega_{k,q}$, we obtain the lower bound
\begin{align}
I(\dv Y_k;\dv U_k|\dv X,\dv Y_0,Q) &= \sum_{q \in \mc Q} p(q) I(\dv Y_k;\dv U_k|\dv X,\dv Y_0,Q=q)  \nonumber\\
&\stackrel{(a)}{\geq} - \sum_{q \in \mc Q} p(q) \log|\dv I- \dv\Omega_{k,q} \dv\Sigma_k| \nonumber\\
&\stackrel{(b)}{\geq} -\log |\dv I - \sum_{q \in \mc Q} p(q) \dv\Omega_{k,q} \dv\Sigma_k| \nonumber\\  
&= -\log |\dv I- \dv\Omega_k \dv\Sigma_k| \;, \label{equation-Gausss-CEO-first-inequality}
\end{align}
where $(a)$ follows from~\eqref{equation-Gausss-CEO-first-inequality-q}; and $(b)$ follows from the concavity of the log-det function and Jensen's Inequality. 

\noindent Besides,  we can derive the following lower bound
\begin{align}
h(\dv X|U_{S^c},\dv Y_0,Q)  &= \sum_{q \in \mc Q} p(q) h(\dv X|U_{S^c,q}, \dv Y_0, Q=q) \nonumber\\ 
&\stackrel{(a)}{\geq} \sum_{q \in \mc Q} p(q) \log\left| (\pi e) \left(\dv\Sigma_{\dv x}^{-1} + \dv H_{\bar{\mc S}}^\dagger \dv\Sigma_{\dv n_{\bar{\mc S}}}^{-1} \big( \dv I - \dv\Lambda_{\bar{\mc S},q} \dv\Sigma_{\dv n_{\bar{\mc S}}}^{-1}  \big) \dv H_{\bar{\mc S}} \right)^{-1} \right| \nonumber\\
&\stackrel{(b)}{\geq} \log\left| (\pi e) \left( \dv\Sigma_{\dv x}^{-1} + \dv H_{\bar{\mc S}}^\dagger \dv\Sigma_{\dv n_{\bar{\mc S}}}^{-1} \big( \dv I - \dv\Lambda_{\bar{\mc S}} \dv\Sigma_{\dv n_{\bar{\mc S}}}^{-1}  \big) \dv H_{\bar{\mc S}} \right)^{-1} \right| \;, \label{equation-Gausss-CEO-second-inequality}  
\end{align}
where $(a)$ is due to~\eqref{equation-Gausss-CEO-second-inequality-q}; and $(b)$ is due to the concavity of the log-det function and Jensen's inequality and the definition of $\dv\Lambda_{\bar{\mc S}}$ given in~\eqref{equation-definition-T}. 

Finally, the outer bound on $\mc{RD}_{\mathrm{VG}\text{-}\mathrm{CEO}}^\star$ is obtained by applying~\eqref{equation-Gausss-CEO-first-inequality}   and~\eqref{equation-Gausss-CEO-second-inequality} in~\eqref{equation-continous-RD1-CEO}, noting that $\dv\Omega_k = \sum_{q\in \mathcal{Q}}p(q) \dv\Omega_{k,q} \preceq\mathbf{\Sigma}_{k}^{-1}$ since $\mathbf{0} \preceq \dv\Omega_{k,q} \preceq\mathbf{\Sigma}_{k}^{-1}$, and taking the union over $\dv\Omega_k$ satisfying $\mathbf{0} \preceq \dv\Omega_k \preceq\mathbf{\Sigma}_{k}^{-1}$.

\section{Proof of Theorem~\ref{theorem-quadratic-RD-CEO}}\label{proof-quadratic-RD-CEO}

We first present the following lemma, which essentially states that Theorem~\ref{theorem-Gauss-RD-CEO} provides an outer bound on $\mc{RD}_{\mathrm{VG}\text{-}\mathrm{CEO}}^\mathrm{det}$.

\begin{lemma}~\label{lemma-connection}
If $(R_1,\ldots,R_K,D) \in \mc{RD}_{\mathrm{VG}\text{-}\mathrm{CEO}}^\mathrm{det}$, then $(R_1,\ldots,R_K, \log ({{\pi}e})^{n_x} D) \in \mc{RD}_\mathrm{CEO}^\mathrm{I}$.  
\end{lemma}

\begin{proof}
Let a tuple $(R_1,\ldots,R_K,D) \in \mc{RD}_{\mathrm{VG}\text{-}\mathrm{CEO}}^\mathrm{det}$ be given. Then, there exist a blocklength $n$, $K$ encoding functions $\{\breve{\phi}^{(n)}_k\}^K_{k=1}$ and a decoding function $\breve{\psi}^{(n)}$ such that    
\begin{align}
R_k &\geq \frac{1}{n}\log M^{(n)}_k \;, \quad\text{for}\:\: k=1,\ldots,K \;, \nonumber\\
D &\geq \bigg|\frac{1}{n} \sum_{i=1}^{n} \mathrm{mmse}(\dv X_i|\breve{\phi}^{(n)}_1(\dv Y_1^n),\ldots,\breve{\phi}^{(n)}_K(\dv Y_K^n),\dv Y_0^n)  \bigg| \;. \label{equation-log-quadrtic-1}
\end{align} 
	
\noindent We need to show that there exist $(U_1,\ldots,U_K,Q)$ such that 
\begin{equation}~\label{required-existence-condition-proof-lemma}
\sum_{k \in \mc S} R_k + \log (\pi e)^{n_x} D \geq \sum_{k \in \mc S} I(\dv Y_k;U_k|\dv X,\dv Y_0,Q) + h(\dv X|U_{\mc S^c},\dv Y_0,Q) \;, \quad\text{for}\:\: \mc S \subseteq \mc K \;.
\end{equation}
	
\noindent Let us define  
\begin{equation*}
\bar{\Delta}^{(n)} := \frac{1}{n} h(\dv X^n|\breve{\phi}^{(n)}_1(\dv Y_1^n),\ldots,\breve{\phi}^{(n)}_K(\dv Y_K^n),\dv Y_0^n) \;.
\end{equation*}

\noindent It is easy to justify that expected distortion $\bar{\Delta}^{(n)}$ is achievable under logarithmic loss (see Theorem~\ref{theorem-continuous-RD1-CEO}).  Then, following straightforwardly the lines in the proof of~\cite[Theorem 10]{CW14}, we have
\begin{align}~\label{equation-log-quadrtic-2}
\sum_{k \in \mc S} R_k  \geq \sum_{k \in \mc S} \frac{1}{n} \sum_{i=1}^{n} I(\dv Y_{k,i};U_{k,i}|\dv X_i,\dv Y_{0,i},Q_i) 
+ \frac{1}{n} \sum_{i=1}^{n} h(\dv X_i|U_{\mc S^c,i},\dv Y_{0,i},Q_i) - \bar{\Delta}^{(n)} \;.
\end{align}
	
Next, we upper bound $\bar{\Delta}^{(n)}$ in terms of $D$ as follows
\begin{align} 
\bar{\Delta}^{(n)} &= \frac{1}{n} h(\dv X^n | \breve{\phi}^{(n)}_1(\dv Y_1^n),\ldots,\breve{\phi}^{(n)}_K(\dv Y_K^n),\dv Y_0^n ) \nonumber\\
& =  \frac{1}{n} \sum_{i=1}^{n} h(\dv X_i|\dv X_{i+1}^n, \breve{\phi}^{(n)}_1(\dv Y_1^n),\ldots,\breve{\phi}^{(n)}_K(\dv Y_K^n),\dv Y_0^n ) \nonumber\\
&= \:  \frac{1}{n} \sum_{i=1}^{n} h(\dv X_i - \E[\dv X_i | J_{\mc K} ] \big| \dv X_{i+1}^n, \breve{\phi}^{(n)}_1(\dv Y_1^n),\ldots,\breve{\phi}^{(n)}_K(\dv Y_K^n),\dv Y_0^n) \nonumber\\
& \stackrel{(a)}{\leq} \frac{1}{n} \sum_{i=1}^{n} h(\dv X_i - \E[\dv X_i|\breve{\phi}^{(n)}_1(\dv Y_1^n),\ldots,\breve{\phi}^{(n)}_K(\dv Y_K^n),\dv Y_0^n ) \nonumber\\
& \stackrel{(b)}{\leq} \frac{1}{n} \sum_{i=1}^{n} \log (\pi e)^{n_x} \left| \mathrm{mmse}(\dv X_i|\breve{\phi}^{(n)}_1(\dv Y_1^n),\ldots,\breve{\phi}^{(n)}_K(\dv Y_K^n),\dv Y_0^n) \right| \nonumber\\
& \stackrel{(c)}{\leq} \log ({{\pi}e})^{n_x} \bigg| \frac{1}{n} \sum_{i=1}^{n}  \mathrm{mmse}(\dv X_i|\breve{\phi}^{(n)}_1(\dv Y_1^n),\ldots,\breve{\phi}^{(n)}_K(\dv Y_K^n),\dv Y_0^n) \bigg| \nonumber\\
& \stackrel{(d)}{\leq} \log ({{\pi}e})^{n_x} D \label{equation-log-quadrtic-3} \;,
\end{align}
where  $(a)$  holds since conditioning reduces entropy; $(b)$ is due to the maximal differential entropy lemma; $(c)$ is due to the convexity of the log-det function and Jensen's inequality; and $(d)$ is due to~\eqref{equation-log-quadrtic-1}. 
	
\noindent Combining~\eqref{equation-log-quadrtic-3} with~\eqref{equation-log-quadrtic-2}, and using standard arguments for single-letterization, we get~\eqref{required-existence-condition-proof-lemma}; and this completes the proof of the lemma.
\end{proof}

The proof of Theorem~\ref{theorem-quadratic-RD-CEO} is as follows. By Lemma~\ref{lemma-connection} and Proposition~\ref{proposition-continous-RD2-CEO}, there must exist Gaussian test channels $(V^\mathrm{G}_1,\ldots,V^\mathrm{G}_K)$ and a time-sharing random variable $Q^\prime$, with joint distribution that factorizes as
\begin{equation*}
P_{\dv X,\dv Y_0}(\dv x,\dv y_0) \prod_{k=1}^K P_{\dv Y_k|\dv X,\dv Y_0}(\dv y_k|\dv x,\dv y_0) \: P_Q^\prime(q^\prime) \prod_{k=1}^{K} P_{V_k|\dv Y_k,Q^\prime}(v_k|\dv y_k,q^\prime) \;,
\end{equation*}
such that the following holds 
\begin{align} 
\sum_{k \in \mc S} R_k &\geq I(\dv Y_{\mc S};V_{\mc S}^\mathrm{G}|V_{\mc S^c}^\mathrm{G},\dv Y_0,Q^\prime) \;, \quad\text{for}\:\: \mc S \subseteq \mc K \;, \\
\log (\pi e)^{n_x} D &\geq h(\dv X|V_1^\mathrm{G},\ldots,V_K^\mathrm{G},\dv Y_0,Q^\prime) \;. \label{equation-optimal-rate-distortion region}
\end{align}

\noindent This is clearly achievable by the Berger-Tung coding scheme with Gaussian test channels and time-sharing $Q^\prime$, since the achievable error matrix under quadratic distortion has determinant that satisfies
\begin{equation*}
\log \big( (\pi e)^{n_x} | \mathrm{mmse}(\dv X|V_1^\mathrm{G},\ldots,V_K^\mathrm{G},\dv Y_0,Q')| \big) = h(\dv X|V_1^G,\ldots,V_K^G,\dv Y_0,Q^\prime) \;.
\end{equation*}

\noindent The above shows that the rate-distortion region of the quadratic vector Gaussian CEO problem with determinant constraint is given by~\eqref{equation-optimal-rate-distortion region}, i.e., $\mc{RD}_\mathrm{CEO}^\mathrm{II}$ (with distortion parameter $\log(\pi e)^{n_x}D$). Recalling that $\mc{RD}_\mathrm{CEO}^\mathrm{II}=\mc{RD}_\mathrm{CEO}^\mathrm{I}=\mc{RD}_{\mathrm{VG}\text{-}\mathrm{CEO}}^\star$, and substituting in Theorem~\ref{theorem-Gauss-RD-CEO} using distortion level $\log({\pi}e)^{n_x}D$ completes the proof.

\section{Proof of Proposition~\ref{proposition-example-distributed-classification}}\label{proof-proposition-example-distributed-classification}

Let a triple mappings $(Q_{U_1|Y_1},Q_{U_2|Y_2},Q_{\hat{X}|U_1,U_2,Y_0})$ be given. It is easy to see that the probability of classification error of the classifier $Q_{\hat{X}|Y_0,Y_1,Y_2}$ as defined by~\eqref{definition-probability-classification-error} satisfies
\begin{equation}~\label{surrogate-on-probability-classification-error}
P_{\mathcal E}(Q_{\hat{X}|Y_0,Y_1,Y_2}) \leq \mathbb{E}_{P_{X,Y_0,Y_1,Y_2}} [- \log Q_{\hat{X}|Y_0,Y_1,Y_2}(X|Y_0,Y_1,Y_2)] \;.
\end{equation}
Applying Jensen's inequality on the right hand side (RHS) of \eqref{surrogate-on-probability-classification-error}, using the concavity of the logarithm function, and combining with the fact that the exponential function increases monotonically, the probability of classification error can be further bounded as
\begin{equation}~\label{bound-on-surrogate-on-probability-classification-error-step1}
P_{\mathcal E}(Q_{\hat{X}|Y_0,Y_1,Y_2}) \leq 1 - \exp\Big(- \mathbb{E}_{P_{X,Y_0,Y_1,Y_2}} [ - \log Q_{\hat{X}|Y_0,Y_1,Y_2}(X|Y_0,Y_1,Y_2) ]\Big) \;.
\end{equation}
Using~\eqref{definition-equivalent-classifier} and continuing from~\eqref{bound-on-surrogate-on-probability-classification-error-step1}, we get  
\begin{align}
P_{\mathcal E}(Q_{\hat{X}|Y_0,Y_1,Y_2}) &\leq 1 - \exp\Big(- \mathbb{E}_{P_{X,Y_0,Y_1,Y_2}} [ - \log \mathbb{E}_{Q_{U_1|Y_1}} \mathbb{E}_{Q_{U_2|Y_2}} [Q_{\hat{X}|U_1,U_2,Y_0}(X|U_1,U_2,Y_0)] ]\Big) \nonumber\\
&\leq 1 - \exp\Big(- \mathbb{E}_{P_{X,Y_0,Y_1,Y_2}} \mathbb{E}_{Q_{U_1|Y_1}} \mathbb{E}_{Q_{U_2|Y_2}} [ - \log  [Q_{\hat{X}|U_1,U_2,Y_0}(X|U_1,U_2,Y_0)] ]\Big) \;,
\label{bound-on-surrogate-on-probability-classification-error-step2}
\end{align}
where the last inequality follows by applying Jensen's inequality and using the concavity of the logarithm function.

\noindent Noticing that the term in the exponential function in the RHS of~\eqref{bound-on-surrogate-on-probability-classification-error-step2},
\begin{equation}
\mc D(Q_{U_1|Y_1},Q_{U_1|Y_1},Q_{\hat{X}|U_1,U_2,Y_0}) := \mathbb{E}_{P_{XY_0Y_1Y_2}} \mathbb{E}_{Q_{U_1|Y_1}} \mathbb{E}_{Q_{U_2|Y_2}} [ - \log  Q_{\hat{X}|U_1,U_2,Y_0}(X|U_1,U_2,Y_0)] \;,
\end{equation}
is the average logarithmic loss, or cross-entropy risk,  of the triple $(Q_{U_1|Y_1},Q_{U_2|Y_2},Q_{\hat{X}|U_1,U_2,Y_0})$; the inequality~\eqref{bound-on-surrogate-on-probability-classification-error-step2} implies that minimizing the average logarithmic loss distortion leads to classifier with smaller (bound on) its classification error. Using Theorem~\ref{theorem-continuous-RD1-CEO}, the minimum average logarithmic loss, minimized over all mappings \mbox{$Q_{U_1|Y_1}\: :\:  \mc Y_1 \longrightarrow \mc P(\mc U_1)$} and  $Q_{U_2|Y_2}\: :\:  \mc Y_2 \longrightarrow \mc P(\mc U_2)$ that have description lengths no more than $R_1$ and $R_2$ bits per-sample, respectively, as well as all choices of $Q_{\hat{X}|U_1,U_2,Y_0}\: :\: \mc U_1 \times \mc U_2 \times \mc Y_0 \longrightarrow \mc P(\mc X)$, is 
\begin{equation}
D^{\star}(R_1,R_2) = \inf\{ D \: : \: (R_1,R_2, D) \in \mc{RD}_\mathrm{CEO}^\star \} \;.
\end{equation}
Thus, the direct part of Theorem~\ref{theorem-continuous-RD1-CEO} guarantees the existence of a classifier $Q^{\star}_{\hat{X}|Y_0,Y_1,Y_2}$ whose probability of error satisfies the bound given in Proposition~\ref{proposition-example-distributed-classification}.

\section{Proof of Lemma~\ref{lemma-CEO-fixed-P}}\label{proof-lemma-CEO-fixed-P}
	
First, we rewrite $F_s(\dv P)$ in~\eqref{equation-CEO-objective1}. To that end, the second term of the RHS of~\eqref{equation-CEO-objective1} can be proceeded as   
\begin{align}
I(Y_1;U_1|U_2,Y_0) \stackrel{(a)}=& \: I(X,Y_1;U_1|U_2,Y_0) 
= I(X;U_1|U_2,Y_0) + I(Y_1;U_1|U_2,Y_0,X) \nonumber\\ 
\stackrel{(b)}=& \: I(X;U_1|U_2,Y_0) + I(Y_1;U_1|X,Y_0) \nonumber\\ 
=& \: I(X;U_1|U_2,Y_0) + I(Y_1,X;U_1|Y_0) - I(X;U_1|Y_0) \nonumber\\
\stackrel{(c)}=& \: I(X;U_1|U_2,Y_0) + I(Y_1;U_1|Y_0) - I(X;U_1|Y_0) \nonumber\\
=& \: H(X|U_2,Y_0) - H(X|U_1,U_2,Y_0) + H(U_1|Y_0) - H(U_1|Y_0,Y_1) 
- H(X|Y_0) + H(X|U_1,Y_0) \nonumber\\   
=& \: H(X|U_2,Y_0) - H(X|U_1,U_2,Y_0) + H(U_1) - H(Y_0) + H(Y_0|U_1) \nonumber\\
& - H(U_1|Y_0,Y_1)- H(X|Y_0) + H(X|U_1,Y_0) \label{equation-I_Y1U_1_g_U2Y0} \;,
\end{align}
and, the third term of the RHS of~\eqref{equation-CEO-objective1} can be written as    
\begin{align}
I(Y_2;U_2|Y_0) &= H(U_2|Y_0) - H(U_2|Y_0, Y_2) 
\stackrel{(d)}= H(U_2|Y_0) - H(U_2|Y_2) \nonumber\\
&= H(U_2) - H(Y_0) + H(Y_0|U_2) - H(U_2|Y_2) \;, \label{equation-I_Y2U_2_g_Y0}
\end{align}
where $(a)$, $(b)$, $(c)$ and $(d)$ follows due to the Markov chain $U_1 \mkv Y_1 \mkv (X, Y_0) \mkv Y_2 \mkv U_2$. 
	
\noindent
By applying~\eqref{equation-I_Y1U_1_g_U2Y0} and~\eqref{equation-I_Y2U_2_g_Y0} in~\eqref{equation-CEO-objective1}, we have 
\begin{align}
F_{\dv s}(\dv P)  =& - s_1 H(X|Y_0) - (s_1 + s_2) H(Y_0)
+ (1 - s_1) H(X|U_1,U_2,Y_0) \nonumber\\
& + s_1 H(X|U_1,Y_0) + s_1 H(X|U_2,Y_0) 
+ s_1  H(U_1) - s_1 H(U_1|Y_1) \nonumber\\
& + s_2  H(U_2) - s_2 H(U_2|Y_2) + s_1 H(Y_0|U_1) + s_2 H(Y_0|U_2) \nonumber\\[0.5em]
=& - s_1 H(X|Y_0) - (s_1 + s_2) H(Y_0) \nonumber\\
& - (1-s_1) \sum_{u_1 u_2 x y_0} p(u_1,u_2,x,y_0) \log p(x|u_1,u_2,y_0) \nonumber\\ 
& - s_1 \sum_{u_1 x y_0} p(u_1,x,y_0) \log p(x|u_1,y_0) 
- s_1 \sum_{u_2 x y_0} p(u_2,x,y_0) \log p(x|u_2,y_0) \nonumber\\ 
& - s_1 \sum_{u_1} p(u_1) \log p(u_1) 
+ s_1 \sum_{u_1 y_1} p(u_1,y_1) \log p(u_1|y_1) \nonumber\\
& - s_2 \sum_{u_2} p(u_2) \log p(u_2) 
+ s_2 \sum_{u_2 y_2} p(u_2,y_2) \log p(u_2|y_2) \nonumber\\
& - s_1 \sum_{u_1 y_0} p(u_1,y_0) \log p(y_0|u_1)  - s_2 \sum_{u_2 y_0} p(u_2,y_0) \log p(y_0|u_2) \;, \label{equation-Fs_P}
\end{align}
	
\noindent
Then, marginalizing~\eqref{equation-Fs_P} over variables $X, Y_0, Y_1, Y_2$, and using the Markov chain $U_1 \mkv Y_1 \mkv (X, Y_0) \mkv Y_2 \mkv U_2$, it is easy to see that $F_{\dv s}(\dv P)$ can be written as  
\begin{align}
F_{\dv s}(\dv P) =& - s_1 H(X|Y_0) - (s_1 + s_2) H(Y_0) \nonumber\\
&+ \E_{P_{X,Y_0,Y_1,Y_2}} \Big[(1-s_1) \E_{P_{U_1| Y_1}}\E_{P_{U_2| Y_2}}[- \log P_{X|U_1,U_2,Y_0}] \nonumber\\
&\hspace{6.5em} + s_1 \E_{P_{U_1| Y_1}}[- \log P_{X|U_1,Y_0}] + s_1 \E_{P_{U_2| Y_2}}[- \log P_{X|U_2,Y_0}] \nonumber\\ 
&\hspace{6.5em} + s_1 D_\mathrm{KL}(P_{U_1|Y_1}\|P_{U_1}) + s_2 D_\mathrm{KL}(P_{U_2|Y_2}\|P_{U_2}) \nonumber\\   
&\hspace{6.5em} + s_1 \E_{P_{U_1| Y_1}}[- \log P_{Y_0|U_1}] + s_2 \E_{P_{U_2| Y_2}}[- \log P_{Y_0|U_2}] \Big] \;.
\end{align}
	
\noindent 
Hence, we have 
\begin{align*}	
F_{\dv s}(\dv P, \dv Q) - F_{\dv s}(\dv P) 
=& \: (1-s_1) \E_{U_1,U_2,Y_0}[D_\mathrm{KL}(P_{X|U_1,U_2,Y_0}\|Q_{X|U_1,U_2,Y_0})] \\
&+ s_1 \E_{U_1,Y_0}[D_\mathrm{KL}(P_{X|U_1,Y_0}\|Q_{X|U_1,Y_0})] 
+ s_1 \E_{U_2,Y_0}[D_\mathrm{KL}(P_{X|U_2,Y_0}\|Q_{X|U_2,Y_0})] \\
&+ s_1 D_\mathrm{KL}(P_{U_1}\|Q_{U_1}) 
+ s_2 D_\mathrm{KL}(P_{U_2}\|Q_{U_2}) \\
& + s_1 \E_{U_1}[D_\mathrm{KL}(P_{Y_0|U_1}\|Q_{Y_0|U_1})]
+ s_2 \E_{U_2}[D_\mathrm{KL}(P_{Y_0|U_2}\|Q_{Y_0|U_2})] \geq 0 \;,
\end{align*}	
where it holds with equality if and only if~\eqref{equation-CEO-optimal-Q} is satisfied. Note that we have the relation $1-s_1\geq 0$ due to Lemma~\ref{lemma-range}. This completes the proof.

\section{Proof of Lemma~\ref{lemma-CEO-fixed-Q}}\label{proof-lemma-CEO-fixed-Q}

We have that $F_{\dv s}(\dv P, \dv Q)$ is convex in $\dv P$ from Lemma~\ref{lemma-CEO-convex}. For a given $\dv Q$ and $\dv s$, in order to minimize $F_{\dv s}(\dv P, \dv Q)$ over the convex set of pmfs $\dv P$, let us define the Lagrangian as      	
\begin{equation*}  
\mc L_{\dv s}(\dv P, \dv Q, \boldsymbol\lambda) :=  F_{\dv s}(\dv P, \dv Q)  + \sum_{y_1} \lambda_1(y_1) [1 - \sum_{u_1} p(u_1|y_1)] + \sum_{y_2} \lambda_2(y_2) [1 - \sum_{u_2} p(u_2|y_2)] \;,
\end{equation*}		
where $\lambda_1(y_1) \geq 0$ and $\lambda_2(y_2) \geq 0$ are the Lagrange multipliers corresponding the constrains $\sum_{u_k} p(u_k|y_k) =1$, $y_k \in \mc Y_k$, $k=1,2$, of the pmfs $P_{U_1|Y_1}$ and $P_{U_2|Y_2}$, respectively. Due to the convexity of $F_{\dv s}(\dv P, \dv Q)$, the KKT conditions are necessary and sufficient for optimality. By applying the KKT conditions
\begin{equation*}
\frac{\partial \mc L_{\dv s}(\dv P, \dv Q, \boldsymbol\lambda)}{\partial p(u_1|y_1)}  = 0 \;,
\quad \quad \quad
\frac{\partial \mc L_{\dv s}(\dv P, \dv Q, \boldsymbol\lambda)}{\partial p(u_2|y_2)}  = 0 \;,
\end{equation*}
and arranging terms, we obtain 
\begin{align}  
\log &\:p(u_k|y_k) \nonumber\\
=& \: \log q(u_k) + \frac{1-s_1}{s_k} \sum_{u_\kbar x y_0} p(x,y_0|y_k) p(u_\kbar|x,y_0) \log q(x|u_k,u_\kbar,y_0) \nonumber\\
& + \frac{s_1}{s_k} \sum_{x y_0} p(x,y_0|y_k) \log q(x|u_k,y_0) + \sum_{y_0} p(y_0|y_k) \log q(y_0|u_k) + \frac{\lambda_k(y_k)}{s_k p(y_k)} - 1 \nonumber\\[0.5em]
=& \: \log q(u_k) + \frac{1-s_1}{s_k} \sum_{u_\kbar y_0} p(u_\kbar, y_0|y_k) \sum_{x} p(x|y_k, u_\kbar, y_0) \log q(x|u_k,u_\kbar,y_0) \nonumber\\
& + \frac{s_1}{s_k} \sum_{y_0} p(y_0|y_k) \sum_{x} p(x|y_k,y_0) \log q(x|u_k,y_0) + \sum_{y_0} p(y_0|y_k) \log q(y_0|u_k) + \frac{\lambda_k(y_k)}{s_k p(y_k)} - 1 \nonumber\\[0.5em] 
=& \: \log q(u_k) - \frac{1-s_1}{s_k} \sum_{u_\kbar y_0} p(u_\kbar, y_0|y_k) \sum_{x} p(x|y_k,u_\kbar,y_0) \log \frac{p(x|y_k,u_\kbar,y_0)}{q(x|u_k,u_\kbar,y_0)} \frac{1}{p(x|y_k,u_\kbar,y_0)} + \frac{\lambda_k(y_k)}{s_k p(y_k)} - 1  \nonumber\\
& - \frac{s_1}{s_k} \sum_{y_0} p(y_0|y_k) \sum_{x} p(x|y_k,y_0) \log \frac{p(x|y_k,y_0)}{q(x|u_k,y_0)} 
\frac{1}{p(x|y_k,y_0)} - \sum_{y_0} p(y_0|y_k) \log \frac{p(y_0|y_k)}{q(y_0|u_k)} \frac{1}{p(y_0|y_k)} \nonumber\\[0.5em] 
=& \: \log q(u_k) - \psi_k(u_k,y_k) + \tilde{\lambda}_k(y_k) \;, \label{equation-CEO-log} 
\end{align}	
where $\psi_k(u_k,y_k)$, $k=1,2$, are given by \eqref{equation-CEO-psi}, and $\tilde{\lambda}_k(y_k)$ contains all terms independent of $u_k$ for $k=1,2$. Then, we proceeded by rearranging \eqref{equation-CEO-log} as follows
\begin{equation}~\label{equation-CEO-rearrange}
p(u_k|y_k) = e^{\tilde{\lambda}_k(y_k)} q(u_k) e^{- \psi_k(u_k,y_k)} \;, \quad\text{for}\:\: k=1,2 \;.
\end{equation}
Finally, the Lagrange multipliers $\lambda_k(y_k)$ satisfying the KKT conditions are obtained by finding $\tilde{\lambda}_k(y_k)$ such that $\sum_{u_k} p(u_k|y_k) = 1$, $k=1,2$. Substituting in~\eqref{equation-CEO-rearrange}, $p(u_k|y_k)$ can be found as in~\eqref{equation-CEO-optimal-P}.  

\section{Derivation of the Update Rules of  Algorithm~\ref{algo-CEO-Gauss}}\label{proof-CEO-Gauss-algorithm}

In this section, we derive the update rules in Algorithm~\ref{algo-CEO-Gauss} and show that the Gaussian distribution is invariant to the update rules in Algorithm~\ref{algo-CEO}, in line with Theorem~\ref{theorem-Gauss-RD-CEO}.

First, we recall that if $(\dv X_1, \dv X_2)$ are jointly Gaussian, then 
\begin{equation*}	
P_{\dv X_2|\dv X_1} \sim\mc{CN} (\boldsymbol\mu_{\dv x_2|\dv x_1}, \dv\Sigma_{\dv x_2|\dv x_1}) \;,
\end{equation*}
where $\boldsymbol\mu_{\dv x_2|\dv x_1} := \dv K_{\dv x_2|\dv x_1} \dv x_1$, $\dv K_{\dv x_2|\dv x_1} := \dv\Sigma_{\dv x_2,\dv x_1} \dv\Sigma_{\dv x_1}^{-1}$. 

Then, for $\dv Q^{(t+1)}$ computed as in \eqref{equation-CEO-optimal-Q} from $\dv P^{(t)}$, which is a set of Gaussian distributions,  we have
\begin{align*}
&Q_{\dv X|\dv U_1,\dv U_2,\dv Y_0} \sim\mc{CN} (\boldsymbol\mu_{\dv x|\dv u_1,\dv u_2,\dv y_0} , \dv\Sigma_{\dv x|\dv u_1,\dv u_2,\dv y_0}) \;, \quad\quad 
&&Q_{\dv X|\dv U_k,\dv Y_0} \sim\mc{CN} (\boldsymbol\mu_{\dv x|\dv u_k,\dv y_0} , \dv\Sigma_{\dv x|\dv u_k,\dv y_0}) \;, \\[0.2em]  
&Q_{\dv Y_0|\dv U_k} \sim\mc{CN} (\boldsymbol\mu_{\dv y_0|\dv u_k} , \dv\Sigma_{\dv y_0|\dv u_k}) \;, \quad\quad
&&Q_{\dv U_k} \sim\mc{CN} (\dv 0 , \dv\Sigma_{\dv u_k}) \;.
\end{align*}

Next, we look at the update $\dv P^{(t+1)}$ as in \eqref{equation-CEO-optimal-P} from given $\dv Q^{(t+1)}$. To compute $\psi_k(\dv u_k^t,\dv y_k)$, first, we note that 
\begin{equation}~\label{equation-rewrite-divergence}	 
\begin{aligned}
\E_{\dv U_\kbar,\dv Y_0|\dv y_k} D(P_{\dv X|\dv y_k,\dv U_\kbar,\dv Y_0}\|Q_{\dv X|\dv u_k,\dv U_\kbar,\dv Y_0}) 
&= D(P_{\dv U_\kbar,\dv X,\dv Y_0|\dv y_k}\|Q_{\dv U_\kbar,\dv X,\dv Y_0|\dv u_k}) - D(P_{\dv U_\kbar,\dv Y_0|\dv y_k}\|Q_{\dv U_\kbar,\dv Y_0|\dv u_k}), \\[0.5em]
\E_{\dv Y_0|\dv y_k} D(P_{\dv X|\dv y_k,\dv Y_0}\|Q_{\dv X|\dv u_k,\dv Y_0})
&= D(P_{\dv X,\dv Y_0|\dv y_k}\|Q_{\dv X,\dv Y_0|\dv u_k}) - D(P_{\dv Y_0|\dv y_k}\|Q_{\dv Y_0|\dv u_k}) \;,
\end{aligned}
\end{equation}
and that for two multivariate Gaussian distributions, i.e., $P_{\dv X_1} \sim \mc{CN} (\boldsymbol\mu_{\dv x_1}, \dv\Sigma_{\dv x_1})$ and $P_{\dv X_2} \sim \mc{CN} (\boldsymbol\mu_{\dv x_2}, \dv\Sigma_{\dv x_2})$ in $\mathds{C}^N$,
\begin{equation}~\label{equation-Gauss-quadratic-divergence}
D(P_{\dv X_1}\|P_{\dv X_2}) = (\boldsymbol\mu_{\dv x_1} - \boldsymbol\mu_{\dv x_2})^\dagger \dv\Sigma_{\dv x_2}^{-1} (\boldsymbol\mu_{\dv x_1} - \boldsymbol\mu_{\dv x_2}) + \log | \dv\Sigma_{\dv x_2} \dv\Sigma_{\dv x_1}^{-1} | + \mathrm{tr}(\dv\Sigma_{\dv x_2}^{-1} \dv\Sigma_{\dv x_1}) - N \;.   
\end{equation}
Applying \eqref{equation-rewrite-divergence} and \eqref{equation-Gauss-quadratic-divergence} in \eqref{equation-CEO-psi} and noting that all involved distributions are Gaussian, it follows that $\psi_k(\dv u_k^t,\dv y_k)$ is a quadratic form. Then, since $q^{(t)}(\dv u_k)$ is also Gaussian, the product $\log(q^{(t)}(\dv u_k)\exp(-\psi_k(\dv u_k^t,\dv y_k)))$ is also a quadratic form, and identifying constant, first and second order terms, we can write    
\begin{equation*}
\log p^{(t+1)}(\dv u_k|\dv y_k) = - (\dv u_k - \boldsymbol\mu_{\dv u_k^{t+1}|\dv y_k})^\dagger \dv\Sigma_{\dv z_k^{t+1}}^{-1} (\dv u_k - \boldsymbol\mu_{\dv u_k^{t+1} |\dv y_k}) +  Z(\dv y_k) \;,    
\end{equation*}
where 
\begin{align}
\dv\Sigma_{\dv z_k^{t+1}}^{-1} 
=& \: \dv\Sigma_{\dv u_k^{t}}^{-1} 
+ \frac{1-s_1}{s_k} \dv K_{(\dv u_\kbar^t,\dv x,\dv y_0)|\dv u_k^t}^\dagger \dv\Sigma_{(\dv u_\kbar^t,\dv x,\dv y_0)|\dv u_k^t}^{-1} \dv K_{(\dv u_\kbar^t,\dv x,\dv y_0)|\dv u_k^t} 
- \frac{1-s_1}{s_k} \dv K_{(\dv u_\kbar^t,\dv y_0)|\dv u_k^t}^\dagger \dv\Sigma_{(\dv u_\kbar^t,\dv y_0)|\dv u_k^t}^{-1} \dv K_{(\dv u_\kbar^t,\dv y_0)|\dv u_k^t} \nonumber\\
&+ \frac{s_1}{s_k} \dv K_{(\dv x,\dv y_0)|\dv u_k^t}^\dagger \dv\Sigma_{(\dv x,\dv y_0)|\dv u_k^t}^{-1} \dv K_{(\dv x,\dv y_0)|\dv u_k^t} 
+ \frac{s_k-s_1}{s_k} \dv K_{\dv y_0|\dv u_k^t}^\dagger \dv\Sigma_{\dv y_0|\dv u_k^t}^{-1} \dv K_{\dv y_0|\dv u_k^t} \label{equation-CEO-Gauss-Sigma} \\[1em]    
\boldsymbol\mu_{\dv u_k^{t+1} |\dv y_k} =& 
\: \dv\Sigma_{\dv z_k^{t+1}} \left( 
\frac{1-s_1}{s_k} \dv K_{(\dv u_\kbar^t,\dv x,\dv y_0)|\dv u_k^t}^\dagger \dv\Sigma_{(\dv u_\kbar^t,\dv x,\dv y_0)|\dv u_k^t}^{-1} \dv K_{(\dv u_\kbar^t,\dv x,\dv y_0)|\dv y_k} 
- \frac{1-s_1}{s_k} \dv K_{(\dv u_\kbar^t,\dv y_0)|\dv u_k^t}^\dagger \dv\Sigma_{(\dv u_\kbar^t,\dv y_0)|\dv u_k^t}^{-1} \dv K_{(\dv u_\kbar^t,\dv y_0)|\dv y_k} \right. \nonumber\\
&\left.+ \frac{s_1}{s_k}  \dv K_{(\dv x,\dv y_0)|\dv u_k^t}^\dagger \dv\Sigma_{(\dv x,\dv y_0)|\dv u_k^t}^{-1} \dv K_{(\dv x,\dv y_0)|\dv y_k} 
+ \frac{s_k-s_1}{s_k} \dv K_{\dv y_0|\dv u_k^t}^\dagger \dv\Sigma_{\dv y_0|\dv u_k^t}^{-1} \dv K_{\dv y_0|\dv y_k} 
\right) \dv y_k \;. \label{equation-CEO-Gauss-mu} 
\end{align}
This shows that $p^{(t+1)}\!(\dv u_k|\dv y_k)$ is a Gaussian distribution and that $\dv U_k^{t+1}$ is distributed as $\dv U_k^{t+1}\!\!\sim\!\mc {CN}(\boldsymbol\mu_{\dv u_k^{t+1} |\dv y_k}, \dv\Sigma_{\dv z_k^{t+1}})$. 

Next, we simplify \eqref{equation-CEO-Gauss-Sigma} to obtain the update rule \eqref{equation-CEO-Gauss-Sigma-update}. From the matrix inversion lemma, similarly to \cite{CGTW05}, for $(\dv X_1,\dv X_2)$ jointly Gaussian we have
\begin{equation}~\label{equation-matrix-inv-1}
\dv \Sigma_{\dv x_2|\dv x_1}^{-1} = \dv\Sigma_{\dv x_2}^{-1} + \dv K_{\dv x_1|\dv x_2}^\dagger\dv\Sigma_{\dv x_1|\dv x_2}^{-1}\dv K_{\dv x_1|\dv x_2} \;.
\end{equation}
Applying \eqref{equation-matrix-inv-1} in \eqref{equation-CEO-Gauss-Sigma}, we have
\begin{align*}
\dv\Sigma_{\dv z_k^{t+1}}^{-1} 
= & \: \dv\Sigma_{\dv u_k^{t}}^{-1} 
+ \frac{1-s_1}{s_k} \left( \dv\Sigma_{\dv u_k^t|(\dv u_\kbar^t,\dv x, \dv y_0)}^{-1}  - \dv\Sigma_{\dv u_k^{t}}^{-1} \right)
- \frac{1-s_1}{s_k} \left( \dv\Sigma_{\dv u_k^t|(\dv u_\kbar^t, \dv y_0)}^{-1} - \dv\Sigma_{\dv u_k^{t}}^{-1} \right) \\
&+ \frac{s_1}{s_k} \left( \dv\Sigma_{\dv u_k^t|(\dv x, \dv y_0)}^{-1} - \dv\Sigma_{\dv u_k^{t}}^{-1} \right) 
+ \frac{s_k-s_1}{s_k} \left( \dv\Sigma_{\dv u_k^t|\dv y_0}^{-1} - \dv\Sigma_{\dv u_k^{t}}^{-1} \right) \\[0.7em]
\stackrel{(a)}{=} &  \: \frac{1}{s_k} \dv\Sigma_{\dv u_k^t|(\dv x, \dv y_0)}^{-1} - \frac{1-s_1}{s_k} \dv\Sigma_{\dv u_k^t|(\dv u_\kbar^t,\dv y_0)}^{-1} + \frac{s_k-s_1}{s_k} \dv\Sigma_{\dv u_k^t|\dv y_0}^{-1} \;,     
\end{align*}
where $(a)$ is due to the Markov chain $\dv U_1 \mkv \dv X \mkv \dv U_2$. We obtain \eqref{equation-CEO-Gauss-Sigma-update} by taking the inverse of both sides of $(a)$.

Also from the matrix inversion lemma \cite{CGTW05}, for $(\dv X_1,\dv X_2)$ jointly Gaussian we have
\begin{equation}~\label{equation-matrix-inv-2}
\dv \Sigma_{\dv x_1}^{-1} \dv \Sigma_{\dv x_1,\dv x_2} \dv\Sigma_{\dv x_2|\dv x_1}^{-1}
= \dv \Sigma_{\dv x_1|\dv x_2}^{-1} \dv \Sigma_{\dv x_1,\dv x_2} \dv\Sigma_{\dv x_2}^{-1} \;.
\end{equation}
Now, we simplify \eqref{equation-CEO-Gauss-mu} to obtain the update rule \eqref{equation-CEO-Gauss-A-update} as follows
\begin{align}
\boldsymbol\mu_{\dv u_k^{t+1}|\dv y_k}
=& \: \dv\Sigma_{\dv z_k^{t+1}} \left( 
\frac{1-s_1}{s_k} \dv\Sigma_{\dv u_k^t}^{-1} \dv\Sigma_{\dv u_k^t,(\dv u_\kbar^t,\dv x,\dv y_0)} \dv\Sigma_{(\dv u_\kbar^t,\dv x,\dv y_0)|\dv u_k^t}^{-1} \dv\Sigma_{(\dv u_\kbar^t,\dv x,\dv y_0),\dv y_k} \dv\Sigma_{\dv y_k}^{-1} \right. \nonumber\\ 
&\left.- \frac{1-s_1}{s_k} \dv\Sigma_{\dv u_k^t}^{-1} \dv\Sigma_{\dv u_k^t,(\dv u_\kbar^t,\dv y_0)} \dv\Sigma_{(\dv u_\kbar^t,\dv y_0)|\dv u_k^t}^{-1} \dv\Sigma_{(\dv u_\kbar^t,\dv y_0),\dv y_k} \dv\Sigma_{\dv y_k}^{-1} \right. \nonumber\\  
&\left.+ \frac{s_1}{s_k} \dv\Sigma_{\dv u_k^t}^{-1} \dv\Sigma_{\dv u_k^t,(\dv x,\dv y_0)} \dv\Sigma_{(\dv x,\dv y_0)|\dv u_k^t}^{-1} \dv\Sigma_{(\dv x,\dv y_0),\dv y_k} \dv\Sigma_{\dv y_k}^{-1} 
+ \frac{s_k-s_1}{s_k} \dv\Sigma_{\dv u_k^t}^{-1} \dv\Sigma_{\dv u_k^t,\dv y_0} \dv\Sigma_{\dv y_0|\dv u_k^t}^{-1} \dv\Sigma_{\dv y_0,\dv y_k} \dv\Sigma_{\dv y_k}^{-1}
\right) \dv y_k \nonumber\\[0.9em]
\stackrel{(a)}{=}& \: \dv\Sigma_{\dv z_k^{t+1}} \left( 
\frac{1-s_1}{s_k} \dv\Sigma_{\dv u_k^t|(\dv u_\kbar^t,\dv x,\dv y_0)}^{-1} \dv\Sigma_{\dv u_k^t,(\dv u_\kbar^t,\dv x,\dv y_0)} \dv\Sigma_{(\dv u_\kbar^t,\dv x,\dv y_0)}^{-1} \dv\Sigma_{(\dv u_\kbar^t,\dv x,\dv y_0),\dv y_k} \dv\Sigma_{\dv y_k}^{-1} \right. \nonumber\\
&\left. - \frac{1-s_1}{s_k} \dv\Sigma_{\dv u_k^t|(\dv u_\kbar^t,\dv y_0)}^{-1} \dv\Sigma_{\dv u_k^t,(\dv u_\kbar^t,\dv y_0)} \dv\Sigma_{(\dv u_\kbar^t,\dv y_0)}^{-1} \dv\Sigma_{(\dv u_\kbar^t,\dv y_0),\dv y_k} \dv\Sigma_{\dv y_k}^{-1}  \right. \nonumber\\
&\left. + \frac{s_1}{s_k} \dv\Sigma_{\dv u_k^t|(\dv x,\dv y_0)}^{-1} \dv\Sigma_{\dv u_k^t,(\dv x,\dv y_0)} \dv\Sigma_{(\dv x,\dv y_0)}^{-1} \dv\Sigma_{(\dv x,\dv y_0),\dv y_k} \dv\Sigma_{\dv y_k}^{-1}  
+ \frac{s_k-s_1}{s_k} \dv\Sigma_{\dv u_k^t|\dv y_0}^{-1} \dv\Sigma_{\dv u_k^t,\dv y_0} \dv\Sigma_{\dv y_0}^{-1} \dv\Sigma_{\dv y_0,\dv y_k} \dv\Sigma_{\dv y_k}^{-1} 
\right) \dv y_k \nonumber\\[0.8em]
\hspace{-2em}\stackrel{(b)}{=}& \: \dv\Sigma_{\dv z_k^{t+1}} \left( 
\frac{1-s_1}{s_k} \dv\Sigma_{\dv u_k^t|(\dv u_\kbar^t,\dv x,\dv y_0)}^{-1} \dv A_k^t \dv\Sigma_{\dv y_k,(\dv u_\kbar^t,\dv x,\dv y_0)} \dv\Sigma_{(\dv u_\kbar^t,\dv x,\dv y_0)}^{-1} \dv\Sigma_{(\dv u_\kbar^t,\dv x,\dv y_0),\dv y_k} \dv\Sigma_{\dv y_k}^{-1} \right. \nonumber\\
&\left. - \frac{1-s_1}{s_k} \dv\Sigma_{\dv u_k^t|(\dv u_\kbar^t,\dv y_0)}^{-1} \dv A_k^t \dv\Sigma_{\dv y_k,(\dv u_\kbar^t,\dv y_0)} \dv\Sigma_{(\dv u_\kbar^t,\dv y_0)}^{-1} \dv\Sigma_{(\dv u_\kbar^t,\dv y_0),\dv y_k} \dv\Sigma_{\dv y_k}^{-1}  \right. \nonumber\\
&\left. + \frac{s_1}{s_k} \dv\Sigma_{\dv u_k^t|(\dv x,\dv y_0)}^{-1} \dv A_k^t \dv\Sigma_{\dv y_k,(\dv x,\dv y_0)} \dv\Sigma_{(\dv x,\dv y_0)}^{-1} \dv\Sigma_{(\dv x,\dv y_0),\dv y_k} \dv\Sigma_{\dv y_k}^{-1} 
+ \frac{s_k-s_1}{s_k} \dv\Sigma_{\dv u_k^t|\dv y_0}^{-1} \dv A_k^t \dv\Sigma_{\dv y_k,\dv y_0} \dv\Sigma_{\dv y_0}^{-1} \dv\Sigma_{\dv y_0,\dv y_k} \dv\Sigma_{\dv y_k}^{-1} 
\right) \dv y_k \nonumber\\[0.8em]
\stackrel{(c)}{=}& \: \dv\Sigma_{\dv z_k^{t+1}} \left( 
\frac{1-s_1}{s_k} \dv\Sigma_{\dv u_k^t|(\dv u_\kbar^t,\dv x,\dv y_0)}^{-1} \dv A_k^t (\dv\Sigma_{\dv y_k} - \dv\Sigma_{\dv y_k|(\dv u_\kbar^t,\dv x,\dv y_0)}) \dv\Sigma_{\dv y_k}^{-1} \right. \nonumber\\
&\left.- \frac{1-s_1}{s_k} \dv\Sigma_{\dv u_k^t|(\dv u_\kbar^t,\dv y_0)}^{-1} \dv A_k^t (\dv\Sigma_{\dv y_k} - \dv\Sigma_{\dv y_k|(\dv u_\kbar^t,\dv y_0)}) \dv\Sigma_{\dv y_k}^{-1} \right. \nonumber\\  
&\left. + \frac{s_1}{s_k} \dv\Sigma_{\dv u_k^t|(\dv x,\dv y_0)}^{-1} \dv A_k^t 
(\dv\Sigma_{\dv y_k} - \dv\Sigma_{\dv y_k|(\dv x,\dv y_0)}) \dv\Sigma_{\dv y_k}^{-1}
+ \frac{s_k-s_1}{s_k} \dv\Sigma_{\dv u_k^t|\dv y_0}^{-1} \dv A_k^t 
(\dv\Sigma_{\dv y_k} - \dv\Sigma_{\dv y_k|\dv y_0}) \dv\Sigma_{\dv y_k}^{-1}
\right) \dv y_k \nonumber\\[1em]
\stackrel{(d)}{=}& \: \dv\Sigma_{\dv z_k^{t+1}} \left(
\frac{1}{s_k}  \dv\Sigma_{\dv u_k^t|(\dv x,\dv y_0)}^{-1} \dv A_k^t (\dv I - \dv\Sigma_{\dv y_k|(\dv x,\dv y_0)} \dv\Sigma_{\dv y_k}^{-1} ) 
- \frac{1-s_1}{s_k} \dv\Sigma_{\dv u_k^t|(\dv u_\kbar^t,\dv y_0)}^{-1} \dv A_k^t (\dv I - \dv\Sigma_{\dv y_k|(\dv u_\kbar^t,\dv y_0)} \dv\Sigma_{\dv y_k}^{-1}) \right. \nonumber\\ 
&\left. + \frac{s_k-s_1}{s_k} \dv\Sigma_{\dv u_k^t|\dv y_0}^{-1} \dv A_k^t 
(\dv I - \dv\Sigma_{\dv y_k|\dv y_0} \dv\Sigma_{\dv y_k}^{-1}) 
\right) \dv y_k \;, \nonumber  
\end{align}
where $(a)$ follows from \eqref{equation-matrix-inv-2}; $(b)$ follows from the relation $\dv\Sigma_{\dv u_k, \dv y_0} = \dv A_k \dv\Sigma_{\dv y_k, \dv y_0}$; $(c)$ is due the definition of $\dv\Sigma_{\dv x_1|\dv x_2}$; and $(d)$ is due to the Markov chain $\dv U_1 \mkv \dv X \mkv \dv U_2$. Equation \eqref{equation-CEO-Gauss-A-update} follows by noting that $\boldsymbol\mu_{\dv u_k^{t+1} |\dv y_k} = \dv A_k^{t+1}\dv y_k$.

\newpage
\bibliographystyle{IEEEtran}
\bibliography{IEEEabrv,mybibfile}

\begin{thebibliography}{10}
\providecommand{\url}[1]{#1}
\csname url@samestyle\endcsname
\providecommand{\newblock}{\relax}
\providecommand{\bibinfo}[2]{#2}
\providecommand{\BIBentrySTDinterwordspacing}{\spaceskip=0pt\relax}
\providecommand{\BIBentryALTinterwordstretchfactor}{4}
\providecommand{\BIBentryALTinterwordspacing}{\spaceskip=\fontdimen2\font plus
\BIBentryALTinterwordstretchfactor\fontdimen3\font minus
  \fontdimen4\font\relax}
\providecommand{\BIBforeignlanguage}[2]{{%
\expandafter\ifx\csname l@#1\endcsname\relax
\typeout{** WARNING: IEEEtran.bst: No hyphenation pattern has been}%
\typeout{** loaded for the language `#1'. Using the pattern for}%
\typeout{** the default language instead.}%
\else
\language=\csname l@#1\endcsname
\fi
#2}}
\providecommand{\BIBdecl}{\relax}
\BIBdecl

\bibitem{UEZ17}
Y.~Ugur, I.~E. Aguerri, and A.~Zaidi, ``A generalization of {B}lahut-{A}rimoto
  algorithm to compute rate-distortion regions of multiterminal source coding
  under logarithmic loss,'' in \emph{Proc. of IEEE Inf. Theory Workshop}, Nov.
  2017, pp. 349--353.

\bibitem{UEZ18}
------, ``Vector {G}aussian {CEO} problem under logarithmic loss,'' in
  \emph{Proc. of IEEE Inf. Theory Workshop}, Nov. 2018, pp. 515--519.

\bibitem{O05}
Y.~Oohama, ``Rate-distortion theory for {G}aussian multiterminal source coding
  systems with several side informations at the decoder,'' \emph{IEEE Trans.
  Inf. Theory}, vol.~51, no.~7, pp. 2577--2593, Jul. 2005.

\bibitem{PTR04}
V.~Prabhakaran, D.~Tse, and K.~Ramachandran, ``Rate region of the quadratic
  {G}aussian {CEO} problem,'' in \emph{Proc. of IEEE Int. Symp. Inf. Theory},
  Jun.-Jul. 2004, p. 117.

\bibitem{CW11}
J.~Chen and J.~Wang, ``On the vector {G}aussian {CEO} problem,'' in \emph{Proc.
  of IEEE Int. Symp. Inf. Theory}, Jul.-Aug. 2011, pp. 2050--2054.

\bibitem{WC12}
J.~Wang and J.~Chen, ``On the vector {G}aussian ${L}$-terminal {CEO} problem,''
  in \emph{Proc. of IEEE Int. Symp. Inf. Theory}, Jul. 2012, pp. 571--575.

\bibitem{LV07}
T.~Liu and P.~Viswanath, ``An extremal inequality motivated by multiterminal
  information-theoretic problems,'' \emph{IEEE Trans. Inf. Theory}, vol.~53,
  no.~5, pp. 1839--1851, May 2007.

\bibitem{XW16}
Y.~Xu and Q.~Wang, ``Rate region of the vector {G}aussian {CEO} problem with
  the trace distortion constraint,'' \emph{IEEE Trans. Inf. Theory}, vol.~62,
  no.~4, pp. 1823--1835, Apr. 2016.

\bibitem{CW11-2}
T.~A. Courtade and R.~D. Wesel, ``Multiterminal source coding with an
  entropy-based distortion measure,'' in \emph{Proc. of IEEE Int. Symp. Inf.
  Theory}, Jul.-Aug. 2011, pp. 2040--2044.

\bibitem{CW14}
T.~A. Courtade and T.~Weissman, ``Multiterminal source coding under logarithmic
  loss,'' \emph{IEEE Trans. Inf. Theory}, vol.~60, no.~1, pp. 740--761, Jan.
  2014.

\bibitem{JCVW15}
J.~Jiao, T.~A. Courtade, K.~Venkat, and T.~Weissman, ``Justification of
  logarithmic loss via the benefit of side information,'' \emph{IEEE Trans.
  Inf. Theory}, vol.~61, no.~10, pp. 5357--5365, Oct. 2015.

\bibitem{NW15}
A.~No and T.~Weissman, ``Universality of logarithmic loss in lossy
  compression,'' in \emph{Proc. of IEEE Int. Symp. Inf. Theory}, Jun. 2015, pp.
  2166--2170.

\bibitem{SRV17}
Y.~Shkel, M.~Raginsky, and S.~Verdu, ``Universal lossy compression under
  logarithmic loss,'' in \emph{Proc. of IEEE Int. Symp. Inf. Theory}, Jun.
  2017, pp. 1157--1161.

\bibitem{TPB99}
N.~Tishby, F.~C. Pereira, and W.~Bialek, ``The information bottleneck method,''
  in \emph{Proc. of the 37th Annu. Allerton Conf. Commun., Control and
  Comput.}, 1999, pp. 368--377.

\bibitem{C-BL06}
N.~Cesa-Bianchi and G.~Lugosi, \emph{Prediction, Learning and Games}.\hskip 1em
  plus 0.5em minus 0.4em\relax New York, USA: Cambridge Univ. Press, 2006.

\bibitem{AABG06}
T.~Andre, M.~Antonini, M.~Barlaud, and R.~M. Gray, ``Entropy-based distortion
  measure for image coding,'' in \emph{Proc. of IEEE Int. Conf. Image
  Process.}, Oct. 2006, pp. 1157--1160.

\bibitem{KCOSW16}
K.~Kittichokechai, Y.-K. Chia, T.~J. Oechtering, M.~Skoglund, and T.~Weissman,
  ``Secure source coding with a public helper,'' \emph{IEEE Trans. Inf.
  Theory}, vol.~62, no.~7, pp. 3930--3949, Jul. 2016.

\bibitem{EU14}
E.~Ekrem and S.~Ulukus, ``An outer bound for the vector {G}aussian {CEO}
  problem,'' \emph{IEEE Trans. Inf. Theory}, vol.~60, no.~11, pp. 6870--6887,
  Nov. 2014.

\bibitem{TP05}
S.~Tavildar and P.~Viswanath, ``On the sum-rate of the vector {G}aussian {CEO}
  problem,'' in \emph{Proc. of 39th Asilomar Conf. Signals, Syst. Comput.},
  Oct.-Nov. 2005, pp. 3--7.

\bibitem{WSS06}
H.~Weingarten, Y.~Steinberg, and S.~{Shamai (Shitz)}, ``The capacity region of
  the {G}aussian multiple-input multiple-output broadcast channel,'' \emph{IEEE
  Trans. Inf. Theory}, vol.~52, no.~9, pp. 3936--3964, Sep. 2006.

\bibitem{PCL03}
D.~P. Palomar, J.~M. Cioffi, and M.~A. Lagunas, ``Joint {T}x-{R}x beamforming
  design for multicarrier {MIMO} channels: A unified framework for convex
  optimization,'' \emph{IEEE Trans. Signal Process.}, vol.~51, no.~9, pp.
  2381--2401, Sep. 2003.

\bibitem{SSBGS02}
A.~Scaglione, P.~Stoica, S.~Barbarossa, G.~B. Giannakis, and H.~Sampath,
  ``Optimal designs for space-time linear precoders and decoders,'' \emph{IEEE
  Trans. Signal Process.}, vol.~50, no.~5, pp. 1051--1064, May 2002.

\bibitem{HT07}
P.~Harremoes and N.~Tishby, ``The information bottleneck revisited or how to
  choose a good distortion measure,'' in \emph{Proc. of IEEE Int. Symp. Inf.
  Theory}, Jun. 2007, pp. 566--570.

\bibitem{B72}
R.~E. Blahut, ``Computation of channel capacity and rate-distortion
  functions,'' \emph{IEEE Trans. Inf. Theory}, vol. IT-18, no.~4, pp. 460--473,
  Jul. 1972.

\bibitem{A72}
S.~Arimoto, ``An algorithm for computing the capacity of arbitrary discrete
  memoryless channels,'' \emph{IEEE Trans. Inf. Theory}, vol. IT-18, no.~1, pp.
  14--20, Jan. 1972.

\bibitem{CGTW05}
G.~Chechik, A.~Globerson, N.~Tishby, and Y.~Weiss, ``Information bottleneck for
  {G}aussian variables,'' \emph{J. of Mach. Learn. Research}, vol.~6, pp.
  165--188, Jan. 2005.

\bibitem{WM14}
A.~Winkelbauer and G.~Matz, ``Rate-information-optimal {Gaussian} channel
  output compression,'' in \emph{Proc. of the 48th Annu. Conf. Inf. Sciences
  and Sys.}, Aug. 2014.

\bibitem{WFM14}
A.~Winkelbauer, S.~Farthofer, and G.~Matz, ``The rate-information trade-off for
  {G}aussian vector channels,'' in \emph{Proc. of IEEE Int. Symp. Inf. Theory},
  Jun. 2014, pp. 2849--2853.

\bibitem{CSX05}
S.~Cheng, V.~Stankovic, and Z.~Xiong, ``Computing the channel capacity and
  rate-distortion function with two-sided state information,'' \emph{IEEE
  Trans. Inf. Theory}, vol.~51, no.~12, pp. 4418--4425, Dec. 2005.

\bibitem{CB04}
M.~Chiang and S.~Boyd, ``Geometric programming duals of channel capacity and
  rate distortion,'' \emph{IEEE Trans. Inf. Theory}, vol.~50, no.~2, pp.
  245--258, Feb. 2004.

\bibitem{DYW04}
F.~Dupuis, W.~Yu, and F.~M.~J. Willems, ``Blahut-{A}rimoto algorithms for
  computing channel capacity and rate-distortion with side information,'' in
  \emph{Proc. of IEEE Int. Symp. Inf. Theory}, Jun.-Jul. 2004, p. 181.

\bibitem{RG04}
M.~Rezaeian and A.~Grant, ``A generalization of {A}rimoto-{B}lahut algorithm,''
  in \emph{Proc. of IEEE Int. Symp. Inf. Theory}, Jun.-Jul. 2004, p. 180.

\bibitem{PW18}
A.~Painsky and G.~Wornell, ``On the universality of the logistic loss
  function,'' in \emph{Proc. of IEEE Int. Symp. Inf. Theory}, Jun. 2018, pp.
  936--940.

\bibitem{LWOG18}
C.~T. Li, X.~Wu, A.~Ozgur, and A.~E. Gamal, ``Minimax learning for remote
  prediction,'' in \emph{Proc. of IEEE Int. Symp. Inf. Theory}, Jun. 2018, pp.
  541--545.

\bibitem{TC09}
C.~Tian and J.~Chen, ``Remote vector {G}aussian source coding with decoder side
  information under mutual information and distortion constraints,'' \emph{IEEE
  Trans. Inf. Theory}, vol.~55, no.~10, pp. 4676--4680, Oct. 2009.

\bibitem{SSSK08}
A.~Sanderovich, S.~{Shamai (Shitz)}, Y.~Steinberg, and G.~Kramer,
  ``Communication via decentralized processing,'' \emph{IEEE Trans. Inf.
  Theory}, vol.~54, no.~7, pp. 3008--3023, Jul. 2008.

\bibitem{SES11}
O.~Simeone, E.~Erkip, and S.~{Shamai (Shitz)}, ``On codebook information for
  interference relay channels with out-of-band relaying,'' \emph{IEEE Trans.
  Inf. Theory}, vol.~57, no.~5, pp. 2880--2888, May 2011.

\bibitem{EZSC17-2}
I.~E. Aguerri, A.~Zaidi, G.~Caire, and S.~{Shamai (Shitz)}, ``On the capacity
  of cloud radio access networks with oblivious relaying,'' in \emph{Proc. of
  IEEE Int. Symp. Inf. Theory}, Jun. 2017, pp. 2068--2072.

\bibitem{EZSC17}
------, ``On the capacity of cloud radio access networks with oblivious
  relaying,'' \emph{IEEE Trans. Inf. Theory}, vol.~65, no.~7, pp. 4575--4596,
  July 2019.

\bibitem{CMMVCD17}
F.~P. Calmon, A.~Makhdoumi, M.~Medard, M.~Varia, M.~Christiansen, and K.-D.
  Duffy, ``Principal inertia components and applications,'' \emph{IEEE Trans.
  Inf. Theory}, vol.~63, no.~8, pp. 5011--5038, Jul. 2017.

\bibitem{AC86}
R.~Ahlswede and I.~Csiszar, ``Hypothesis testing with communication
  constraints,'' \emph{IEEE Trans. Inf. Theory}, vol. IT-32, no.~4, pp.
  533--542, Jul. 1986.

\bibitem{H87}
T.~Han, ``Hypothesis testing with multiterminal data compression,'' \emph{IEEE
  Trans. Inf. Theory}, vol. IT-33, no.~6, pp. 759--772, Nov. 1987.

\bibitem{RW12}
M.~S. Rahman and A.~B. Wagner, ``On the optimality of binning for distributed
  hypothesis testing,'' \emph{IEEE Trans. Inf. Theory}, vol.~58, no.~10, pp.
  6282--6303, Oct. 2012.

\bibitem{TC08}
C.~Tian and J.~Chen, ``Successive refinement for hypothesis testing and
  lossless one-helper problem,'' \emph{IEEE Trans. Inf. Theory}, vol.~54,
  no.~10, pp. 4666--4681, Oct. 2008.

\bibitem{SWT18}
S.~Salehkalaibar, M.~Wigger, and R.~Timo, ``On hypothesis testing against
  conditional independence with multiple decision centers,'' \emph{IEEE Trans.
  Commun.}, vol.~66, no.~6, pp. 2409--2420, Jun. 2018.

\bibitem{ZXYC16}
Y.~Zhou, Y.~Xu, W.~Yu, and J.~Chen, ``On the optimal fronthaul compression and
  decoding strategies for uplink cloud radio access networks,'' \emph{IEEE
  Trans. Inf. Theory}, vol.~62, no.~12, pp. 7402--7418, Dec. 2016.

\bibitem{C15}
T.~A. Courtade, ``Gaussian multiterminal source coding through the lens of
  logarithmic loss,'' in \emph{Inf. Theory and Appl. Workshop}, 2015.

\bibitem{C18}
------, ``A strong entropy power inequality,'' \emph{IEEE Trans. Inf. Theory},
  vol.~64, no.~4, pp. 2173--2192, Apr. 2018.

\bibitem{WTV08}
A.~B. Wagner, S.~Tavildar, and P.~Viswanath, ``Rate region of the quadratic
  {G}aussian two-encoder source-coding problem,'' \emph{IEEE Trans. Inf.
  Theory}, vol.~54, no.~5, pp. 1938--1961, May 2008.

\bibitem{CJ14}
T.~A. Courtade and J.~Jiao, ``An extremal inequality for long {M}arkov
  chains,'' in \emph{Proc. of the 52nd Annu. Allerton Conf. Commun., Control
  and Comput.}, 2014.

\bibitem{ST01}
N.~Slonim and N.~Tishby, ``The power of word clusters for text
  classification,'' in \emph{Proc. of 23rd European Colloq. Inf. Retrieval
  Research}, 2001, pp. 191--200.

\bibitem{BE-YL04}
Y.~Baram, R.~El-Yaniv, and K.~Luz, ``Online choice of active learning
  algorithms,'' \emph{J. of Mach. Learn. Research}, vol.~5, pp. 255--291, Mar.
  2004.

\bibitem{EZ18}
I.~E. Aguerri and A.~Zaidi, ``Distributed information bottleneck method for
  discrete and {G}aussian sources,'' in \emph{Proc. of IEEE Int. Zurich Seminar
  Inf. and Commun.}, Feb. 2018.

\bibitem{EZ20}
------, ``Distributed variational representation learning,'' \emph{IEEE Trans.
  Pattern Anal. Mach. Intell.}, 2020.

\bibitem{RZ15}
D.~Russo and J.~Zou, ``How much does your data exploration overfit?
  {C}ontrolling bias via information usage,'' \emph{arXiv: 1511.05219}, 2015.

\bibitem{XR17}
A.~Xu and M.~Raginsky, ``Information-theoretic analysis of generalization
  capability of learning algorithms,'' in \emph{Proc. of Conf. Neural Inf.
  Process. Sys.}, 2017, pp. 2524--2533.

\bibitem{AAV18}
A.~R. Asadi, E.~Abbe, and S.~Verdu, ``Chaining mutual information and
  tightening generalization bounds,'' in \emph{Proc. of the 32nd Conf. Neural
  Inf. Process. Syst.}, 2018.

\bibitem{CB08}
J.~Chen and T.~Berger, ``Successive {W}yner-{Z}iv coding scheme and its
  application to the quadratic {G}aussian {CEO} problem,'' \emph{IEEE Trans.
  Inf. Theory}, vol.~54, no.~4, pp. 1586--1603, Apr. 2008.

\bibitem{RHL13}
M.~Razaviyayn, M.~Hong, and Z.-Q. Luo, ``A unified convergence analysis of
  block successive minimization methods for nonsmooth optimization,''
  \emph{SIAM J. Optim.}, vol.~23, no.~2, pp. 1126--1153, Jun. 2013.

\bibitem{cvx}
M.~Grant and S.~Boyd, ``{CVX}: Matlab software for disciplined convex
  programming,'' \url{http://cvxr.com/cvx}, Mar. 2014.

\bibitem{CJW14}
T.~A. Courtade, J.~Jiao, and T.~Weissman, ``On an extremal data processing
  inequality for long {M}arkov chains,'' in \emph{Proc. of IEEE Int. Zurich
  Seminar Inf. and Commun.}, Feb. 2014, pp. 33--36.

\bibitem{C12}
T.~A. Courtade, ``Information masking and amplification: The source coding
  setting,'' in \emph{Proc. of IEEE Int. Symp. Inf. Theory}, 2012, pp.
  189--193.

\bibitem{DCT91}
A.~Dembo, T.~M. Cover, and J.~A. Thomas, ``Information theoretic
  inequalities,'' \emph{IEEE Trans. Inf. Theory}, vol.~37, no.~6, pp.
  1501--1518, Nov. 1991.

\bibitem{PV06}
D.~P. Palomar and S.~Verdu, ``Gradient of mutual information in linear vector
  gaussian channels,'' \emph{IEEE Trans. Inf. Theory}, vol.~52, no.~1, pp.
  141--154, Jan. 2006.

\end{thebibliography}
\end{document}